\documentclass[runningheads,a4paper,orivec]{llncs}

\newcommand*{\INCLUDEAPPENDIX}{}%
\newcommand*{\BUILDARXIV}{}%

\ifdefined\INCLUDEAPPENDIX
\usepackage{titletoc}
\newcommand{\appref}{appendix}
\else
\newcommand{\appref}{appendix~\cite{appendix}}
\fi

\ifdefined\BUILDARXIV
\else
\fi

\usepackage[numbers]{natbib}

\usepackage{amssymb}  \usepackage{graphicx}

\makeatletter%
\@ifundefined{basedir}{%
\newcommand\basedir{}%
}{}%
\makeatother%
\usepackage{amsfonts}
\usepackage{amssymb}
\usepackage{stmaryrd}
\usepackage{multirow}
\usepackage{tabularx}

\usepackage{\basedir mathpartir}

\usepackage{\basedir pftools}
\usepackage{\basedir iris}

\usepackage{xcolor}  %

\usepackage{graphicx}
\usepackage{enumitem}
\usepackage{semantic}
\usepackage{csquotes}
\usepackage{letltxmacro}

\usepackage{hyperref}

\SetSymbolFont{stmry}{bold}{U}{stmry}{m}{n} %

\extrarowheight=\jot	%
\newcolumntype{.}{@{}}
\newcolumntype{M}{@{\mskip\thickmuskip}}

\definecolor{StringRed}{rgb}{.637,0.082,0.082}
\definecolor{CommentGreen}{rgb}{0.0,0.55,0.3}
\definecolor{KeywordBlue}{rgb}{0.0,0.3,0.55}
\definecolor{LinkColor}{rgb}{0.55,0.0,0.3}
\definecolor{CiteColor}{rgb}{0.55,0.0,0.3}
\definecolor{HighlightColor}{rgb}{0.0,0.0,0.0}

\definecolor{grey}{rgb}{0.5,0.5,0.5}
\definecolor{red}{rgb}{1,0,0}

\hypersetup{%
  linktocpage=true, pdfstartview=FitV,
  breaklinks=true, pageanchor=true, pdfpagemode=UseOutlines,
  plainpages=false, bookmarksnumbered, bookmarksopen=true, bookmarksopenlevel=3,
  hypertexnames=true, pdfhighlight=/O,
  colorlinks=true,linkcolor=LinkColor,citecolor=CiteColor,
  urlcolor=LinkColor
}

\spnewtheorem{defn}[theorem]{Definition}{\bfseries}{\itshape}
\spnewtheorem{cor}[theorem]{Corollary}{\bfseries}{\itshape}
\spnewtheorem{conj}[theorem]{Conj}{\bfseries}{\itshape}
\spnewtheorem{lem}[theorem]{Lemma}{\bfseries}{\itshape}
\spnewtheorem{thm}[theorem]{Theorem}{\bfseries}{\itshape}

\newcommand*{\Sref}[1]{\hyperref[#1]{\S\ref*{#1}}}
\let\secref\Sref
\newcommand*{\lemref}[1]{\hyperref[#1]{Lemma~\ref*{#1}}}
\newcommand*{\thmref}[1]{\hyperref[#1]{Theorem~\ref*{#1}}}
\newcommand{\corref}[1]{\hyperref[#1]{Cor.~\ref*{#1}}}
\newcommand*{\defref}[1]{\hyperref[#1]{Definition~\ref*{#1}}}
\newcommand*{\egref}[1]{\hyperref[#1]{Example~\ref*{#1}}}
\newcommand*{\appendixref}[1]{\hyperref[#1]{Appendix~\ref*{#1}}}
\newcommand*{\figref}[1]{\hyperref[#1]{Figure~\ref*{#1}}}
\newcommand*{\tabref}[1]{\hyperref[#1]{Table~\ref*{#1}}}

\newcommand{\ie}{\emph{i.e.,} }
\newcommand{\eg}{\emph{e.g.,} }
\newcommand{\etal}{\emph{et~al.}}

\renewcommand*{\mathellipsis}{%
  \mathinner{{\ldotp}{\ldotp}{\ldotp}}%
}
\makeatletter
\@ifdefinable{\org@ldots}{%
  \LetLtxMacro\org@ldots\ldots
  \DeclareRobustCommand*{\ldots}{%
    \ifmmode
      \expandafter\my@ldots
    \else
      \expandafter\textellipsis
    \fi
  }%
}
\newcommand*{\neghalfmskip}{%
  \nonscript\mskip-.5\muexpr\thinmuskip\relax%
}
\newcommand*{\my@ldots}{%
  \mathellipsis
  \@ifnextchar,\neghalfmskip{%
  \@ifnextchar:\neghalfmskip{%
  \@ifnextchar;\neghalfmskip{%
  \@ifnextchar.\neghalfmskip{%
  \@ifnextchar!\neghalfmskip{%
  \@ifnextchar?\neghalfmskip{%
    \rightdelim@
    \ifgtest@
      \mskip-.5\muexpr\thinmuskip\relax%
    \fi
  }}}}}}%
}
\makeatother

\newcommand{\heaploc}{l}

\newcommand{\chanloc}{c}
\newcommand{\side}{s}
\newcommand{\vunit}{()}

\newcommand{\buflr}{b_{\rightarrow}}
\newcommand{\bufrl}{b_{\leftarrow}}
\newcommand{\locside}[2]{{#1}_{#2}}
\newcommand{\locsidevar}{\locside{\chanloc}{\side}}
\newcommand{\lside}{\textlog{left}}
\newcommand{\rside}{\textlog{right}}
\newcommand{\fork}[1]{\textlog{fork}\{#1\}}
\newcommand{\newch}{\textlog{newch}}
\newcommand{\recv}[1]{\textlog{recv}(#1)}
\newcommand{\send}[2]{\textlog{send}(#1, #2)}
\newcommand{\letp}[4]{\textlog{let}\;(#1, #2)\;\textlog{=}\;#3\;\textlog{in}\;#4}
\newcommand{\bind}[3]{\textlog{let}\;#1\;\textlog{=}\;#2\;\textlog{in}\;#3}

\newcommand{\typ}{\tau}

\newcommand{\styp}{S}
\newcommand{\dual}[1]{\overline{#1}}
\newcommand{\lolli}{\multimap}
\newcommand{\otensor}{\otimes}
\newcommand{\unit}{\textlog{Unit}}
\newcommand{\inttyp}{\textlog{Int}}
\newcommand{\booltyp}{\textlog{Bool}}

\newcommand{\sendtyp}[2]{{! #1.\, #2}}
\newcommand{\recvtyp}[2]{{? #1.\, #2}}
\newcommand{\stypend}{\textlog{end}}

\newcommand{\heapsend}{\textlog{heapSend}}
\newcommand{\heaprecv}{\textlog{heapRecv}}
\newcommand{\heapalloc}{\textlog{heapNewch}}
\let\heapnewch\heapalloc

\newcommand{\rec}[2]{\textlog{rec}\,#1\,#2.\,}
\newcommand{\recN}[1]{\textlog{rec}\ #1}
\newcommand{\store}[2]{#1\;\textlog{:=}\;#2}
\newcommand{\match}[1]{\textlog{match}\;#1\;\textlog{with}}
\newcommand{\ifmatch}[1]{\textlog{if}\;#1\;\textlog{then}}
\newcommand{\matchend}{\textlog{end}}
\newcommand{\load}[1]{! #1}
\newcommand{\alloc}[1]{\textlog{ref}\;#1}
\newcommand{\swap}[2]{\textlog{swap}(#1, #2)}
\newcommand{\fai}[1]{\textlog{FAI}(#1)}
\newcommand{\some}[1]{\textlog{some}\; #1}
\newcommand{\none}{\textlog{none}}
\newcommand{\patcase}[2]{#1 \Rightarrow #2}
\newcommand{\fst}[1]{\textlog{fst}\,#1}
\newcommand{\snd}[1]{\textlog{snd}\,#1}
\newcommand{\inl}{\textlog{inl}}
\newcommand{\inr}{\textlog{inr}}
\newcommand{\caseelim}[3]{\textlog{case}(#1, #2, #3)}

\newcommand{\locktyp}{\textlog{Lock}}
\newcommand{\reftyp}[1]{\textlog{Ref}\,#1}
\newcommand{\typtrans}{\leadsto}

\let\mapsto\hookrightarrow

\newcommand\leftTok[1]{[\textlog{Left}\ #1]} %
\newcommand\rightTok[1]{[\textlog{Right}\ #1]}

\newcommand{\refines}{\sqsubseteq} %
\newcommand{\valobsrel}{\approx} %

\newcommand\Expr{E}
\newcommand\State{\Sigma}
\newcommand\Val{V}
\newcommand\Cfg{\rho}

\newcommand\compile[1]{\widehat{#1}}

\newcommand\smapsto{\hookrightarrow_{\textlog{s}}}

\newcommand{\ownThread}[2]{\textlog{source}(#1, #2)}
\newcommand{\ownThreadD}[3]{\textlog{source}(#1, #2, #3)}
\newcommand{\ownThreadNoArg}{\textlog{source}}
\newcommand{\delay}{d}

\newcommand{\aff}[1]{\mathcal{A}(#1)}
\newcommand{\emp}{\textlog{Emp}}
\newcommand{\stopped}{\textlog{Stopped}}

\newcommand{\sessionInv}[3]{\pred_{#1, #2}(#3)}
\newcommand{\sessionInvPre}[4]{\pred_{#1, #2, #3}(#4)}
\newcommand{\sessionInvNoArg}{\pred}

\newcommand{\typInterp}{\Theta}
\newcommand{\scl}{\textlog{l}}
\newcommand{\scr}{\textlog{r}}
\newcommand{\sch}{\textlog{h}}
\newcommand{\scc}{\textlog{c}}
\newcommand{\linklist}[3]{\textlog{linklist}(#1, #2, #3)}
\newcommand{\llength}[1]{|#1|}
\newcommand{\sadvance}[2]{{#1}^{#2}}
\newcommand{\interpRel}[5]{{#1} \mathrel{\simeq^{\mathcal{#3}}_{#4}} {#2} : {#5}}
\newcommand{\interpRelNoArgs}[2]{\simeq^{\mathcal{#1}}_{#2}}
\newcommand{\interpValRel}[4]{\interpRel{#1}{#2}{V}{#3}{#4}}
\newcommand{\interpExprRel}[4]{\interpRel{#1}{#2}{E}{#3}{#4}}
\newcommand{\interpListRel}[4]{\interpRel{#1}{#2}{L}{#3}{#4}}
\newcommand{\interpOpenRel}[4]{#1 \vdash \interpRel{#2}{#3}{}{}{#4}}

\newcommand{\SessionProt}[4]{\textlog{Session}_{#1}(#2,#3,#4)} %

\newcommand{\dmax}{\ensuremath{D}}

\newcommand{\thefp}{\ensuremath{\mathrel{\Omega}}}
\newcommand{\islock}[4]{\textlog{isLock}(#1, #2, #3, #4)}
\newcommand{\locked}[3]{\textlog{locked}(#1, #2, #3)}

\newcommand{\substh}{\gamma_h}
\newcommand{\substc}{\gamma_c}
\newcommand{\substhB}{\chi_h}
\newcommand{\substcB}{\chi_c}

\newcommand{\bigast}{\scalebox{3}{\raisebox{-0.3ex}{$\ast$}}}

\newcommand{\fv}[1]{\textlog{fv}(#1)}

\makeatletter %
\def\arcr{\@arraycr}
\makeatother

\newcommand{\isnew}[1]{{\color{blue} #1}}
 
\usepackage{url}

\begin{document}

\mainmatter  %

\title{A Higher-Order Logic for Concurrent Termination-Preserving Refinement}

\titlerunning{A Higher-Order Logic for Concurrent Termination-Preserving Refinement}
\author{Joseph Tassarotti\inst{1}\and Ralf Jung\inst{2}\and Robert
  Harper\inst{1}}

\institute{Carnegie Mellon University, Pittsburgh,
  USA\and%
MPI-SWS, Saarland, Germany}

\maketitle

\begin{abstract}

Compiler correctness proofs for higher-order concurrent languages are
difficult: they involve establishing a termination-preserving
refinement between a \emph{concurrent} high-level source language and
an implementation that uses low-level shared memory primitives.  
However, existing logics for proving concurrent refinement
either neglect properties such as termination, or only handle
first-order state.  In this paper, we address these limitations by
extending Iris, a recent higher-order concurrent separation logic,
with support for reasoning about termination-preserving refinements.
To demonstrate the power of these extensions, we prove the correctness of an efficient
implementation of a higher-order, session-typed language. To our
knowledge, this is the first program logic capable of giving a
compiler correctness proof for such a language.  The soundness of our
extensions and our compiler correctness proof have been mechanized in
Coq.

\end{abstract}

\section{Introduction}

Parallelism and concurrency impose great challenges on both programmers and
compilers.  In order to make compiled code more efficient and help programmers avoid
errors, languages can
provide type systems or other features to constrain the structure of
programs and provide useful guarantees. The design of these kinds of concurrent
languages is an active area of research.  However, it is frequently
difficult to prove that efficient compilers for these languages are correct, and
that important properties of the source-level language are preserved under
compilation.

For example, in work on session types~\cite{honda93,yoshida07,GayV10,CairesP10,Wadler14}, processes communicate by sending
messages over channels. These channels are given a type which
describes the kind of data sent over the channel, as well as the order
in which each process sends and receives messages. Often, the type
system in these languages ensures the absence of undesired behaviors like races and deadlocks;
for instance, two threads cannot both be trying to send a message on
the same channel simultaneously.

Besides preventing errors, the invariants enforced by session types also permit these
language to be compiled efficiently to a shared-memory target
language~\cite{Willsey16}. For example, because
only one thread can be sending a message on a given channel at a
time, channels can be implemented without performing
locking to send and receive messages.
It is particularly important to
prove that such an implementation does not \emph{introduce} races or deadlocks, since
this would destroy  the very properties that make certain
session-typed languages so interesting.

In this paper, we develop a higher-order program logic for proving the correctness of such concurrent language implementations, in a way that ensures that termination is preserved. 
We have used this program logic to give
a machine-checked proof of correctness for a lock-free implementation of a
\emph{higher-order} session-typed language, \ie a language in which closures and channels can be sent over channels.
To our knowledge, this is the \emph{first such proof} of its kind.

As we describe below, previously developed program logics
cannot be used to obtain these kinds of correctness results due to various limitations.
In the remainder of the introduction, we will explain why it is so hard to prove refinements between higher-order, concurrent languages. To this end, we first have to provide some background.

\paragraph{Refinement for concurrent languages.}
To show that a compiler is correct, one typically proves that if a source expression $\Expr$ is
well-typed, its translation $\compile\Expr$ \emph{refines} $\Expr$. In the
sequential setting, this notion of refinement is easy to
define\footnote{Setting aside issues of IO behavior.}: (1) if the target program $\compile\Expr$ terminates in some value $v$, we expect $\Expr$ to
also have an execution that terminates with value $v$, and (2) if $\compile\Expr$ diverges, then $\Expr$
should also have a diverging execution.

In the concurrent setting, however, we need to change this definition.
In particular, the condition (2) concerning diverging executions is too weak. To see why, consider
the following program, where \texttt{x} initially contains $0$:
\[ \texttt{while (*x == 0)~ \{\}} \qquad 
   || \qquad \texttt{*x = 1;} \]
Here, $||$ represents parallel composition of two threads. In every execution where the thread on the right
 eventually gets to run, this program will terminate. However, the program does
 have a diverging execution in which only the left thread runs: because \texttt{x} remains $0$, the left thread continues to loop. Such executions
 are ``unrealistic'' in the sense that generally, we rely on schedulers to be
 \emph{fair} and not let a thread starve.  As a consequence, for purposes of compiler correctness,
 we do not want to consider these ``unrealistic'' executions which only diverge because the
 scheduler never lets a thread run.

 Formally, an infinite execution is said to be \emph{fair}~\cite{LehmannPS81} if every thread which
 does not terminate in a value takes infinitely many steps.\footnote{This
   definition is simpler than the version found in \citet{LehmannPS81}, because there
   threads can be temporarily \emph{disabled}, \ie blocked and unable to take a step.
   In the languages we consider, threads can always take a step unless
   they have finished executing or have ``gone wrong''.} In the definition
 of refinement above, we change (2) to demand that if $\compile\Expr$ has a \emph{fair} diverging
 execution, then $\Expr$ also has a \emph{fair} diverging execution.  We impose no
 such requirement about unfair diverging executions.
This leads us to \emph{fair termination-preserving refinement}.

\paragraph{Logics for proving refinement.}
To prove our compiler correct, we need to reason about the concurrent execution and (non)termination of the source and target programs.
Rather than reason directly about all possible executions of these programs, we prefer to use a concurrent program logic in order to re-use ideas found in rely-guarantee reasoning~\cite{rg} and concurrent separation logic~\cite{ohearn:csl}.
However, although a number of concurrency logics have recently been developed for reasoning about termination and refinements, they cannot be used to prove our compiler correctness result because they either:

\begin{itemize}
\item are restricted to first-order state~\cite{hoffmann:lock-freedom,total-tada,liang:refinement14,liang:termination14,liang:lili},
\item only deal with termination, not refinement~\cite{hoffmann:lock-freedom,total-tada}, or
\item handle a weaker form of refinement that is not fair termination-preserving~\cite{caresl,liang:refinement14,liang:termination14}.
\end{itemize}

Although the limitations are different in each of the above papers, let us focus on the
approach by \citet{caresl} since we will build on it.
That paper establishes a termination-\emph{insensitive} form of refinement, \ie a diverging program refines every program.
Refinement is proven in a
higher-order concurrent separation logic which, in addition to the usual
points-to assertions $l \gmapsto v$, also provides assertions
\emph{about the source language's state}.  For instance, the
assertion\footnote{The notation in \citet{caresl} is different.}
$\ownThread{i}{\Expr}$
  says thread $i$ in the source language's execution is running
expression $\Expr$. A thread which ``owns'' this resource is allowed to
modify the state of the source program by simulating steps of the execution of $\Expr$. Then, we can
prove that $\expr$ refines $\Expr$ by showing:
\[
{\hoare{\ownThread{i}{\Expr}}{\expr}{\Ret\val. \ownThread{i}{\val}}[]}
\]
As usual, the triple enforces that the post-condition holds on
termination of $\expr$.  Concretely for the triple above, the soundness
theorem for the logic implies that if target expression $\expr$
terminates with a value $\val$, then there is an execution of source
expression $\Expr$ that also terminates with value $\val$. However,
the Hoare triple above only expresses \emph{partial
  correctness}.  That means if $\expr$ does not terminate, then the
triple above is trivial, and so these triples can only be used to
prove termination-insensitive refinements.

Ideally, one would like to overcome this limitation by adapting ideas
from logics that deal with termination for first-order state. Notably,
Liang~\etal~\cite{liang:lili} have recently developed a logic for establishing
\emph{fair} refinements (as defined above). 

However, there is a serious difficulty in trying to adapt these ideas.
Semantic models of concurrency logics for
higher-order state usually involve \emph{step-indexing}~\cite{appel-mcallester,birkedal:metric-space}.
In step-indexed logics, the validity of Hoare triples is restricted to program executions of arbitrary \emph{but finite} length.
How can we use these to reason about fairness, a property which is inherently about \emph{infinite} executions?

In this paper, we show how to overcome this difficulty: the key insight is that when the source language has only \emph{bounded non-determinism}, step-indexed Hoare triples are actually sufficient to establish properties of infinite program executions.
Using this observation, we extend Iris~\cite{iris,iris2}, a recent higher-order concurrent separation logic, to support reasoning about fair termination-preserving refinement.
The soundness of our extensions to Iris and our case studies have been verified in Coq.

\paragraph{Overview.}

We start by introducing the case study that we will focus on in this paper: a session-typed source language, a
compiler into an ML-like language, and the compiler's correctness property -- fair, termination-preserving refinement (\Sref{sec:session-lang}).
Then we present our higher-order concurrent separation logic for establishing said refinement (\Sref{sec:iris}).
We follow on by explaining the key changes to Iris that were necessary to perform this kind of reasoning (\Sref{sec:extensions}).  We then use the
extended logic to prove the correctness of the compiler for our session-typed language (\Sref{sec:session-proof}). Finally, we conclude by
describing connections to related work and limitations of our approach
that we hope to address in future work (\Sref{sec:conclusion}).

\newcommand\bufferfig{\begin{figure}\centering
  \begin{tikzpicture}
    \newcommand\vdist{0.4}
    \tikzstyle{buffer_box} = [draw,rectangle,minimum size=0.5cm,text width=0.35cm,inner sep=0];
    \node[buffer_box] at (0,\vdist) {};
    \node[buffer_box] at (0.5,\vdist) {};
    \node[buffer_box] at (1.0,\vdist) {};
    \node at (1.75,\vdist) {$\cdots$};
    \node[buffer_box] at (2.5,\vdist) {};
    \node[buffer_box] at (3.0,\vdist) {};
    \node[buffer_box] at (3.5,\vdist) {$\Val$};

    \node[buffer_box] at (0,-\vdist) {};
    \node[buffer_box] at (0.5,-\vdist) {};
    \node[buffer_box] at (1.0,-\vdist) {};
    \node at (1.75,-\vdist) {$\cdots$};
    \node[buffer_box] at (2.5,-\vdist) {};
    \node[buffer_box] at (3.0,-\vdist) {};
    \node[buffer_box] at (3.5,-\vdist) {};

    \node[align=right] at (-1.15,0) {End-point\\$\chanloc_\lside$};
    \node[align=left] at (4.65,0) {End-point\\$\chanloc_\rside$};
    \node at ($(1.75,2.55*\vdist)$) {Buffer $b_1$};
    \node at ($(1.75,-2.55*\vdist)$) {Buffer $b_2$};

    \draw[->]  ($(0,2*\vdist)$) -- ($(3.5,2*\vdist)$);
    \draw[<-]  ($(0,-2*\vdist)$) -- ($(3.5,-2*\vdist)$);
  \end{tikzpicture}
  \caption{Depiction of a channel. Here, message $\Val$ has been sent from the left end-point to the right end-point.}\label{fig:buffers}
\end{figure}
}

\newcommand\sessionrulefig{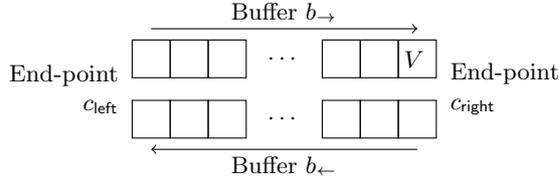
\begin{figure*}
\noindent \textbf{Syntax:}
\begin{mathpar}
\begin{array}{llcl}
\textdom{Side} & \side  & ::= & \lside \ALT \rside \\
\textdom{Val} & \Val & ::= & 
   \locsidevar \ALT
   \lambda \var.\, \Expr_1 \ALT
   (\Val_1, \Val_2) \ALT 
   \vunit \ALT 
   n \qquad\qquad \text{where $\chanloc \in \mathbb{N}$} \\
\textdom{Expr} & \Expr & ::= & \var \ALT 
   \Val \ALT
   \Expr_1 \, \Expr_2 \ALT
   (\Expr_1, \Expr_2) \ALT
   \fork{\Expr} \ALT
   \newch \ALT
   \recv{\Expr}  \\ &&&  \ALT
   \send{\Expr_1}{\Expr_2} \ALT
   \letp{\var}{\varB}{\Expr_1}{\Expr_2} \ALT
   \dots \\
\textdom{Eval Ctx} & \lctx & ::= & [] \ALT
   \lctx \, \Expr \ALT 
   \Val \, \lctx \ALT 
   (\lctx, \Expr) \ALT 
   (\Val, \lctx) \ALT 
   \recv{\lctx}  \ALT 
   \send{\lctx}{\Expr} \\ &&&\ALT
   \send{\Val}{\lctx} \ALT
   \letp{\var}{\varB}{\lctx}{\Expr} \ALT 
   \dots \\
\textdom{State} & \State & \in & \mathbb N \rightarrow
                      \textdom{List Val} \times \textdom{List Val} \\
\textdom{Config} & \Cfg & ::= & \cfg{\Expr_1, \dots, \Expr_n}{\State} \\
\textdom{Type} & \typ & ::= & 
    \inttyp \ALT 
    \unit \ALT
    \typ_1 \otensor \typ_2 \ALT
    \typ_1 \lolli \typ_2 \ALT
    \styp \\
\textdom{Session Type} & \styp  & ::= & 
   \sendtyp{\typ}{\styp} \ALT
   \recvtyp{\typ}{\styp} \ALT
   \stypend \qquad (\text{co-inductive}) \\
\textdom{Dual Type} & \multicolumn{3}{c}{\dual{\recvtyp{\typ}{\styp}} \eqdef \sendtyp{\typ}{\dual{\styp}}\qquad~
\dual{\sendtyp{\typ}{\styp}} \eqdef \recvtyp{\typ}{\dual{\styp}}\qquad~
\dual{\stypend} \eqdef \stypend
}
\end{array}
\end{mathpar}
\noindent \textbf{Per-Thread Reduction} $\Expr; \State \step \Expr'; \State'$: \hfill (Pure and symmetric rules ommitted.)
\begin{mathpar}
\inferH{NewCh}{\chanloc = \min\{\chanloc' \ | \ \chanloc' \not \in \dom(\State)\}}
      {\newch; \State 
       \step
       (\locside{\chanloc}{\lside}, \locside{\chanloc}{\rside}); [\chanloc \mapsto ([], [])]\State
      }

\inferH{SendLeft}{\State(\chanloc) = (\buflr, \bufrl)}
      {\send{\locside{\chanloc}{\lside}}{\Val}; \State
       \step
       \locside{\chanloc}{\lside}; [\chanloc \mapsto (\buflr \Val, \bufrl)]\State
      }

\inferH{RecvRightIdle}{\State(\chanloc) = ([],\bufrl)}
      {\recv{\locside{\chanloc}{\rside}}; \State
       \step
       \recv{\locside{\chanloc}{\rside}}; \State
      }

\inferH{RecvRight}{\State(\chanloc) = (\Val\, \buflr, \bufrl )}
      {\recv{\locside{\chanloc}{\rside}}; \State
       \step
       (\locside{\chanloc}{\rside}, \Val); [\chanloc \mapsto (\buflr, \bufrl)]\State
      }
\end{mathpar}
\noindent \textbf{Concurrent Semantics} $\Cfg \step \Cfg'$:
\begin{mathpar}
\newcommand\sqStep{\kern-0.5ex\step\kern-0.5ex}
\infer{\Expr_i ; \State \sqStep \Expr_i' ; \State' }
      {\cfg{\dots,\lctx[\Expr_i],\dots}{\State}
       \sqStep
       \cfg{\dots,\lctx[\Expr_i'],\dots}{\State'}
      }

\hspace{-1em}
\infer{}
      {\cfg{\dots,\lctx[\fork{\Expr_\f}],\dots}{\State}
       \sqStep
       \cfg{\dots,\lctx[\vunit],\dots,\Expr_\f}{\State}
      }
\end{mathpar}
\noindent \textbf{Type system:}\hfill(Standard rules for variables, integers and lambda omitted.)
\begin{mathpar}
\inferH{Fun-Elim}{\Gamma \vdash \Expr : \typ_1 \lolli \typ_2 \\
       \Gamma' \vdash \Expr' : \typ_1 }
      {\Gamma \uplus \Gamma' \vdash \Expr\, \Expr' : \typ_2}

\inferH{Pair-Intro}{\Gamma_1 \vdash \Expr_1 : \typ_1 \\
       \Gamma_2 \vdash \Expr_2 : \typ_2 }
      {\Gamma_1 \uplus \Gamma_2 \vdash (\Expr_1, \Expr_2) : \typ_1 \otensor \typ_2}
      
\inferH{Pair-Elim}{\Gamma \vdash \Expr : \typ_1 \otensor \typ_2 \\
       \Gamma', \var : \typ_1, \varB : \typ_2 \vdash \Expr' : \typ'}
      {\Gamma \uplus \Gamma' \vdash \letp{\var}{\varB}{\Expr}{\Expr'} : \typ'}
      
\inferH{Fork}{\Gamma_1 \vdash \Expr_\f : \typ' \and \Gamma_2 \vdash \Expr : \typ}
      {\Gamma_1 \uplus \Gamma_2 \vdash \fork{\Expr_\f}; \Expr  : \typ}
      
\inferH{NewChTyp}{}
      {\Gamma \vdash \newch  : \styp \otensor \dual{\styp}}
      
\inferH{Send}{\Gamma_1 \vdash \Expr_1 : \sendtyp{\typ}{\styp} \\
       \Gamma_2 \vdash \Expr_2 : \typ
      }
      {\Gamma_1 \uplus \Gamma_2 \vdash \send{\Expr_1}{\Expr_2}  : \styp}

\inferH{Recv}{\Gamma \vdash \Expr : \recvtyp{\typ}{\styp} \\
      }
      {\Gamma \vdash \recv{\Expr} : \styp \otensor \typ }
\end{mathpar}
\noindent \textbf{Buffer visualization:}\hfill Message $\Val$ has been sent from the left end-point to the right.
  \centering\begin{tikzpicture}
    \newcommand\vdist{0.4}
    \tikzstyle{buffer_box} = [draw,rectangle,minimum size=0.5cm,text width=0.35cm,inner sep=0];
    \node[buffer_box] at (0,\vdist) {};
    \node[buffer_box] at (0.5,\vdist) {};
    \node[buffer_box] at (1.0,\vdist) {};
    \node at (1.75,\vdist) {$\cdots$};
    \node[buffer_box] at (2.5,\vdist) {};
    \node[buffer_box] at (3.0,\vdist) {};
    \node[buffer_box] at (3.5,\vdist) {$\Val$};

    \node[buffer_box] at (0,-\vdist) {};
    \node[buffer_box] at (0.5,-\vdist) {};
    \node[buffer_box] at (1.0,-\vdist) {};
    \node at (1.75,-\vdist) {$\cdots$};
    \node[buffer_box] at (2.5,-\vdist) {};
    \node[buffer_box] at (3.0,-\vdist) {};
    \node[buffer_box] at (3.5,-\vdist) {};

    \node[align=right] at (-1.15,0) {End-point\\$\chanloc_\lside$};
    \node[align=left] at (4.65,0) {End-point\\$\chanloc_\rside$};
    \node at ($(1.75,2.55*\vdist)$) {Buffer $\buflr$};
    \node at ($(1.75,-2.55*\vdist)$) {Buffer $\bufrl$};

    \draw[->]  ($(0,2*\vdist)$) -- ($(3.5,2*\vdist)$);
    \draw[<-]  ($(0,-2*\vdist)$) -- ($(3.5,-2*\vdist)$);
  \end{tikzpicture}
\vspace{-1em}

\caption{Syntax, semantics, and session type system of message-passing source language}
\label{fig:session-lang-rules}
\end{figure*}
}

\newcommand\implfig{\begin{figure}[t]
\begin{mathpar}
\begin{minipage}[t]{0.35\textwidth}
\[
\begin{array}{ll}
&\heapalloc \eqdef\\
&\quad \bind{l}{\alloc{\none}}{(l, l)}
\end{array}
\] 
\end{minipage}
\begin{minipage}[t]{0.2\textwidth}
\[
\begin{array}{ll}
&\heapsend\  l\, v\eqdef  \\
&\quad \letp{l'}{v'}{(l, v)}{} \\
&\quad \bind{l_{new}}{\alloc{\none}}{} \\
&\quad \store{l'}{\some{(l_{new}, v')}}; \\
&\quad l_{new}
\end{array}
\] 
\end{minipage}
\hfill
\begin{minipage}[t]{0.35\textwidth}
\[
\begin{array}{ll} 
&\heaprecv \eqdef \rec{f}{l} \\
&\quad \match{\load{l}} \\
&\quad\quad \ALT \patcase{\none}{f\, l} \\
&\quad\quad \ALT \patcase{\some{(l',v)}}{(l', v)} \\
&\quad \matchend
\end{array}
\]
\end{minipage}
\end{mathpar}
\caption{Implementation of message passing primitives.}
\label{fig:implementation}
\end{figure}
}

\section{Session-Typed Language and Compiler}
\label{sec:session-lang}

This section describes the case study that we chose to demonstrate our logic: a concurrent message-passing language and a type system establishing safety and race-freedom for this language.
On top of that, we explain how to implement the message-passing primitives in terms of shared-memory concurrency, \ie we define a compiler translating the source language into an ML-like target language.
Finally, we discuss the desired correctness statement for this compiler.

\subsection{Source Language}
\sessionrulefig

The source language for our compiler is a simplified version of the
language described in \citet{GayV10}. The syntax and semantics are given
in \figref{fig:session-lang-rules}. It is a
functional language extended with primitives for message passing and a command $\fork{\Expr}$ for creating threads.
The semantics is defined
by specifying a reduction relation for a single thread, which is then
lifted to a concurrent semantics on thread-pools in which at each step a thread is
selected non-deterministically to take the next step.

Threads can
communicate asynchronously with each other by sending messages over
\emph{channels}. For example, consider the following program (which will be a running example of the paper):
\begin{equation}
\begin{aligned}
& \letp{\var}{\varB}{\newch}{}  %
\big(\fork{\send\var{42}};\, %
 \letp\_{v}{\recv\varB}{v}\big)
\end{aligned}\label{eq:running-example}
\end{equation}

The command $\newch$ creates a new channel and returns two \emph{end-points} (bound to $\var$ and $\varB$ in the example). 
An end-point consists of a channel id $\chanloc$ and a side $\side$ (either $\lside$ or $\rside)$, and is written as $\locsidevar$.
Each channel is a pair of buffers $(\buflr, \bufrl$), which are
lists of messages.  Buffer $\buflr$ stores messages traveling
left-to-right (from $\var$ to $\varB$, in the example above), and $\bufrl$ is for right-to-left messages, as shown in the visualization in \figref{fig:session-lang-rules}.

A thread can then use $\send{\locsidevar}{\Val}$ to send a value $\Val$ along
the channel $c$, with the side $\side$ specifying which buffer is used to store the message. For instance, when  $\side$ is $\lside$, it inserts the value at the end of the first buffer (\ruleref{SendLeft}).
This value will then later be taken by a thread receiving on the \emph{right} side (\ruleref{RecvRight}).
Alternatively, if the buffer is empty when receiving, $\textlog{recv}$ takes an ``idle'' step and tries again (\ruleref{RecvRightIdle}).
(The reason $\textlog{send}$ and $\textlog{recv}$ return the end-point again will become clear when we explain the type system.)

In the example above, after creating a new channel, the initial thread
forks off a child which will send $42$ from the left end-point,~$\var$. %
Meanwhile, the parent thread tries to receive from the right
end-point~$\varB$, and returns the message it gets. If the parent
thread does this $\textlog{recv}$ \emph{before} the child has done its
$\textlog{send}$, there will be no message and the parent thread will take an idle
step.
Otherwise, the receiver will see the message and the program will evaluate to $42$.

\subsection{Session Type System}
\label{sec:session-types}

 A type system for this language is shown in
 \figref{fig:session-lang-rules}. This is a simplified version of
 the type system given in \citet{GayV10}.\footnote{For the reader
   familiar with that work: we leave out subtyping and choice
   types. Also, we present an affine type system instead of a linear
   one.}  In addition to base types $\inttyp$ and $\unit$, we have
 pair types $\tau_1 \otensor \tau_2$, function types $\tau_1
 \lolli \tau_2$, and \emph{session types}~$\styp$.
 Session types are used to type the end-points of a channel. These
types describe a kind of \emph{protocol} specifying what types of data
will flow over the channel, and in what order messages are sent.
Notice that this type system is \emph{higher-order} in the sense that both closures and channel end-points are first-class values and can, in particular, be sent over channels.

\paragraph{Session types.}
The possible session types are specified by the grammar in
\figref{fig:session-lang-rules}. If an end-point has the session type
$\sendtyp{\tau}{\styp}$, this means that the next use of this
end-point must be to send a value of type $\tau$
(\ruleref{Send}). Afterward, the
end-point that is returned by the send will have type $\styp$. Dually, $\recvtyp{\tau}{\styp}$ says that the
end-point can be used in a receive (\ruleref{Recv}), in which case the message read will
have type $\tau$, and the returned end-point will have type $\styp$.
Notice that this is the same end-point that was passed to the command, but \emph{at a different type}.
The type
of the end-point \emph{evolves} as messages are sent and
received, always representing the current state of the protocol.
Finally, $\stypend$ is a session type for an end-point on
which no further messages will be sent or received.

When calling $\newch$ to create a new channel, it is important that the types of the two end-points match: whenever one side sends a message of type $\typ$, the other side should be expecting to receive a message of the same type.
This relation is called \emph{duality}.
Given a session type $\styp$, its \emph{dual} $\dual{\styp}$ is the result of swapping sends and receives in $\styp$.
In our example~(\ref{eq:running-example}), the end-point $\var$
is used to send a single integer,
so it can be given the type $\sendtyp{\inttyp}{\stypend}$. Conversely,
$\varB$ receives a single integer, so it has the dual type
$\dual{\sendtyp{\inttyp}{\stypend}} = \recvtyp{\inttyp}{\stypend}$.

\paragraph{Affinity.}
The type system of the source language is \emph{affine}, which means that a variable in the context can be used at most once.
This can be seen, \eg in the rule \ruleref{Fork}: 
the forked-off thread $\Expr_\f$ and the local continuation $\Expr$ are typed using the two disjoint contexts $\Gamma_1$ and $\Gamma_2$, respectively.

One consequence of affinity is that 
 after using an end-point to send or receive, the variable passed to $\textlog{send}$/$\textlog{recv}$ has been ``used up'' and cannot be used anymore.
Instead, the program has to use the channel returned from $\textlog{send}$/$\textlog{recv}$,
which has the new ``evolved'' type for the end-point.

The type system given here ensures safety and race-freedom.  However,
it does not guarantee termination. We discuss alternative type systems
guaranteeing different properties in the conclusion.

\subsection{Compilation}
We now describe a simple translation from this session-typed
source language to a MiniML language with references and a forking
primitive like the one in the source language.
We omit the details of the MiniML syntax and semantics as they are standard.

Our translation needs
to handle essentially one feature: the implementation of channel
communication in terms of shared memory references.

The code for the implementation of the channel primitives is shown in
\figref{fig:implementation}. We write $\compile{\Expr}$ for the
translation in which we replace the primitives of the source language
with the corresponding implementations.
Concretely, applying the translation to our running example program we get:
\begin{align*}
  & \letp{\var}{\varB}{\newch}{} & & \letp{\var}{\varB}{\heapnewch}{} \\
& \fork{\send{\var}{42}}; & \quad\ \Ra\quad\quad & \fork{\heapsend\,\var\,{42}}; \\
& \letp\_{v}{\recv{\varB}}{v} & & \letp\_{v}{\heaprecv\,\varB}{v}
\end{align*}

Each channel is implemented as
a linked list which represents \emph{both} buffers.  Nodes in this
list are pairs $(\heaploc, \val$), where $\heaploc$ is a reference to
the (optional) next node, and $\val$ is the message that was sent.
Why is it safe to use just one list?  Duality in the session types guarantees that if a thread is
sending from one end-point, no thread can at the same time be sending a message on the
other end-point. This ensures that at least one of the two
buffers in a channel is always empty.
Hence we just need one list to
represent both buffers.

The implementation of $\newch$, given by $\heapalloc$, creates a new
empty linked list by allocating a new reference $\heaploc$ which initially
contains $\none$.
The function $\heapsend$ implements $\textlog{send}$ by appending a
node to the end ($l'$) of the list, and returning the new end.
Meanwhile, for $\textlog{recv}$, $\heaprecv$ takes an end-point
$\heaploc$ and waits in a loop until it finds that the end-point contains a node.

\implfig

\subsection{Refinement}
\label{sec:refinement}

Having given the implementation, let us now clarify what it means for
the compiler to be correct. Intuitively, we want to show that if we
take a well-typed source expression $\Expr$, all the \emph{behaviors} of its
translation $\compile{\Expr}$ are also \emph{possible} behaviors of $\Expr$.
We say that $\compile{\Expr}$ \emph{refines} $\Expr$.

Before we come to the formal definition of refinement, we need to answer the question: which behaviors do we consider equivalent?
In our case, the only observation that can be made about a whole program is its return value, so classifying ``behaviors'' amounts to relating return values.
Formally speaking:
\begin{mathpar}

\infer{}{n \valobsrel n}

\infer{}{\vunit \valobsrel \vunit}

\infer{}{\heaploc \valobsrel \locsidevar}

\infer{}{\lambda x. \expr \valobsrel \lambda x. \Expr}

\infer{\val_1 \valobsrel \Val_1 \\ 
       \val_2 \valobsrel \Val_2}{(\val_1, \val_2) \valobsrel (\Val_1, \Val_2)}
\end{mathpar}

For integer and unit values, we
expect them to be exactly equal; similarly, pairs are the same if
their components are.
Coming to locations/end-points and closures, we do not consider them to be interpretable by the user looking at the result of a closed program.
So, we just consider all closures to be equivalent, and all heap locations to relate to all channel end-points.
Of course, the \emph{proof} of compiler correctness will use a more fine-grained logical relation between source and target values.

Based on this notion of equivalent observations, we define what it means for a MiniML program $\expr$ to \emph{refine} a source program $\Expr$, written $\expr \refines \Expr$.
When executing from an initial ``empty'' state $\initstate$, the following
 conditions must hold:
\begin{enumerate}
\item If $([\expr], \initstate) \step^{*} ([\expr_{1}, \dots, \expr_{n}], \state)$ then
         no $\expr_{i}$ is stuck in state $\state$.

         In other words: the target program does not reach a stuck state.

\item If $([\expr], \initstate) \step^{*} ([\val_{1}, \dots,
  \val_{n}], \state)$ then either:
   \begin{enumerate}
    \item  $([\Expr], \initstate) \step^{*}
  ([\Val_{1}, \dots, \Val_{m}], \State)$ and $\val_{1} \valobsrel
  \Val_{1}$, or
 
    \item there is an execution of $([\Expr], \initstate)$
  in which some thread gets stuck.
  \end{enumerate}

  That is, if \emph{all threads} of the target program terminate with a value, then either \emph{all threads} of the source program terminate in some execution \emph{and} the return values of the first (main) source thread and target thread are equivalent;  or the source program can get stuck.
         
\item 
  If $([\expr], \initstate)$ has a fair diverging execution, then $([\Expr], \initstate)$ also has a fair diverging execution.
  Recall that an infinite execution is \emph{fair} if every non-terminating thread takes infinitely many steps.
  This last condition makes the refinement a \emph{fair, termination-preserving} refinement.
\end{enumerate}

To understand why we have emphasized the importance of fair termination-preservation, suppose
we had miscompiled our running example as:
\begin{align*}
& \letp{\var}{\varB}{\heapnewch}{}  \letp\_{v}{\heaprecv\,\varB}{v}
\end{align*}
That is, we removed the sender thread. We consider this to be
an incorrect compilation; \ie
this program should \emph{not} be considered a refinement of the
source program. But imagine that we removed the word ``fair'' from
condition (3) above: then this bad target program would be considered a
refinement of the source. How is that? The program does not get stuck, so it satisfies condition (1).
Condition (2) holds vacuously since the target program will never terminate; it will loop in
$\heaprecv\,\varB$, forever waiting for a message. Finally,
to satisfy condition (3), we have to exhibit a diverging execution in the source program.
Without the fairness constraint, we can pick the (unfair) execution in which
the sender source thread never gets to run.

Notice that this unfair execution is very much like the example we gave in the introduction, where a thread waited forever for another one to perform a change in the shared state.

We consider such unfair executions to be unrealistic~\cite{LehmannPS81}; they should not give license to a compiler to entirely remove a thread from the compiled program.
That's why our notion of refinement restricts condition (3) to \emph{fair} executions, \ie executions in which all non-terminating threads take infinitely many steps.

\paragraph{Compiler correctness.}
We are now equipped to formally express the correctness statement of our compiler:
\begin{thm}\label{thm:compiler}
For every \emph{well-typed} source program $\Expr$, we have that:
\[\compile{\Expr} \refines \Expr\]
\end{thm}
We prove this theorem in \secref{sec:session-proof}. In the intervening sections, we first develop and explain a logic to help carry out this proof.

\newcommand\sourcerulefig{\begin{figure}
\noindent \textbf{Step Shift Rules:} (all $d$ and $d'$ must be $\leq$ some fixed upper-bound $D$)
\let\oldcr\\
\begin{mathpar}
\inferH{src-newch}{}
{
{\begin{aligned}
    &\ownThreadD{i}{\lctx[\newch]}{\delay} \svs %
    \Exists\chanloc. \ownThreadD{i}{\lctx[(\chanloc_\lside, \chanloc_\rside)]}{\delay'} *
    \chanloc \smapsto ([],[])
  \end{aligned}}
}

\inferH{src-recv-right-miss}{}
{
{\begin{aligned}
    &\ownThreadD{i}{\lctx[\recv{\chanloc_\rside}]}{\delay} * \chanloc \smapsto ([],
    \bufrl) \svs %
     \ownThreadD{i}{\lctx[\recv{\chanloc_\rside}]}{\delay'} * \chanloc
    \smapsto ([], \bufrl)
  \end{aligned}}
}

\inferH{src-recv-right-hit}{}
{
{\begin{aligned}
    &\ownThreadD{i}{\lctx[\recv{\chanloc_\rside}]}{\delay} * \chanloc \smapsto (\val\,\buflr,
    \bufrl) \svs %
    \ownThreadD{i}{\lctx[(\chanloc_\rside,\val)]}{\delay'} * \chanloc
    \smapsto (\buflr, \bufrl)
  \end{aligned}}
}

\inferH{src-send-left}{}
{
  {\begin{aligned}
    &\ownThreadD{i}{\lctx[\send{\chanloc_\lside}\val]}{\delay} * \chanloc \smapsto (\buflr, \bufrl)
    \svs %
    \ownThreadD{i}{\lctx[\chanloc_\lside]}{\delay'} * \chanloc \smapsto (\buflr\,\val,
    \bufrl)
  \end{aligned}}
}

\inferH{src-fork}{}
{
{\begin{aligned}
  &\ownThreadD{i}{\lctx[\fork\Expr]}{\delay} \svs %
   \Exists j. \ownThreadD{i}{\lctx[()]}{\delay'} 
   * \ownThreadD{j}{\Expr}{\delay_\f}
  \end{aligned}}
}

\inferH{src-delay}{}
{d' < d \vdash \ownThreadD{i}{\lctx[\Expr]}{\delay} \svs \ownThreadD{i}{\lctx[\Expr]}{\delay'}}

\inferH{src-pure-step}
{\expr_1 \step \expr_2}
{\ownThreadD{i}{\expr_1}{d} \svs \ownThreadD{i}{\expr_2}{d'}}

\inferH{src-stopped}{}
{\ownThreadD{i}{\Val}{0} \vdash \stopped}

{\text{(Symmetric rules and side-condition on $d'$ omitted.)}}
\end{mathpar}
\noindent \textbf{Basic Hoare Triples:}
\begin{mathpar}
\inferH{ml-alloc} {\forall \var.\,P \svs Q}
       {\hoare{\prop}{\alloc{\val}}{\Ret\var. Q * \var\mapsto\val}[]}

\inferH{ml-load}
{P  \svs{}{} [\val/\varB]Q}
{\hoare{P * \var \mapsto \val}{\load\var}{\Ret\varB. Q * \var \mapsto \val}[]}

\inferH{ml-store}
{P \svs{}{} Q}
{\hoare{P * \var \mapsto \val}{\store\var\valB}{Q * \var \mapsto \valB}}

\inferH{ml-fork}
{P \svs Q_0 * Q_1 \\\\ \hoare{Q_0}{\expr}{\stopped} \and \hoare{Q_1}{\expr'}\propC}
{\hoare{P}{\fork\expr; \expr'}{\propC}[]}

\inferH{ml-rec}
{ \prop \svs \prop' \and (\All \val. \hoare {\prop} {(\rec f x  \expr)\,\val} {\Ret \valB. \propB}) \Ra  \All \val. \hoare{\prop'} {[\rec f x  \expr/f,\val/x]\expr} {\Ret \valB. \propB} 
 }
{   \All\val. \hoare {\prop} {(\rec f x  \expr)\,\val} {\Ret\valB. \propB} }

\inferH{Ht-frame}
  {\hoare{\prop}{\expr}{\Ret\val.\propB}[]}
  {\hoare{\prop * \aff{\propC}}{\expr}{\Ret\val.\propB * \aff{\propC}}[]}

\inferH{step-frame}
{ \prop \svs \propB}
{\prop * \aff{\propC} \svs \propB * \aff{\propC}}
  
\inferH{Ht-csq}
  {\prop \vs \prop' \\
    \hoare{\prop'}{\expr}{\Ret\val.\propB'} \\   
   \All \val. \propB' \vs \propB}
  {\hoare{\prop}{\expr}{\Ret\val.\propB}}
\end{mathpar}
\noindent \textbf{Refinement Rule:}
\begin{mathpar}
  \inferH{Ht-refine}
  {\hoare{\ownThreadD{i}{\Expr}{\delay}}{\expr}{\Ret \val. \Exists\Val. \ownThreadD{i}{\Val}{0} * \val \valobsrel \Val}[]}
  {\expr \refines \Expr}
\end{mathpar}
\caption{Selection of rules for step shifts and Hoare triples}
\label{fig:hoare-rules}
\end{figure}
}

\newcommand\examplestsfig{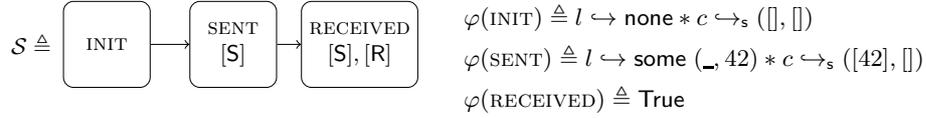
\begin{figure}[t]
  \centering
  \begin{minipage}[c]{0.5\linewidth}
   \begin{tikzpicture}[sts, transform canvas={scale=0.95}]
    \node[draw=none] at (.5,0) {$\mathcal{\STSS}\eqdef{}$};
    \node at (1.5,0)  (init){$\textsc{init}$};
    \node at (3.25,0)  (sent) {$\textsc{sent}$\\$[\textlog{S}]$};
    \node at (5,0)  (recv) {$\textsc{received}$\\$[\textlog{S}],[\textlog{R}]$};
    \path[sts_arrows] (init) edge  (sent)
        (sent) edge (recv);
  \end{tikzpicture}
  \end{minipage}\hfill%
  \begin{minipage}[c]{0.5\linewidth}
  \begin{align*}
    &\pred(\textsc{init}) \eqdef{} l \mapsto \none * \chanloc \smapsto ([],[]) \\
    &\pred(\textsc{sent}) \eqdef{} l \mapsto \some{(\any, 42)} * \chanloc \smapsto ([42],[]) \\
    &\pred(\textsc{received}) \eqdef{} \TRUE 
  \end{align*}
  \end{minipage}
  \caption{STS for the example}
  \label{fig:example-sts}
\end{figure}
}

\section{A Logic for Proving Refinement}
\label{sec:iris}

Proving \thmref{thm:compiler} is a challenging exercise.
Both the source and the target program are written in a concurrent language with higher-order state, which is always a difficult combination to reason about. Moreover, the invariant relating the channels and buffers to their implementation as linked lists is non-trivial and relies on well-typedness of the source program.

The contribution of this paper is to provide a logic powerful enough to prove theorems like \thmref{thm:compiler}.
In this section, we will give the reader an impression of both the logic and the proof by working through a proof of one concrete instance of our general result: we will prove that the translation of our running example is in fact a refinement of its source.

\subsection{Refinement as a Hoare Logic}
\label{sec:iris-overview}

Our logic is an extension of \emph{Iris}~\citep{iris,iris2}, a concurrent higher-order separation logic.
We use the ideas presented by~\citet{caresl} to extend this (unary) Hoare logic with reasoning principles for refinement.
Finally, we add some further extensions which become necessary due to the \emph{termination-preserving} nature of our refinement.
We will highlight these extensions as we go.

\sourcerulefig

The following grammar covers the assertions from our logic that we will need:\footnote{Note that many of these assertions are not primitive to the logic, but are themselves defined using more basic assertions provided by the logic. For instance, the Hoare triple is actually defined in terms of a \emph{weakest precondition} assertion. See \citet{iris,iris2} for further details.}
\begin{align*}
  \prop \bnfdef{}&
    \FALSE \mid
    \TRUE \mid
    \prop \lor \prop \mid
    \prop * \prop \mid
    \aff\prop \mid
    \Exists \var. \prop \mid
    \All \var. \prop \mid
    l \mapsto \val \mid
    \ownThreadD{i}{\Expr}{d} \mid
\\&
    \stopped \mid
    \chanloc \smapsto (\buflr, \bufrl) \mid
    \STSSt(s, T) \mid
    \hoare\prop\expr{\Ret\var.\propB} \mid
    \prop \vs \propB \mid
    \prop \svs \propB \mid
    \dots
\end{align*}
Many of these assertions are standard in separation logics, and our example proof will illustrate the non-standard ones.

Recalling the example and its translation,we want to prove:
\begin{align*}
  & \letp{\var}{\varB}{\heapnewch}{} & & \letp{\var}{\varB}{\newch}{} \\
& \fork{\heapsend\,\var\,{42}}; & \quad\refines\quad\quad & \fork{\send\var{42}}; \\
& \letp\_{v}{\heaprecv\,\varB}{v} & & \letp\_{v}{\recv\varB}{v}
\end{align*}
or, for short, $\expr_{\textrm{ex}} \refines \Expr_{\textrm{ex}}$.
Following \ruleref{Ht-refine} (\figref{fig:hoare-rules}), it is enough to prove
\begin{equation} \hoare{\ownThreadD{i}{\Expr_{\textrm{ex}}}\delay}{\expr_{\textrm{ex}}}{\Ret \val. \Exists\Val. \ownThreadD{i}{\Val}0 * \val \valobsrel \Val} \label{eq:ht-refine} \end{equation}

In other words, we ``just'' prove a Hoare triple for $\expr_{\textrm{ex}}$ (the MiniML program).
In order to obtain a refinement from a Hoare proof, we equip our logic with assertions talking about the source program $\Expr$.
The assertion $\ownThreadD{i}{\Expr}\delay$ states that source-level thread $i$ is about to execute $\Expr$,
and we have \emph{delay} $\delay$ left. (We will come back to delays shortly.)
The assertion $\chanloc \smapsto (\buflr, \bufrl)$ says that source-level channel $\chanloc$ currently has buffer contents $(\buflr, \bufrl)$.
As usual in separation logic, both of these assertions furthermore assert \emph{exclusive ownership} of their thread or channel.
For example, in the case of $\chanloc \smapsto (\buflr, \bufrl)$, this means that no other thread can access the channel and we are free to mutate it (\ie send or receive messages) -- we will see later how the logic allows threads to share these resources.
Put together, these two assertions let us control the complete state of the source program's execution.

So far, we have not described anything new.
However, to establish \emph{termination-preserving} refinement, we have to add two features to this logic: \emph{step shifts} and \emph{linear assertions}.

\paragraph{Step shifts.}
The rules given in \figref{fig:hoare-rules} let us manipulate the state of the source program's execution by \emph{taking steps in the source program}.
Such steps are expressed using \emph{step shifts} $\svs$.
Every step shift corresponds to one rule in the operational semantics (\figref{fig:session-lang-rules}).
For example, \ruleref{src-newch} expresses that if we have $\ownThreadD{i}{\lctx[\newch]}\delay$ (which means that the source is about to create a new channel), we can ``execute'' that $\newch$ and obtain some fresh channel $\chanloc$ and ownership of the channel ($\chanloc \smapsto ([], [])$).
We also obtain $\ownThreadD{i}{\lctx[\chanloc]}{\delay'}$, so we can go on executing the source thread.

Crucially, having $\prop \svs \propB$ shows that in going from $\prop$ to $\propB$, the source \emph{has} taken a step.
We need to force the source to take steps because the refinement we show is \emph{termination-preserving}.
If a proof could just decide not to ever step the source program, we could end up with a MiniML program $\expr$ diverging, while the corresponding source program $\Expr$ cannot actually diverge.
That would make \ruleref{Ht-refine} unsound.
So, to avoid this, all rules that take a step in the MiniML program (\figref{fig:hoare-rules}) force us to also take a step shift.

A strict implementation of this idea requires a lock-step execution of source and target program.
This is too restrictive.
For that reason, the $\textlog{source}$ assertion does not just record the state of the source thread, but also a \emph{delay} $\delay$.
Decrementing the delay counts as taking a step in the source (\ruleref{src-delay}). When we take an actual source step, we get to reset the delay to some new $\delay'$ -- so long as $\delay'$ is less than or equal to some fixed upper bound $D$ that we use throughout the proof.
There are also rules that allow executing \emph{multiple} source steps when taking just a single step in the target program; we omit these rules for brevity.
For the remainder of this proof, we will also gloss over the bookkeeping for the delay and just write $\ownThread{i}{\expr}$.

The assertion $\stopped$ expresses that a source thread can no longer take steps. As expected, this happens when the source thread reaches a value (\ruleref{src-stopped}).

\paragraph{Linearity.}
There is one last ingredient we have to explain before we start the actual verification: \emph{linearity}.
Assertions in our logic are generally \emph{linear}, which means they cannot be ``thrown away'', \ie $\prop * \propB \proves \prop$ does not hold generically in $\prop$ and $\propB$.
As a consequence, assertions represent not only the \emph{right} to perform certain actions (like modifying memory), but also the \emph{obligation} to keep performing steps in the source program.
This ensures that we do not ``lose track'' of a source thread and stop performing step shifts justifying its continued execution.

The modality $\aff\prop$ says that we have a proof of $\prop$, and that this is an \emph{affine} proof -- so there are no obligations encoded in this assertion, and we can throw it away.
Some rules are restricted to affine assertions, \eg rules for framing around a Hoare triple or a step shift (\figref{fig:hoare-rules}; the rule \ruleref{Ht-csq} will be explained later).
Again, this affine requirement ensures that we do not ``smuggle'' a source thread around the obligation to perform steps in the source.
All the base assertions, with the exception of $\ownThread{i}\expr$, are affine.

Coming back to the Hoare triple~(\ref{eq:ht-refine}) above that we have to prove, the pre-condition $\ownThread{i}{\Expr_{\textrm{ex}}}$ expresses that we start out with a source program executing $\Expr_{\textrm{ex}}$ (and not owning any channels), and we somehow have to take steps in the source program to end up with $\ownThread{i}{\Val}$ such that $\Val$ is ``equivalent'' (in the sense defined in \Sref{sec:refinement}) to the return value of the target program.
Intuitively, because we can only manipulate $\textlog{source}$ by taking steps in the source program, and because we end up stepping from $\ownThread{i}{\Expr_{\textrm{ex}}}$ to ``the same'' return value as the one obtained from $\expr$, proving the Hoare triple actually establishes a refinement between the two programs.
Furthermore, since $\textlog{source}$ is linear and we perform a step shift at every step of the MiniML program, the refinement holds even for diverging executions.

\subsection{Proof of the Example}
The rest of this section will present in great detail the proof of our example~(\ref{eq:ht-refine}).
The rough structure of this proof goes as follows:  after a small introduction covering the allocation of the channel, we will motivate the need for \emph{state-transition systems} (STS), a structured way of controlling the interaction between cooperating threads.
We will define the STS used for the example and decompose the remainder of the proof into two pieces: one covering the sending thread and one for the receiving thread.

\paragraph{Getting started.}
The first statement in both source and target program is the allocation of a channel.
The following Hoare triple that's easily derived from \ruleref{ml-alloc} summarizes the action of $\heapnewch$: It allocates a channel in \emph{both} programs.
\begin{equation} \hoareHV{\ownThread{i}{\lctx[\newch]}}{\heapnewch}{\Ret\var.
      \Exists l, \chanloc. \var = (l,l) * l \mapsto \none * \chanloc \smapsto
      ([],[]) * %
      \ownThread{i}{\lctx[(\chanloc_\lside, \chanloc_\rside)]}
} 
\label{eq:ht-newch}
\end{equation}
Let us pause a moment to expand on that post-condition.
On the source side, we have a channel $c$ with both buffers being empty; on the target side we have a location $l$ representing the empty buffer with $\none$.
The return value $\var$ is a pair with both components being $l$.
Finally, the source thread changed from $\lctx[\newch]$ in the pre-condition to $\lctx[(\chanloc_\lside, \chanloc_\rside)]$, meaning that the $\newch$ has been executed and the context can now go on with its evaluation based on the pair $(\chanloc_\lside, \chanloc_\rside)$.

We apply this triple for $\heapnewch$ with the appropriate evaluation context $\lctx$ for the source program, and the post-condition of (\ref{eq:ht-newch}) becomes our new context of current assertions.
Next, we reduce the $\textlog{let}$ on both sides, so we end up with
\begin{equation} l \mapsto \none *
    \chanloc \smapsto ([],[])  * \ownThread{i}{\expr_{\mathrm{comm}}(\chanloc)}\label{eq:proofstate-1} \end{equation}
where
\[ \expr_{\mathrm{comm}}(\chanloc) \eqdef \fork{\send{\chanloc_\lside}{42}}; \letp\_{v}{\recv{\chanloc_\rside}}{v} \]
and the remaining MiniML code is
\[  \fork{\heapsend\,{l}\,{42}}; \letp\_{v}{\heaprecv\,{l}}{v} \]
(In the following, we will perform these pure reduction steps and the substitutions implicitly.)

As we can see, both programs are doing a $\textlog{fork}$ to concurrently send and receive messages on the same channel.
Usually, this would be ruled out by the exclusive nature of ownership in separation logic.
To enable sharing, the logic provides a notion of \emph{protocols} coordinating the interaction of multiple threads on the same shared state.
The protocol governs ownership of both $l$ (in the target) and $c$ (in the source), and describes which thread can perform which actions on this shared state.

\paragraph{State-transition systems.}
A structured way to describe protocols is the use of state-transition systems (STS), following the ideas of \citet{caresl}.
An STS $\STSS$ consists of a directed graph with the nodes denoting \emph{states} and the arrows denoting \emph{transitions}.

\examplestsfig

The STS for our example is given in \figref{fig:example-sts}.
It describes the interaction of our two threads over the shared buffer happening in three phases.
In the beginning, the buffer is empty (\textsc{init}).
Then the message is sent by the forked-off sending thread (\textsc{sent}).
Finally, the message is received by the main thread (\textsc{received}).

The STS also contains two \emph{tokens}.
Tokens are used to represent actions that only particular threads can perform.
In our example, the state \textsc{sent} requires the token $[\textlog{S}]$.
The STS enforces that, in order to step from \textsc{init} to \textsc{sent}, a thread must \emph{provide} (and give up) ownership of $[\textlog{S}]$.
This is called the \emph{law of token preservation}~\cite{caresl}: Because \textsc{sent} contains more tokens than \textsc{init}, the missing tokens have to be provided by the thread performing the transition.
Similarly, $[\textlog{R}]$ is needed to transition to the final state \textsc{received}.

To tie the abstract state of the STS to the rest of the verification, every STS comes with an \emph{interpretation} $\pred$.
For every state, it defines an affine assertion that has to hold at that state.
In our case, we require the buffer to be initially empty, and to contain $42$ in state \textsc{sent}.
Once we reach the final state, the programs no longer perform any action on their respective buffers, so we stop keeping track.

We need a way to track the state of the STS in our proof.
To this end, the assertion $\STSSt(s, T)$ states that the STS is \emph{at least} in state $s$, and that we own tokens $T$.
We cannot know the \emph{exact} current state of the STS because other threads may have performed further transitions in the mean time.
The proof rules for STSs can be found in the \appref; in the following, we will keep the reasoning about the STS on an intuitive level to smooth the exposition.

\paragraph{Plan for finishing the proof.}
Let us now come back to our example program.
We already described the STS we are going to use for the verification (\figref{fig:example-sts}).
The next step in the proof is thus to initialize said STS.

Remember our current context is (\ref{eq:proofstate-1}).
When allocating an STS, we get to pick its initial state -- that would be \textsc{init}, of course.
We have to provide $\pred(\textsc{init})$ to initialize the STS, so we give up ownership of $l$ and $\chanloc$.
In exchange, we obtain $\STSSt$ and the tokens.
Our  context is now
\begin{equation} \STSSt(\textsc{init}, \{[\textlog{S}],[\textlog{R}]\}) * \ownThread{i}{\expr_{\mathrm{comm}}(\chanloc)} \label{eq:proofstate-2} 
     \end{equation}

The next command executed in both programs is $\textlog{fork}$.
We are thus going to apply \ruleref{ml-fork} and prove the step shift using \ruleref{src-fork}.
The two remaining premises of \ruleref{ml-fork} are the following two Hoare triples:
\begin{align}
  &\!\!\hoare{\STSSt(\textsc{init}, [\textlog{S}]) * \ownThread{j}{\send{\chanloc_\lside}{42}}}{\heapsend\,l\,{42}}{\stopped} \label{eq:ht-send}
\\[0.7em]
  &\!\!\hoareV{\STSSt(\textsc{init}, [\textlog{R}]) * \ownThread{j}{\letp\_{v}{\recv{\chanloc_\rside}}{v}}}{\letp\_{v}{\heaprecv\,l}{v}}{\Ret n. n = 42 *  \ownThread{j}{42} } \label{eq:ht-recv}
\end{align}
Showing these will complete the proof.
The post-condition $\stopped$ of (\ref{eq:ht-send}) is mandated by \ruleref{ml-fork}; we will discuss it when verifying that Hoare triple.
Note that we are splitting the $\STSSt$ to hand the two tokens that we own to two different threads.

\paragraph{Verifying the sender.}
To prove the sending Hoare triple (\ref{eq:ht-send}), the context we have available is
$ \STSSt(\textsc{init}, [\textlog{S}]) * \ownThread{j}{\send{\chanloc_\lside}{42}}$,
and the code we wish to verify is (unfolding the definition of $\heapsend$, and performing some pure reductions):
\begin{align*}
& \bind{l_{new}}{\alloc{\none}}{} \store{l}{\some{(l_{new}, 42)}}; l_{new}
\end{align*}

The allocation is easily handled with \ruleref{ml-alloc}, and it turns out we don't even need to remember anything about the returned $l_{new}$.

The next step is the core of this proof: showing that we can change the value stored in $l$.
Notice that we do not own $l \mapsto \_\,$; the STS ``owns'' $l$ as part of its interpretation.
So we will \emph{open} the STS to get access to $l$.

Looking at \figref{fig:example-sts}, we can see that doing the transition from \textsc{init} to \textsc{sent} requires the token $[\textlog{S}]$, \emph{which we own} -- as a consequence, nobody else could perform this transition.
It follows that the STS is currently in state \textsc{init}.
We obtain  $\pred(\textsc{init})$, so that we %
can apply \ruleref{ml-store} with \ruleref{src-send-left}, yielding
\begin{equation}
 l \mapsto \some(l',42) * \chanloc \smapsto ([],[]) * \ownThread{j}{\chanloc_\lside} \label{eq:proofstate-sender}
\end{equation}

To finish up accessing the STS, we have to pick a new state and show that we actually possess the tokens to move to said state.
In our case, we \emph{cannot} pick \textsc{received}, since we do not own the token $[\textlog{R}]$ necessary for that step.
Instead, we will pick $\textsc{sent}$ and give up our token.
This means we have to establish $\pred(\textsc{sent})$.
Doing so consumes most of our context (\ref{eq:proofstate-sender}), leaving only
$\ownThread{j}{\chanloc_\lside}$.
What remains to be done?
We have to establish the post-condition of our triple (\ref{eq:ht-send}), which is $\stopped$.
By \ruleref{src-stopped}, this immediately follows from the fact that we reduced the source thread to $\chanloc_\lside$, which is a value.

Notice that this last step was important:  We showed that when the MiniML thread terminates, so does the source thread.
The original $\textlog{fork}$ rule for Iris allows picking \emph{any} post-condition for the forked-off thread, because nothing happens any more with this thread once it terminates.
However, we wish to establish that if all MiniML threads terminate, then so do all source threads -- and for this reason, \ruleref{ml-fork} forces us to prove $\stopped$, which asserts that all the threads we keep track of have reduced to a value.
This finishes the proof of the sender.

\paragraph{Verifying the receiver.}
The next (and last) step in establishing the refinement (\ref{eq:ht-refine}) is to prove the Hoare triple for the receiving thread (\ref{eq:ht-recv}).
This is the target code to verify:
\[ \letp\_{v}{\heaprecv\,l}{v} \]
Since $\heaprecv$ is a recursive function, we use \ruleref{ml-rec}, which says that we can assume that recursive occurrences of $\heaprecv$ have already been proven correct.
It may be surprising to see this rule~-- after all, rules like \ruleref{ml-rec} are usually justified by saying that all we do is partial correctness.
Notice, however, that we are \emph{not} showing that $\Expr_{\textrm{ex}}$ terminates.
All we show is that, \emph{if} $\Expr_{\textrm{ex}}$ diverges, then so does $\expr_{\textrm{ex}}$.
That is, we are establishing termination-\emph{preservation}, not termination.

In continuing the proof, we thus get to assume correctness of the recursive call.
Our current context is
\begin{equation}
  \STSSt(\textsc{init}, [\textlog{R}]) * \ownThread{j}{\letp\_{v}{\recv{\chanloc_\rside}}{v}}
 \label{eq:proofstate-recv}
\end{equation}
and the code we are verifying is
\begin{align*}
& \match{\load{l}} \;
 \patcase{\none}{\heaprecv\, l} 
 \ALT \patcase{\some{(l',v)}}{(l', v)}
\;\matchend
\end{align*}
with post-condition $\Ret (\any, n). n = 42 *  \ownThread{j}{42}$.

The first command of this program is $\load l$.
To access $l$, we have to again open the STS.
Since we own $[\textlog{R}]$, we can rule out being in state \textsc{received}.
We perform a case distinction over the remaining two states.
\begin{itemize}
\item If we are in \textsc{init}, we get $l \mapsto \none *  \chanloc \smapsto ([],[])$ from the STS's $\pred(\textsc{received})$.
  We use \ruleref{ml-load} with \ruleref{src-recv-right-miss}.
  Notice how we use $\chanloc \smapsto ([],[])$ to justify performing an ``idle'' step in the source.
  This is crucial -- after all, we are potentially looping indefinitely in the target, reading $l$ over and over; we have to exhibit a corresponding diverging execution in the source.

  Since we did not change any state, we close the invariant again in the \textsc{init} state.
  Next, the program executes the $\none$ arm of the match: $\heaprecv\, l$.
  Here, we use our assumption that the recursive call is correct to finish the proof.

\item Otherwise, the current state is \textsc{sent}, and we obtain  $l \mapsto \some(\_, 42) * \chanloc \smapsto ([42],[])$.
  We use \ruleref{ml-load} with \ruleref{src-recv-right-hit}; 
  this time we know that the $\textlog{recv}$ in the source will succeed.
  We also know that we are loading $(\_, 42)$ from $l$.
  We pick \textsc{received} as the next state (giving up our STS token), and trivially establish $\pred(\textsc{received})$.
  We can now throw away ownership of $l$ and $\chanloc$ as well as $\STSSt(\textsc{received})$ since we no longer need them -- we can do this because all these assertions are affine.

  All that remains is the source thread:
  \begin{align*}
  & \ownThread{j}{\letp\_{v}{(\chanloc_\rside,42)}{v}} 
  \end{align*}

  Next, the target program  will execute the $\textlog{some}$ branch of the $\textlog{match}$.
  To finish, we need to justify the post-condition:
  $\Ret (\any, n). n = 42 *  \ownThread{j}{42}$.
  We already established that the second component of the value loaded from $l$ is $42$, and the source thread is easily reduced to $42$ as well.
\end{itemize}
  This finishes the proof of (\ref{eq:ht-recv}) and therefore of (\ref{eq:ht-refine}): we proved that
$ \expr_{\textrm{ex}} \refines \Expr_{\textrm{ex}} $.

\section{Soundness of the Logic}
\label{sec:extensions}

We have seen how to use our logic to establish a refinement
for a particular simple instance of our translation. We now need to show
that this logic is sound.

As already mentioned, our logic is an extension of Iris, so 
we need to adapt the soundness proof of Iris~\cite{iris2}.
The two extensions that were described in~\Sref{sec:iris-overview} are:
\begin{enumerate}

\item We add a notion of a \emph{step shift}, which is used to simulate source program threads.

\item We move from an affine logic to a linear logic. This is needed to capture
the idea that some resources (like $\ownThreadNoArg$) represent
\emph{obligations} that cannot be thrown away.
\end{enumerate}

In this section we describe how we adapt the semantic model of Iris to handle these changes.
Although our extensions sound simple, the modification of the model requires some care. 
Many of the features we used in \Sref{sec:iris}, such as STSs~\cite{iris} and reasoning about the source language, are \emph{derived} constructions that are not ``baked-in'' to the logic.
As we change the model, we need to ensure that all of these features can still be encoded.
We also strive to keep our extensions as general as possible so as to not unnecessarily restrict the flexibility of Iris.

\paragraph{Brief review of the Iris model.}
We start by recalling some aspects of the Iris model~\cite{iris2} that we modify in our extensions. A key concept is the notion of a \emph{resource}.
Resources describe the physical state of the program as well as additional \emph{ghost state} that is added for the purpose of verification and used, \eg to interpret STSs or the assertions talking about source programs.
Resources are instances of a partial commutative monoid-like algebraic structure; in particular, two resources $\melt$, $\meltB$ can be \emph{composed} to $\melt \mtimes \meltB$.
This operation is used to combine resources held by different threads. When the composition $\melt \mtimes \meltB$ is defined, the elements $\melt$ and $\meltB$ are said to be \emph{compatible}.
Iris always ensures that the resources held by different threads are compatible.
This guarantees that, \eg different threads cannot own the same channel or the same STS token.
The operation also gives rise to a pre-order on resources, defined as $\melt_1 \mincl
\melt_2 \eqdef \Exists \melt_3. \melt_1 \mtimes \melt_3 = \melt_2 $, \ie $\melt_1$ is included in $\melt_2$ if the former can be \emph{extended} to the latter by adding some additional resource $\melt_3$.

Ideally, we would just interpret an assertion $\prop$ as a set of resources.
For technical reasons (that we will mostly gloss over), Iris needs an additional component: the \emph{step-index} $n$.
An assertion is thus interpreted as a set of pairs $(n, \melt)$ of step-indices and resources.
We write $n, \melt
\models \prop$ to indicate that $(n, \melt) \in \prop$, and read this
as saying that $\melt$ satisfies $\prop$ for $n$ steps of the
target program's execution.

Iris furthermore demands that assertions (interpreted as sets) satisfy two \emph{closure properties}:
They must be closed under larger resources and smaller step-indices.  Formally:
 \begin{enumerate}
 \item If $n, \melt \models \prop$ and $\melt \mincl \melt'$, then $n, \melt' \models \prop$.
 \item If $n, \melt \models \prop$ and $n' \leq n$, then $n', \melt \models \prop$.
 \end{enumerate}
The first point above makes Iris an \emph{affine} as opposed to a linear logic:
 we can always ``add-on'' more
resources and continue to satisfy an assertion.
Put differently, there is no way to state an \emph{upper bound} on our resources.
The second point says that
if $\prop$ holds for $n$ steps, then it also holds for fewer than $n$ steps.

To give a model to assertions like $\heaploc \mapsto v$, we need a function $\textlog{HeapRes}(\heaploc, v)$ describing, as a resource, a heap which maps location $\heaploc$ to $v$.
We then define:
\begin{alignat*}{3}
n, \melt &\models \heaploc \mapsto v
   &\quad \text{iff} \quad & \textlog{HeapRes}(\heaploc, v) \mincl \melt 
\end{alignat*}
Notice the use of $\mincl$, ensuring that the closure property (1) holds.

\paragraph{Equipping Iris with linear assertions.}
In order to move to a linear setting with minimal disruption to the
existing features of Iris, we replace the judgment
$n, \melt \models \prop$ with $n, \melt, \meltB \models \prop$. That
is, assertions are now sets of triples: a step-index and \emph{two}
resources. The downward closure condition on $n$ and the
upward closure condition on $\melt$ still apply, but we do not impose
such a condition on $\meltB$: %
this second resource will represent the ``linear piece'' of an assertion. Crucially, whereas affine assertions like $\heaploc \mapsto \val$ continue to ``live'' in the $a$ piece, the linear $\ownThreadNoArg$ resides in $\meltB$:
\begin{alignat*}{3}
n, \melt, \meltB &\models \heaploc \mapsto v
   &\quad \text{iff} \quad & \textlog{HeapRes}(\heaploc, v) \mincl \melt \land \meltB = \munit  \\
n, \melt, \meltB &\models \ownThreadD{i}{\Expr}{\delay}
   &\quad \text{iff} \quad & \textlog{SourceRes}(i, \Expr, \delay) = \meltB
\end{alignat*}
where $\munit$ is the unit of the monoid. We assume $\textlog{SourceRes}(i, \Expr, \delay)$ to define, as a resource, a source thread $i$ executing $\Expr$ with $\delay$ delay steps left.

As we can see, $\ownThreadNoArg$ describes the \emph{exact} linear resources $\meltB$ that we own, whereas $\mapsto$ merely states a \emph{lower bound} on the affine resources $\melt$ (due to the upwards closure on $\melt$).
Notice that $\melt$ and $\meltB$ are both elements of the same set of resources; it is just their treatment in the closure properties of assertions which makes one affine and the other linear.
Because there is no upward closure condition on the second monoid
element, the resulting logic is not affine: if $n, \melt, \meltB
\models \prop * \propB$, then it is not necessarily the case that $n,
\melt, \meltB \models \prop$.

We define the affine modality by:
\begin{alignat*}{3}
n, \melt, \meltB &\models \aff{\prop}
   &\quad \text{iff} \quad n, \melt, \meltB \models \prop \wedge  \meltB = \munit
\end{alignat*}
This says that in addition to satisfying $\prop$, $\meltB$ should
equal the unit of the monoid. That is, the linear part is
``empty''; there are no obligations encoded in $\prop$.
That makes it sound to throw away $\prop$ or to frame it.

The advantage of this ``two world'' model is that it does
not require us to change many of the encodings already present in
Iris, like STSs.

\paragraph{Step Shifts.} We are now ready to explain the ideas behind the \emph{step shift}.
Remember the goal here is to account for the steps taken in the source program, in a way that we can prove refinements by proving Hoare triples (\ruleref{Ht-refine}).
This is subtle because by the definition of refinement (\Sref{sec:refinement}), we need to make statements even about infinite executions, \ie executions that never have to satisfy the post-condition.

The key idea is to equip the resources of Iris with a relation that represents a notion of \emph{taking a (resource) step}. We write $\melt \mstep \meltB$, and say that $\melt$ \emph{steps to} $\meltB$.
We will then pick the resources in such a way as to represent the status of a source program,\footnote{Iris is designed to be parametric in the choice of resources, so we can pick a particular resource for this source language and still use most of the general Iris machinery.} and we define the resource step to be taking a step in the source program.
All the other components of the resource, like STSs, will not be changed by resource steps.

Recall that the resources owned by different threads always need to be compatible. To ensure this, we define a relation that performs a step while maintaining compatibility with the resources owned by other threads.
Formally, a \emph{frame-preserving step-update} $\melt, \meltB \mupd \melt', \meltB'$
holds if $\meltB \mstep \meltB'$ and for all $\meltC$ such that $\melt \mtimes \meltB \mtimes \meltC$ is defined, so is $\melt' \mtimes \meltB' \mtimes \meltC$. 
The intuition is that, if a thread owns some resources $\melt$ and $\meltB$, that restricts the ownership of other threads to \emph{frames} $\meltC$ that are compatible with $\melt$ and $\meltB$. %
Since $\melt'$ and $\meltB'$ are also compatible with the frame, the step is guaranteed not to interfere with resources owned by other threads.

These frame-preserving step-updates are reflected into the logic through the \emph{step shift} assertions:
$\prop \svs \propB$ holds if, whenever some resources satisfy $\prop$, it is possible to perform a frame-preserving step-update to resources satisfying $\propB$.

We then connect Hoare triples to these resource steps.
To this end, we change the definition of Hoare triples so that whenever a target thread takes a
step, we have to also take a step on our resources. 
 This gives rise to the proof rules in \figref{fig:hoare-rules}, which force the user of the logic to perform a step shift alongside every step of the MiniML program.
We also enforce that forked-off threads must have a post-condition of $\stopped$, ensuring that target language threads cannot stop executing while source language threads are still running.

\paragraph{Soundness of the refinement.}
Having extended the definition of Hoare triples in this way, we can
prove our refinement theorem. Recall that the
definition of refinement had three parts.
For each of these parts, we proved an adequacy theorem for our extensions relating Hoare triples to properties of program executions.
These theorems are parameterized by the kind of resource picked by the user, and in particular the kind of resource \emph{step}.
Below, we show these theorems specialized to the case where resource steps correspond to source language steps.

The
first refinement condition, which says that the target program must not get stuck, follows
from a ``safety'' theorem  that was already present in the original Iris:
\begin{lem}\label{lem:adequate-safe}
If $\hoare{\ownThreadD{i}{\Expr}{d}}{\expr}{\Ret\val. \ownThreadD{i}{\Val}{0} *
  \aff{\val \valobsrel \Val}}$ holds and we have $([\expr], \initstate)
\step^\ast ([\expr_1, \dots, \expr_n], \state)$, then each $\expr_i$ is
either a value or it can take a step in state $\state$.
\end{lem}

The second refinement condition says that if the execution of $\expr$
terminates, then there should be a related terminating execution in
the source. Remember that the definition of the Hoare triple requires
us to take a step in the source whenever the target steps (modulo a
finite number of delays). Hence a proof of such a triple must have
``built-up'' the desired source execution:

\begin{lem}\label{lem:adequate-finite}
If
$\hoare{\ownThreadD{i}{\Expr}{d}}{\expr}{\Ret\val. \ownThreadD{i}{\Val}{0}
  * \aff{\val \valobsrel \Val}}$ holds and we have $([\expr], \initstate)
\step^\ast ([\val_1, \dots, \val_n], \state)$, then there exists
$\Val_1$, $\Expr_2$, $\dots$, $\Expr_m$, $\State$ {s.t.}\ $([\Expr],
\initstate) \step^\ast ([\Val_1, \Expr_2, \dots, \Expr_m], \State)$. Moreover, each
$\Expr_i$ is either stuck or a value, and $\val_1 \valobsrel \Val_1$.
\end{lem}

Here, we are already making crucial use of both linearity of $\ownThreadNoArg$ and the fact that forked-off threads must have post-condition $\stopped$: if it were not for these requirements, even when all target threads terminated with a value $v_i$, we could not rule out the existence of source threads that can go on executing.

Finally, we come to the third condition, which says fair diverging executions of the target
should correspond to fair diverging executions of the source:

\begin{lem}\label{lem:inf-fair}
If $\hoare{\ownThreadD{i}{\Expr}{d}}{\expr}{\Ret\val. \ownThreadD{i}{\Val}{0} *
  \aff{\val \valobsrel \Val}}$ holds and $([\expr], \initstate)$ has a diverging execution,
then $([\Expr], \initstate)$ has a diverging execution as well. Moreover, if the diverging target execution is fair, then the source execution is too.
\end{lem}

This is the hardest part of the soundness proof. We would like to
start by arguing that, just as for the finite case, if the target program took an
infinite number of steps, then the proof of the refinement triple must
give a corresponding infinite number of steps in the source
program.
Unfortunately, this argument is not so simple because of
step-indexing.  

In Iris, Hoare triples are themselves step-indexed sets.
We write $n \models \hoare{\prop}\expr{\propB}$ to say that the triple holds at step-index $n$. Then, when we say we have proved a Hoare triple, we mean
the triple holds for all step-indices $n$ and all resources satisfying the precondition.
As is usual with step-indexing, when a triple
$\hoare{\prop}\expr{\propB}$ holds for step-index $n$, that means when
the precondition is satisfied, execution of $\expr$ is safe for up-to
$n$ steps, and if it terminates within those $n$ steps, the
post-condition holds. In our case, it also means that each step of the
target program gives a step of the source program, for up to $n$
target steps. 

This restriction to only hold ``up to $n$ steps''
arises due to the way Hoare triples are defined in the model: when
proving the Hoare triple at step-index $n$, if $\expr$ steps to $\expr'$,
we are only required to show $(n-1) \models \hoare{\prop'}{\expr'}{\propB}$ for some $\prop'$.

The restriction to a finite number of steps did not bother us for
\lemref{lem:adequate-safe} and \lemref{lem:adequate-finite}. Since
they only deal with finite executions, and the Hoare triple holds for
\emph{all} starting indices $n$, we can simply pick $n$ to be greater
than the finite execution we are considering.  But we cannot do this
when we want to prove something about a diverging execution of the
target. Whatever $n$ we start with, it is not big enough to get the
infinite source execution we need.

\paragraph{Bounded non-determinism, infinite executions, and step-indexing.}

Our insight is that when the source language has only \emph{bounded
  non-determinism}, we can set up a more careful inductive
argument. By bounded non-determinism, we mean that each configuration
$([\Expr_, \ldots], \State)$ only has \emph{finitely many} possible
successor configurations. The key result is the following quantifier
inversion lemma:

\begin{lem}\label{lem:quantifier-inversion}
  Let $\propC$ be a step-indexed predicate on a finite set $X$. Then:
\[ (\forall n.\,\exists x.\, n \models \propC(x)) \Ra (\exists x.\, \forall n.\, n \models \propC(x)) \]
\end{lem}
\begin{proof}
By assumption, for each $n$, there exists $x_n \in X$ such that $n \models \propC(x_n)$. Since $X$ is finite, by the pigeon-hole principle, there must be some $x \in X$ such that $m \models \propC(x)$ for infinitely many values of $m$. Now, given arbitrary $n$, this means there exists $m > n$ such that $m \models \propC(x)$. Since step-indexed predicates are downward-closed, $n \models \propC(x)$. Hence $\forall n.\, n \models \propC(x)$.
\end{proof}
Ignoring delay steps for the moment, we apply this lemma to our setting to get:

\begin{lem}\label{lem:bootstrap} Suppose $\expr$ steps to $\expr'$ and $\forall n.\,\exists \prop_n.\, n \models \hoare{\ownThread{i}{E} * \prop_n}{\expr}{\propB}$. Then, $\exists \Expr'$ such that $\Expr$ steps to $\Expr'$ and $\forall n.\, \exists \prop_n'.\, n\models \hoare{\ownThread{i}{\Expr'} * \prop_n'}{\expr'}{\propB}$.
\end{lem}

\begin{proof} Let $X$ by the set of $\Expr'$ that $\Expr$ can step to, which we know to be finite.\footnote{To be precise we ought to mention the initial states $\state$ and $\State$ that $\expr$ and $\Expr$ run in and assume they satisfy the precondition of the triple.}
 Consider the step-indexed predicate $\propC$ on $X$ defined by
$n \models \propC(\Expr') \eqdef \left(\Expr \step \Expr' \land \exists \prop'_n.\, n \models \hoare{\ownThread{i}{\Expr'} * \prop'_n}{\expr'}{\propB}\right)$.
By assumption, for each $n > 0$, $n \models \hoare{\ownThread{i}{\Expr} * \prop_n}{\expr}{\propB}$ for some $\prop_n$. The definition of Hoare triples implies that there exists some $\Expr'$ such that $(n - 1) \models \propC(\Expr')$. Thus, $\forall n.\exists\Expr'.\, n\models \propC(\Expr')$, so we can apply \lemref{lem:quantifier-inversion} to get the desired result. \end{proof}

Notice that in the conclusion of \lemref{lem:bootstrap}, if $\expr'$
takes another step, we can apply \lemref{lem:bootstrap} again to the
triples for $\expr'$. So, given some initial triple
$\hoare{\ownThread{i}{\Expr}}{\expr}{Q}$ and a diverging execution of
$\expr$, by induction we can repeatedly apply \lemref{lem:bootstrap}
to construct an infinite execution of the source program. Finally, we
prove that if the execution of $\expr$ was fair, this source execution
will be fair as well, giving us \lemref{lem:inf-fair}.  Of course, for
the full mechanized proof we have to take into account the delay steps
and consider the case where the target thread multiple source
threads.  But all of these are \emph{finite} additional possibilities,
they do not fundamentally change the argument sketched above.

\section{Proof of Compiler Correctness}
\label{sec:session-proof}

We now give a brief overview of our proof of \thmref{thm:compiler}. Recall that we want
to show that if $E$ is a well-typed source expression, then
$\compile{\Expr} \refines \Expr$.

Our proof is a binary logical relations argument. We interpret each
type $\tau$ as a relation on values from the target and
source language, writing $\interpValRel{\val}{\Val}{}{\tau}$ to say
that $\val$ and $\Val$ are related at type $\tau$. 
However, following the example of \citep{iris-effects, iris-proofmode}, these are relations \emph{in our refinement logic},
which means we can use all of the constructs of the logic to describe
the meaning of types. We then prove a fundamental lemma 
showing that well-typed expressions are logically related to
their translation. Next, we show that our logical relation implies the triple used in \ruleref{Ht-refine}.
\thmref{thm:compiler} is then a direct consequence of these two lemmas.%

Details of these proofs can be found in the \appref; here we focus on the definition of the logical relation itself. For most types, the interpretation is straight-forward and fairly standard. For instance, $\interpValRel{\val}{\Val}{}{\inttyp}$ holds
exactly when $\val = \Val = n$, for some integer $n$. The important exception, of course, is the interpretation of session types, in which we need to relate the encoding of channels as linked-lists to the source language's primitive buffers.

\paragraph{Sessions as an STS.} To interpret session types,
we generalize the state transition system from the example in \Sref{sec:iris}
to handle the more complicated ``protocols'' that session types
represent.

What should the states of this STS be? In the STS
used in \secref{sec:iris}, we had three states: \textsc{init}, in
which the message had not been sent; \textsc{sent}, where a message
had been sent from the left end-point, but not received; and
\textsc{received}, where the message had now been received at the
right end-point.  In the general case, we will have more than one
message, so our states need to track
how many messages have been sent/received on each end-point.
We also need to know the ``current'' type of the end-points, but notice that if we know the starting type of an end-point,
and how many messages have been sent/received on it, we can always
recover these current types.
We write $\sadvance{\styp}{n}$ for the type after
$n$ messages have been sent/received starting from $\styp$.

We also need to know which heap locations $\heaploc_\scl$ and
$\heaploc_\scr$ currently represent the end-points of the channel.
All together then, the states will be tuples $(n_\scl, n_\scr, \heaploc_\scl,
\heaploc_\scr)$ describing how many messages have been sent/received on
each end-point, and the corresponding heap locations.

Remember that we also need to define the tokens and transitions
associated with each state of our STS.
The transitions are simple: we can either advance
the left end-point, incrementing $n_\scl$ and updating $\heaploc_\scl$, and similarly for the right
end-point.
For the tokens, recall that in our example proof, we had $[\textlog{S}]$ and $[\textlog{R}]$ tokens used by each thread to advance the state when they had interacted with their respective
end-points. 
In general, the threads will now use the end-points multiple times,
so we need a token for each of these uses on both sides.
Concretely, we will have two kinds of tokens, $\leftTok{n}$ and
$\rightTok{n}$, which are used when advancing the left and right end-point counter to $n$, respectively.

To complete the description of the STS, we have to talk about the interpretation of the states.
This interpretation has to relate the messages in the
source channel's current buffers to the nodes in the linked list on the target heap.
The individual messages should, of course, be related by our logical relation ($\interpRelNoArgs{V}{}$). We lift this relation to lists of messages ($\interpRelNoArgs{L}{}$) as follows:
\begin{mathpar}
\infer{}{\interpListRel{[]}{[]}{}{\styp}}

\inferH{L-cons}{\later{\left(\interpValRel{\val}{\Val}{}{\typ}\right)} *
       \left(\interpListRel{L_\sch}{L_\scc}{}{\styp}\right)}
      {\interpListRel{\val L_\sch }{\Val L_\scc}{}{\recvtyp{\typ}{\styp}}}
\label{eq:listrel}
\end{mathpar}
For now, ignore the $\later{}$ symbol. The left rule
says that two empty lists are equivalent at any session type.  The
right rule says two lists are related \emph{at a receive type}
$\recvtyp{\typ}{\styp}$, if their heads are related under $\typ$, and
the remainders of each list are related at $\styp$. It is important
that this is a receive type: if the current type of the end-point is a
send type, then there should not be any messages in its receive
buffer, so the rule for empty lists is the only one that applies.

We can now give our state interpretation, $\sessionInvNoArg$, which
is parameterized by 
(a) the starting type $\styp$ of the left end-point (the right
end-point's starting type is by necessity dual so there is no need to track it),
and (b) the name $c$ of the channel:
\begin{align}
\sessionInv{\styp}{\chanloc}{n_\scl, n_\scr, \heaploc_\scl, \heaploc_\scr} \eqdef 
 \exists L_{\scc}, L_{\sch}.\, 
\quad& \Big(\chanloc \smapsto (L_\scc, []) * \linklist{L_\sch}{\heaploc_\scl}{\heaploc_\scr} * {}
   \label{eq:states} \\ 
& ~~~(\interpListRel{L_\sch}{L_\scc}{}{S^{n_{\scl}}}) * n_{\scl} + \llength{L_\scc} = n_\scr \Big) \lor \dots
   \label{eq:related}  
  \end{align}

Let us explain this piece by piece.  To start, we have that there
exists a list of \emph{source} values $L_\scc$ and a list of \emph{target} values $L_\sch$, representing the messages that are stored in the buffer right now.  We then distinguish between two
cases: either the first buffer is empty or the second buffer is
empty.  We omit the second case (corresponding to the second
disjunct) because it is symmetric.  In the first case, the
channel's first buffer contains $L_\scc$ and the second buffer is
empty (\ref{eq:states}, left).
On the target side, the buffer is represented as a linked list from $\heaploc_\scl$ to $\heaploc_\scr$ containing the values $L_\sch$ (\ref{eq:states}, right).
Of course, the lists of values need to be related according to the end-point's current type $S^{n_\scl}$ (\ref{eq:related}, left).  Finally, 
the number of messages sent/received through the
left end-point, plus the number of messages still in the buffer,
should equal the total number of messages sent/received through the
right end-point (\ref{eq:related}, right).
Therefore,
when these remaining messages are received by the left end-point,
the two types will again be dual.

Informally then, the value relation at session types $\interpValRel{\heaploc}{\locsidevar}{}{\styp}$ says that there exists an appropriate STS and tokens for the session $S$ which relates $\heaploc$ and $\locsidevar$.
We can then prove Hoare triples for the message-passing primitives that manipulate this STS. For instance, for $\heaprecv$ we have (omitting delay steps):%
\begin{align*}
&\hoareHV{ \ownThread{i}{\lctx[\recv{\locside{\chanloc}{\side}}]}
          * \interpValRel{\heaploc}{\locsidevar}{}{\recvtyp{\tau}{\styp}}}
       {\quad~~\heaprecv\ \heaploc}
       {\Ret(\heaploc', \val). \Exists \Val.
    \ownThread{i}{\lctx[(\locside{\chanloc}{\side}, \Val)]} *
    (\interpValRel{\val}{\Val}{}{\tau}) * %
    \interpValRel{\heaploc'}{\locsidevar}{}{\styp}
}
\end{align*}

This triple closely corresponds to the typing rule \ruleref{Recv}
(\figref{fig:session-lang-rules}): typing judgments in the premise
become value relations in the pre-condition, and the conclusion is
analogously transformed into the postconditon. Indeed, the proof of
the fundamental lemma for the logical relation essentially just appeals to these triples.

There is something we have glossed over: when we defined the logical relation, we used the STS, but the STS interpretation used the logical relation!
This circularity is the reason for the $\later{}$ symbol guarding the recursive occurrence of ($\interpRelNoArgs{V}{}$) in \ruleref{L-cons}. The details are spelled out in the appendix.

\section{Conclusion and Related Work}
\label{sec:conclusion}

We have presented a logic for establishing \emph{fair, termination-preserving} refinement of higher-order, concurrent languages.
To our knowledge, this is the first logic combining higher-order reasoning (and in particular, step-indexing) with reasoning for termination-sensitive concurrent refinement.
Moreover, we applied this logic to verify the correctness of a compiler that translates a session-typed source language with channels into an ML-like language with a shared heap.

All of these results have been fully mechanized in Coq. Our mechanization builds on the Coq development described in \citet{iris2} and the proof-mode from \citet{iris-proofmode}. The proofs use the axioms of excluded middle and indefinite description. The proof scripts can be found online~\citep{supplement-coq}.

\paragraph{Second Case Study.} 

Our logic is not tied to this source language and translation: we have used it to mechanize
a proof that the
Craig-Landin-Hagersten queue lock~\citep{Craig93, MagnussonLH94}
refines a ticket lock. Further details  can be found in the \appref.

\paragraph{Linearity.} Linearity has been used in separation logics to verify the absence of memory leaks: if heap assertions like $l \mapsto v$ are linear, and the only way to ``dispose'' of them is by freeing the location $l$, then post conditions must mention all memory that persists after a command completes~\citep{IshtiaqO01}. Our treatment of linearity has limitations that make it unsuitable for tracking resources like the heap. First, in our logic, only affine assertions can be framed (see \ruleref{Ht-frame}), because framing could hide the obligation to perform steps on source threads. Of course, for resources like the heap this would be irrelevant, and this rule could be generalized.
Second, linear resources cannot be put in STS interpretations, so they cannot be shared between threads. Since STSs are implemented in terms of a more primitive feature in Iris called \emph{invariants}, which are affine, allowing linear resources to be put inside would circumvent the precise accounting that motivates linearity in the first place. Thus, we would need to extend Iris with a useful form of ``linear'' shared invariants, which we leave to future work.

\paragraph{Session Types.}
Starting from the seminal work of \citet{honda93}, a number of session-type systems have been presented with
different features~\cite{yoshida07,GayV10,CairesP10,ToninhoCP13,Wadler14} (among many others).  The language presented here is a simplified version of
the one in \citet{GayV10}. 
\citet{Wadler14} has shown that a restricted
subset of the language in \citep{GayV10} does enjoy a deadlock freedom
property. This property holds only when the type system is
\emph{linear}, like the original in \citep{GayV10}. 
\citet{PerezCPT12} and \citet{CairesPPT13} give logical relations for
session-typed languages, which they use to prove strong normalization
and contextual equivalence results.  Their logical relation is defined
``directly'', instead of translating into an
intermediary logic.
Early versions of another
session-typed system~\cite{Willsey16} used a ring-buffer to represent channels instead of linked 
lists, which would be interesting to verify. 

\paragraph{Logics for Concurrency, Termination, and Refinement.}

There is a vast literature on program logics for
concurrency~\cite{ohearn:csl,Brookes06a,caresl,iris,iris2,cap,icap,tada,total-tada,fcsl,views,rgsep,lrg,liang:refinement14,liang:termination14,liang:lili,hoffmann:lock-freedom}.
Indeed, the reason for constructing a logical relation on top of a program logic, as in \citet{iris-effects}, is so that we can
take advantage of the many ideas that have proliferated in this community.

Focusing on logics for refinement and termination
properties: \citet{benton:popl04} pioneered the use of a
\emph{relational} Hoare logic for showing the correctness of compiler
transformations in the sequential setting. \citet{yang:relational}
generalized this to relational separation logic. We have already
described \cite{caresl}, which
developed a higher-order concurrent separation logic for
termination-insensitive refinement. 
\citet{liang:refinement14} also allow non-terminating
programs to refine terminating ones. This was extended
in~\cite{liang:termination14} for a termination-preserving refinement,
but this deals with termination-preservation \emph{without} fairness.
Most recently \citet{liang:lili} addressed fair termination-preserving
refinement. In their logic,
threads can explicitly reason about how their actions may or
may not further delay other threads, which is more general than our approach and may be needed for verifying
some of the examples they consider.
It would be interesting to
adapt this more explicit fairness reasoning to the higher-order setting.

\citet{hoffmann:lock-freedom} features a concurrent separation logic
for total correctness. Threads own resources called ``tokens'', which
must be ``used up'' every time a thread repeats a while loop. 
This ``using up'' of
tokens inspired our step shifts. Later, \citet{total-tada}
generalized this by using ordinals instead of tokens: threads decrease the
ordinal they own as they repeat a loop. This is useful for
languages with unbounded non-determinism.
Our technique for coping with step-indexing in
\Sref{sec:extensions} relied on bounded non-determinism.
It may be possible to remove this limitation by using
\emph{transfinite} step-indexing~\citep{BirkedalBS13,SvendsenSB16} instead.%

\paragraph{Acknowledgments.} 
The authors thank Robbert Krebbers, Jeehoon Kang, Max Willsey, Frank
Pfenning, Derek Dreyer, Lars Birkedal, and Jan Hoffmann for helpful
discussions and feedback.  This research was conducted with
U.S. Government support under and awarded by DoD, Air Force Office of
Scientific Research, National Defense Science and Engineering Graduate
(NDSEG) Fellowship, 32 CFR 168a; and with support
by a European Research Council (ERC) Consolidator Grant for the
project ``RustBelt'', funded under the European Union's Horizon 2020
Framework Programme (grant agreement no.\ 683289). Any opinions,
findings and conclusions or recommendations expressed in this material
are those of the authors and do not necessarily reflect the views of
these funding agencies.

\bibliographystyle{splncsnat}
\bibliography{bib}

\ifdefined\INCLUDEAPPENDIX
\clearpage
\appendix
\newcommand{\myauthcount}[1]{}
{
\let\authcount\myauthcount \begingroup
\startcontents
\printcontents{}{1}{\textbf{Contents of Appendices}\vskip3pt\hrule\vskip5pt}
\vskip3pt\hrule\vskip5pt
\section{Extensions to Iris 2.0}

This section is a reproduction of the manual for the Iris 2.0 logic
(taken from a revised version of the technical appendix of
\citet{iris2}), with the modifications needed for our extension. This
reproduction is done with permission of the authors of
\citet{iris2}. Our modifications are highlighted in blue, \isnew{like
  so}. Hence, this section is best viewed in color. If a section or paragraph
heading is highlighted in blue, everything in that section is new.

\emph{Disclaimer:} The Coq development of the logic is taken to be
authoritative. Any discrepancies between this document and the Coq
code should therefore be regarded as errors. Some important rules
developed in the Coq proof may have been left out. Moreover, any
mistakes in this document may have been introduced by the present
authors: readers interested in the original Iris should consult its
documentation. Also, the reader should be aware that the latest
version of Iris has departed significantly from the version with which
we began our extensions.

{
\let\section\subsection
\let\subsection\subsubsection\begingroup
\section{Algebraic Structures}

\subsection{COFE}

The model of Iris lives in the category of \emph{Complete Ordered Families of Equivalences} (COFEs).
This definition varies slightly from the original one in~\cite{catlogic}.

\begin{defn}[Chain]
  Given some set $\cofe$ and an indexed family $({\nequiv{n}} \subseteq \cofe \times \cofe)_{n \in \mathbb{N}}$ of equivalence relations, a \emph{chain} is a function $c : \mathbb{N} \to \cofe$ such that $\All n, m. n \leq m \Ra c (m) \nequiv{n} c (n)$.
\end{defn}

\begin{defn}
  A \emph{complete ordered family of equivalences} (COFE) is a tuple $(\cofe, ({\nequiv{n}} \subseteq \cofe \times \cofe)_{n \in \mathbb{N}}, \lim : \chain(\cofe) \to \cofe)$ satisfying
  \begin{align*}
    \All n. (\nequiv{n}) ~& \text{is an equivalence relation} \tagH{cofe-equiv} \\
    \All n, m.& n \geq m \Ra (\nequiv{n}) \subseteq (\nequiv{m}) \tagH{cofe-mono} \\
    \All x, y.& x = y \Lra (\All n. x \nequiv{n} y) \tagH{cofe-limit} \\
    \All n, c.& \lim(c) \nequiv{n} c(n) \tagH{cofe-compl}
  \end{align*}
\end{defn}

The key intuition behind COFEs is that elements $x$ and $y$ are $n$-equivalent, notation $x \nequiv{n} y$, if they are \emph{equivalent for $n$ steps of computation}, \ie if they cannot be distinguished by a program running for no more than $n$ steps.
In other words, as $n$ increases, $\nequiv{n}$ becomes more and more refined (\ruleref{cofe-mono})---and in the limit, it agrees with plain equality (\ruleref{cofe-limit}).
In order to solve the recursive domain equation in \Sref{sec:model} it is also essential that COFEs are \emph{complete}, \ie that any chain has a limit (\ruleref{cofe-compl}).

\begin{defn}
  An element $x \in \cofe$ of a COFE is called \emph{discrete} if
  \[ \All y \in \cofe. x \nequiv{0} y \Ra x = y\]
  A COFE $A$ is called \emph{discrete} if all its elements are discrete.
  For a set $X$, we write $\Delta X$ for the discrete COFE with $x \nequiv{n} x' \eqdef x = x'$

\end{defn}

\begin{defn}
  A function $f : \cofe \to \cofeB$ between two COFEs is \emph{non-expansive} (written $f : \cofe \nfn \cofeB$) if
  \[\All n, x \in \cofe, y \in \cofe. x \nequiv{n} y \Ra f(x) \nequiv{n} f(y) \]
  It is \emph{contractive} if
  \[ \All n, x \in \cofe, y \in \cofe. (\All m < n. x \nequiv{m} y) \Ra f(x) \nequiv{n} f(y) \]
\end{defn}
Intuitively, applying a non-expansive function to some data will not suddenly introduce differences between seemingly equal data.
Elements that cannot be distinguished by programs within $n$ steps remain indistinguishable after applying $f$.
The reason that contractive functions are interesting is that for every contractive $f : \cofe \to \cofe$ with $\cofe$ inhabited, there exists a \emph{unique} fixed-point $\fix(f)$ such that $\fix(f) = f(\fix(f))$.

\begin{defn}
  The category $\COFEs$ consists of COFEs as objects, and non-expansive functions as arrows.
\end{defn}

Note that $\COFEs$ is cartesian closed. In particular:
\begin{defn}
  Given two COFEs $\cofe$ and $\cofeB$, the set of non-expansive functions $\set{f : \cofe \nfn \cofeB}$ is itself a COFE with
  \begin{align*}
    f \nequiv{n} g \eqdef{}& \All x \in \cofe. f(x) \nequiv{n} g(x)
  \end{align*}
\end{defn}

\begin{defn}
  A (bi)functor $F : \COFEs \to \COFEs$ is called \emph{locally non-expansive} if its action $F_1$ on arrows is itself a non-expansive map.
  Similarly, $F$ is called \emph{locally contractive} if $F_1$ is a contractive map.
\end{defn}
The function space $(-) \nfn (-)$ is a locally non-expansive bifunctor.
Note that the composition of non-expansive (bi)functors is non-expansive, and the composition of a non-expansive and a contractive (bi)functor is contractive.
The reason contractive (bi)functors are interesting is that by America and Rutten's theorem~\cite{America-Rutten:JCSS89,birkedal:metric-space}, they have a unique\footnote{Uniqueness is not proven in Coq.} fixed-point.

\subsection{RA}

\begin{defn}
  A \emph{resource algebra} (RA) is a tuple \\
  $(\monoid, \mval \subseteq \monoid, \mcore{{-}}:
  \monoid \to \maybe\monoid, (\mtimes) : \monoid \times \monoid \to \monoid, \isnew{(\mstep) \subseteq \monoid \times \monoid})$ satisfying:
  \begin{align*}
    \All \melt, \meltB, \meltC.& (\melt \mtimes \meltB) \mtimes \meltC = \melt \mtimes (\meltB \mtimes \meltC) \tagH{ra-assoc} \\
    \All \melt, \meltB.& \melt \mtimes \meltB = \meltB \mtimes \melt \tagH{ra-comm} \\
    \All \melt.& \mcore\melt \in \monoid \Ra \mcore\melt \mtimes \melt = \melt \tagH{ra-core-id} \\
    \All \melt.& \mcore\melt \in \monoid \Ra \mcore{\mcore\melt} = \mcore\melt \tagH{ra-core-idem} \\
    \All \melt, \meltB.& \mcore\melt \in \monoid \land \melt \mincl \meltB \Ra \mcore\meltB \in \monoid \land \mcore\melt \mincl \mcore\meltB \tagH{ra-core-mono} \\
    \isnew{\All \melt, \meltB.}& \isnew{(\melt\mtimes\meltB) \in \mval \wedge \mcore\melt \in \monoid%
                         \wedge \mcore\meltB \in \monoid \Ra} \\ &\qquad
                   \isnew{\mcore{\mcore\melt \mtimes \mcore\meltB} = \mcore\melt \mtimes \mcore\meltB}
                          \tagH{ra-core-distrib} \\
    \All \melt, \meltB.& (\melt \mtimes \meltB) \in \mval \Ra \melt \in \mval \tagH{ra-valid-op} \\
    \text{where}\qquad %
    \maybe\monoid \eqdef{}& \monoid \uplus \set{\mnocore} \qquad\qquad\qquad \melt^? \mtimes \mnocore \eqdef \mnocore \mtimes \melt^? \eqdef \melt^? \\
    \melt \mincl \meltB \eqdef{}& \Exists \meltC \in \monoid. \meltB = \melt \mtimes \meltC \tagH{ra-incl}
  \end{align*}
\end{defn}
\noindent
RAs are closely related to \emph{Partial Commutative Monoids} (PCMs), with two key differences:
\begin{enumerate}
\item The composition operation on RAs is total (as opposed to the partial composition operation of a PCM), but there is a specific subset $\mval$ of \emph{valid} elements that is compatible with the composition operation (\ruleref{ra-valid-op}).

This take on partiality is necessary when defining the structure of \emph{higher-order} ghost state, CMRAs, in the next subsection.

\item Instead of a single unit that is an identity to every element, we allow
for an arbitrary number of units, via a function $\mcore{{-}}$ assigning to an element $\melt$ its \emph{(duplicable) core} $\mcore\melt$, as demanded by \ruleref{ra-core-id}.
  We further demand that $\mcore{{-}}$ is idempotent (\ruleref{ra-core-idem}) and monotone (\ruleref{ra-core-mono}) with respect to the \emph{extension order}, defined similarly to that for PCMs (\ruleref{ra-incl}).

  Notice that the domain of the core is $\maybe\monoid$, a set that adds a dummy element $\mnocore$ to $\monoid$.
  Thus, the core can be \emph{partial}: not all elements need to have a unit.
  We use the metavariable $\maybe\melt$ to indicate elements of  $\maybe\monoid$.
  We also lift the composition $(\mtimes)$ to $\maybe\monoid$.
  Partial cores help us to build interesting composite RAs from smaller primitives.

Notice also that the core of an RA is a strict generalization of the unit that any PCM must provide, since $\mcore{{-}}$ can always be picked as a constant function.

\item \isnew{We add an aditional relation $\mstep$ that captures
  ``taking a step'' on a resource. In most cases this will be the
  full relation, as we have no useful notion of ``stepping'' such
  resources.But, for resources used for establishing refinements, this
  will correspond to taking some kind of step in the source program.}

\end{enumerate}

\begin{defn}
  It is possible to do a \emph{frame-preserving update} from $\melt \in \monoid$ to $\meltsB \subseteq \monoid$, written $\melt \mupd \meltsB$, if
  \[ \All \maybe{\melt_\f} \in \maybe\monoid. \melt \mtimes \maybe{\melt_\f} \in \mval \Ra \Exists \meltB \in \meltsB. \meltB \mtimes \maybe{\melt_\f} \in \mval \]

  We further define $\melt \mupd \meltB \eqdef \melt \mupd \set\meltB$.
\end{defn}
The assertion $\melt \mupd \meltsB$ says that every element $\maybe{\melt_\f}$ compatible with $\melt$ (we also call such elements \emph{frames}), must also be compatible with some $\meltB \in \meltsB$.
Notice that $\maybe{\melt_\f}$ could be $\mnocore$, so the frame-preserving update can also be applied to elements that have \emph{no} frame.
Intuitively, this means that whatever assumptions the rest of the program is making about the state of $\gname$, if these assumptions are compatible with $\melt$, then updating to $\meltB$ will not invalidate any of these assumptions.
Since Iris ensures that the global ghost state is valid, this means that we can soundly update the ghost state from $\melt$ to a non-deterministically picked $\meltB \in \meltsB$.

\isnew{\begin{defn}
It is possible to do a \emph{frame-preserving step update} from $\meltB \in \monoid$  with $\melt \in \monoid$ to $\meltsB \subseteq \monoid \times \monoid$, written $\melt, \meltB \msupd \meltsB$, if
  \[ \All \maybe{\melt_\f} \in \maybe\monoid. \melt \mtimes \meltB \mtimes \maybe{\melt_\f} \in \mval \Ra \Exists (\melt', \meltB') \in \meltsB. \melt' \mtimes \meltB' \mtimes \maybe{\melt_\f} \in \mval \wedge \meltB \mstep \meltB'  \]

  We further define $\melt, \meltB \msupd \melt', \meltB' \eqdef \melt, \meltB \msupd \set{(\melt', \meltB')}$.
\end{defn}}

\isnew{We can regard this as a transformation of two compatible resources $\melt$ and $\meltB$ in which we do a frame-preserving update on the first component $a$ and a step on the second component $b$.}

\subsection{CMRA}

\begin{defn}
  A \emph{CMRA} is a tuple $(\monoid : \COFEs, (\mval_n \subseteq \monoid)_{n \in \mathbb{N}},\\ \mcore{{-}}: \monoid \nfn \maybe\monoid, (\mtimes) : \monoid \times \monoid \nfn \monoid, \isnew{((\mstepN{n}) \subseteq \monoid \times \monoid)_{n \in \mathbb{N}}})$ satisfying:
  \begin{align*}
    \All n, \melt, \meltB.& \melt \nequiv{n} \meltB \land \melt\in\mval_n \Ra \meltB\in\mval_n \tagH{cmra-valid-ne} \\
    \All n, m.& n \geq m \Ra \mval_n \subseteq \mval_m \tagH{cmra-valid-mono} \\
    \isnew{\All n, m.}& \isnew{n \geq m \Ra \mstepN{n} \subseteq \mstepN{m} \tagH{cmra-step-mono}} \\
    \All \melt, \meltB, \meltC.& (\melt \mtimes \meltB) \mtimes \meltC = \melt \mtimes (\meltB \mtimes \meltC) \tagH{cmra-assoc} \\
    \All \melt, \meltB.& \melt \mtimes \meltB = \meltB \mtimes \melt \tagH{cmra-comm} \\
    \All \melt.& \mcore\melt \in \monoid \Ra \mcore\melt \mtimes \melt = \melt \tagH{cmra-core-id} \\
    \All \melt.& \mcore\melt \in \monoid \Ra \mcore{\mcore\melt} = \mcore\melt \tagH{cmra-core-idem} \\
    \All \melt, \meltB.& \mcore\melt \in \monoid \land \melt \mincl \meltB \Ra \mcore\meltB \in \monoid \land \mcore\melt \mincl \mcore\meltB \tagH{cmra-core-mono} \\
    \isnew{\All n, \melt, \meltB.}& \isnew{(\melt\mtimes\meltB) \in \mval_n
                                            \wedge \mcore\melt \in \monoid%
                                            \wedge \mcore\meltB \in \monoid \Ra} \\ &\qquad
    \isnew{\mcore{\mcore\melt \mtimes \mcore\meltB} \nequiv{n}
      \mcore\melt \mtimes \mcore\meltB}
                          \tagH{cmra-core-distrib} \\
    \All n, \melt, \meltB.& (\melt \mtimes \meltB) \in \mval_n \Ra \melt \in \mval_n \tagH{cmra-valid-op} \\
    \All n, \melt, \meltB_1, \meltB_2.& \omit\rlap{$\melt \in \mval_n \land \melt \nequiv{n} \meltB_1 \mtimes \meltB_2 \Ra {}$} \\
    &\Exists \meltC_1, \meltC_2. \melt = \meltC_1 \mtimes \meltC_2 \land \meltC_1 \nequiv{n} \meltB_1 \land \meltC_2 \nequiv{n} \meltB_2 \tagH{cmra-extend} \\
    \text{where}\qquad\qquad\\
    \melt \mincl \meltB \eqdef{}& \Exists \meltC. \meltB = \melt \mtimes \meltC \tagH{cmra-incl} \\
    \melt \mincl[n] \meltB \eqdef{}& \Exists \meltC. \meltB \nequiv{n} \melt \mtimes \meltC \tagH{cmra-inclN}
  \end{align*}
\end{defn}

This is a natural generalization of RAs over COFEs.
All operations have to be non-expansive, and the validity predicate $\mval$ can now also depend on the step-index.
We define the plain $\mval$ as the ``limit'' of the $\mval_n$:
\[ \mval \eqdef \bigcap_{n \in \mathbb{N}} \mval_n \]

\paragraph{The extension axiom (\ruleref{cmra-extend}).}
Notice that the existential quantification in this axiom is \emph{constructive}, \ie it is a sigma type in Coq.
The purpose of this axiom is to compute $\melt_1$, $\melt_2$ completing the following square:

\begin{center}
\begin{tikzpicture}[every edge/.style={draw=none}]
  \node (a) at (0, 0) {$\melt$};
  \node (b) at (1.7, 0) {$\meltB$};
  \node (b12) at (1.7, -1) {$\meltB_1 \mtimes \meltB_2$};
  \node (a12) at (0, -1) {$\melt_1 \mtimes \melt_2$};

  \path (a) edge node {$\nequiv{n}$} (b);
  \path (a12) edge node {$\nequiv{n}$} (b12);
  \path (a) edge node [rotate=90] {$=$} (a12);
  \path (b) edge node [rotate=90] {$=$} (b12);
\end{tikzpicture}\end{center}
where the $n$-equivalence at the bottom is meant to apply to the pairs of elements, \ie we demand $\melt_1 \nequiv{n} \meltB_1$ and $\melt_2 \nequiv{n} \meltB_2$.
In other words, extension carries the decomposition of $\meltB$ into $\meltB_1$ and $\meltB_2$ over the $n$-equivalence of $\melt$ and $\meltB$, and yields a corresponding decomposition of $\melt$ into $\melt_1$ and $\melt_2$.
This operation is needed to prove that $\later$ commutes with separating conjunction:
\begin{mathpar}
  \axiom{\later (\prop * \propB) \Lra \later\prop * \later\propB}
\end{mathpar}

\begin{defn}
  An element $\munit$ of a CMRA $\monoid$ is called the \emph{unit} of $\monoid$ if it satisfies the following conditions:
  \begin{enumerate}[itemsep=0pt]
  \item $\munit$ is valid: \\ $\All n. \munit \in \mval_n$
  \item $\munit$ is a left-identity of the operation: \\
    $\All \melt \in M. \munit \mtimes \melt = \melt$
  \item $\munit$ is a discrete COFE element
  \item $\munit$ is its own core: \\ $\mcore\munit = \munit$
  \end{enumerate}
\end{defn}

\begin{lem}\label{lem:cmra-unit-total-core}
  If $\monoid$ has a unit $\munit$, then the core $\mcore{{-}}$ is total, \ie $\All\melt. \mcore\melt \in \monoid$.
\end{lem}

\begin{defn}
  It is possible to do a \emph{frame-preserving update} from $\melt \in \monoid$ to $\meltsB \subseteq \monoid$, written $\melt \mupd \meltsB$, if
  \[ \All n, \maybe{\melt_\f}. \melt \mtimes \maybe{\melt_\f} \in \mval_n \Ra \Exists \meltB \in \meltsB. \meltB \mtimes \maybe{\melt_\f} \in \mval_n \]

  We further define $\melt \mupd \meltB \eqdef \melt \mupd \set\meltB$.
\end{defn}

\isnew{
\begin{defn}
  It is possible to do a \emph{frame-preserving step update} from $\meltB  \in \monoid$ with $\melt \in \monoid$ to $\meltsB \subseteq \monoid \times \monoid$, written $\melt \msupd \meltsB$, if
  \[ \All n, \maybe{\melt_\f}. \melt \mtimes \meltB \mtimes \maybe{\melt_\f} \in \mval_n \Ra \Exists (\melt', \meltB') \in \meltsB. \melt' \mtimes \meltB' \mtimes \maybe{\melt_\f} \in \mval_n \wedge \meltB \mstepN{n} \meltB' \]

  We further define $\melt, \meltB \msupd \melt', \meltB' \eqdef
  \melt \msupd \set{(\melt', \meltB')}$.
\end{defn}
}

Note that for RAs, this and the RA-based definition of a frame-preserving update \isnew{and frame-preserving step update} coincide.

\begin{defn}
  A CMRA $\monoid$ is \emph{discrete} if it satisfies the following conditions:
  \begin{enumerate}[itemsep=0pt]
  \item $\monoid$ is a discrete COFE
  \item $\mval$ ignores the step-index: \\
    $\All \melt \in \monoid. \melt \in \mval_0 \Ra \All n, \melt \in \mval_n$
  \end{enumerate}
\end{defn}
Note that every RA is a discrete CMRA, by picking the discrete COFE for the equivalence relation.
Furthermore, discrete CMRAs can be turned into RAs by ignoring their COFE structure, as well as the step-index of $\mval$.

\begin{defn}
  A function $f : \monoid_1 \to \monoid_2$ between two CMRAs is \emph{monotone} (written $f : \monoid_1 \monra \monoid_2$) if it satisfies the following conditions:
  \begin{enumerate}[itemsep=0pt]
  \item $f$ is non-expansive
  \item $f$ preserves validity: \\
    $\All n, \melt \in \monoid_1. \melt \in \mval_n \Ra f(\melt) \in \mval_n$
  \item $f$ preserves CMRA inclusion:\\
    $\All \melt \in \monoid_1, \meltB \in \monoid_1. \melt \mincl \meltB \Ra f(\melt) \mincl f(\meltB)$
  \end{enumerate}
\end{defn}

\begin{defn}
  The category $\CMRAs$ consists of CMRAs as objects, and monotone functions as arrows.
\end{defn}
Note that every object/arrow in $\CMRAs$ is also an object/arrow of $\COFEs$.
The notion of a locally non-expansive (or contractive) bifunctor naturally generalizes to bifunctors between these categories.
 \endgroup\begingroup
\section{COFE constructions}

\subsection{Next (type-level later)}

Given a COFE $\cofe$, we define $\latert\cofe$ as follows (using a datatype-like notation to define the type):
\begin{align*}
  \latert\cofe \eqdef{}& \latertinj(x:\cofe) \\
  \latertinj(x) \nequiv{n} \latertinj(y) \eqdef{}& n = 0 \lor x \nequiv{n-1} y
\end{align*}
Note that in the definition of the carrier $\latert\cofe$, $\latertinj$ is a constructor (like the constructors in Coq), \ie this is short for $\setComp{\latertinj(x)}{x \in \cofe}$.

$\latert(-)$ is a locally \emph{contractive} functor from $\COFEs$ to $\COFEs$.

\subsection{Uniform Predicates}

Given a CMRA $\monoid$, we define the COFE $\UPred(\monoid)$ of \emph{uniform predicates} over $\monoid$ as follows:
\begin{align*}
  \UPred(\monoid) \eqdef{} \setComp{\pred: \mathbb{N} \times \monoid \isnew{\times \monoid}
    \to \mProp}{
  \begin{inbox}[c]
    \isnew{(\All n, x, y, x', y'. \pred(n, x, y) \land x \mincl x' \land y \nequiv{n} y'
       \Ra \pred(n, x', \isnew{y'})) \land} {}\\
    \isnew{(\All n, m, x, y. \pred(n, x, y) \land x \in \mval_m \wedge y \in \mval_m \Ra \pred(m, x, y))}
  \end{inbox}
}
\end{align*}

One way to understand this definition is to re-write it a little.
We start by defining the COFE of \emph{step-indexed propositions}: For every step-index, the proposition either holds or does not hold.
\begin{align*}
  \SProp \eqdef{}& \psetdown{\mathbb{N}} \\
    \eqdef{}& \setComp{X \in \pset{\mathbb{N}}}{ \All n, m. n \geq m \Ra n \in X \Ra m \in X } \\
  X \nequiv{n} Y \eqdef{}& \All m \leq n. m \in X \Lra m \in Y
\end{align*}
Notice that this notion of $\SProp$ is already hidden in the validity predicate $\mval_n$ of a CMRA:
We could equivalently require every CMRA to define $\mval_{-}(-) : \monoid \nfn \SProp$, replacing \ruleref{cmra-valid-ne} and \ruleref{cmra-valid-mono}.

\isnew{Now we can rewrite $\UPred(\monoid)$ as step-indexed predicates over pairs of $\monoid$ eleements, which is ``monotone'' in a certain sense with respect to the first element:}
\begin{align*}
  \UPred(\monoid) \cong{}& \monoid \monra \isnew{\monoid \ra} \SProp \\
     \eqdef{}& \setComp{\pred: \monoid \nfn \isnew{\monoid \nfn} \SProp}{\begin{aligned}%
\All n, m, x, y, x'.& n \in \pred(x) \land x \mincl x' \land m \leq n \land x' \in \mval_m%
\\& \quad {} \Ra m \in \pred(x', y)\end{aligned}}
\end{align*}

\section{RA and CMRA constructions}

\isnew{When describing a CMRA construction, unless specified otherwise, the step relation is taken to be the full relation.}

\subsection{Product}
\label{sec:prodm}

Given a family $(M_i)_{i \in I}$ of CMRAs ($I$ finite), we construct a CMRA for the product $\prod_{i \in I} M_i$ by lifting everything pointwise.

Frame-preserving updates on the $M_i$ lift to the product:
\begin{mathpar}
  \inferH{prod-update}
  {\melt \mupd_{M_i} \meltsB}
  {f[i \mapsto \melt] \mupd \setComp{ f[i \mapsto \meltB]}{\meltB \in \meltsB}}
\end{mathpar}

\subsection{Finite partial function}
\label{sec:fpfnm}

Given some infinite countable $K$ and some CMRA $\monoid$, the set of finite partial functions $K \fpfn \monoid$ is equipped with a COFE and CMRA structure by lifting everything pointwise.

We obtain the following frame-preserving updates:
\begin{mathpar}
  \inferH{fpfn-alloc-strong}
  {\text{$G$ infinite} \and \melt \in \mval}
  {\emptyset \mupd \setComp{[\gname \mapsto \melt]}{\gname \in G}}

  \inferH{fpfn-alloc}
  {\melt \in \mval}
  {\emptyset \mupd \setComp{[\gname \mapsto \melt]}{\gname \in K}}

  \inferH{fpfn-update}
  {\melt \mupd_\monoid \meltsB}
  {f[i \mapsto \melt] \mupd \setComp{ f[i \mapsto \meltB]}{\meltB \in \meltsB}}
\end{mathpar}
Above, $\mval$ refers to the validity of $\monoid$.

$K \fpfn (-)$ is a locally non-expansive functor from $\CMRAs$ to $\CMRAs$.

\subsection{Agreement}

Given some COFE $\cofe$, we define $\agm(\cofe)$ as follows:
\begin{align*}
  \agm(\cofe) \eqdef{}& \set{(c, V) \in (\mathbb{N} \to \cofe) \times \SProp}/\ {\sim} \\[-0.2em]
  \textnormal{where }& \melt \sim \meltB \eqdef{} \melt.V = \meltB.V \land 
    \All n. n \in \melt.V \Ra \melt.c(n) \nequiv{n} \meltB.c(n)  \\
  \melt \nequiv{n} \meltB \eqdef{}& (\All m \leq n. m \in \melt.V \Lra m \in \meltB.V) \land (\All m \leq n. m \in \melt.V \Ra \melt.c(m) \nequiv{m} \meltB.c(m)) \\
  \mval_n \eqdef{}& \setComp{\melt \in \agm(\cofe)}{ n \in \melt.V \land \All m \leq n. \melt.c(n) \nequiv{m} \melt.c(m) } \\
  \mcore\melt \eqdef{}& \melt \\
  \melt \mtimes \meltB \eqdef{}& \left(\melt.c, \setComp{n}{n \in \melt.V \land n \in \meltB.V \land \melt \nequiv{n} \meltB }\right)
\end{align*}

$\agm(-)$ is a locally non-expansive functor from $\COFEs$ to $\CMRAs$.

You can think of the $c$ as a \emph{chain} of elements of $\cofe$ that has to converge only for $n \in V$ steps.
The reason we store a chain, rather than a single element, is that $\agm(\cofe)$ needs to be a COFE itself, so we need to be able to give a limit for every chain of $\agm(\cofe)$.
However, given such a chain, we cannot constructively define its limit: Clearly, the $V$ of the limit is the limit of the $V$ of the chain.
But what to pick for the actual data, for the element of $\cofe$?
Only if $V = \mathbb{N}$ we have a chain of $\cofe$ that we can take a limit of; if the $V$ is smaller, the chain ``cancels'', \ie stops converging as we reach indices $n \notin V$.
To mitigate this, we apply the usual construction to close a set; we go from elements of $\cofe$ to chains of $\cofe$.

We define an injection $\aginj$ into $\agm(\cofe)$ as follows:
\[ \aginj(x) \eqdef \record{\mathrm c \eqdef \Lam \any. x, \mathrm V \eqdef \mathbb{N}} \]
There are no interesting frame-preserving updates for $\agm(\cofe)$, but we can show the following:
\begin{mathpar}
  \axiomH{ag-val}{\aginj(x) \in \mval_n}

  \axiomH{ag-dup}{\aginj(x) = \aginj(x)\mtimes\aginj(x)}
  
  \axiomH{ag-agree}{\aginj(x) \mtimes \aginj(y) \in \mval_n \Ra x \nequiv{n} y}
\end{mathpar}

\subsection{Exclusive CMRA}

Given a COFE $\cofe$ equipped with a step-indexed relation $(\mstepN{n})$, we define a CMRA $\exm(\cofe)$ such that at most one $x \in \cofe$ can be owned:
\begin{align*}
  \exm(\cofe) \eqdef{}& \exinj(\cofe) + \bot \\
  \mval_n \eqdef{}& \setComp{\melt\in\exm(\cofe)}{\melt \neq \bot}
\end{align*}
All cases of composition go to $\bot$.
\begin{align*}
  \mcore{\exinj(x)} \eqdef{}& \mnocore &
  \mcore{\bot} \eqdef{}& \bot
\end{align*}
Remember that $\mnocore$ is the ``dummy'' element in $\maybe\monoid$ indicating (in this case) that $\exinj(x)$ has no core.

The step-indexed equivalence is inductively defined as follows:
\begin{mathpar}
  \infer{x \nequiv{n} y}{\exinj(x) \nequiv{n} \exinj(y)}

  \axiom{\bot \nequiv{n} \bot}
\end{mathpar}
$\exm(-)$ is a locally non-expansive functor from $\COFEs$ to $\CMRAs$.

We obtain the following frame-preserving update:
\begin{mathpar}
  \inferH{ex-update}{}
  {\exinj(x) \mupd \exinj(y)}
\end{mathpar}

We lift the stepping relation to:

\isnew{
\begin{mathpar}
  \infer{x \mstepN{n} y}{\exinj(x) \mstepN{n} \exinj(y)}
\end{mathpar}
}

\subsection{STS with tokens}
\label{sec:stsmon}

Given a state-transition system~(STS, \ie a directed graph) $(\STSS, {\stsstep} \subseteq \STSS \times \STSS)$, a set of tokens $\STST$, and a labeling $\STSL: \STSS \ra \wp(\STST)$ of \emph{protocol-owned} tokens for each state, we construct an RA modeling an authoritative current state and permitting transitions given a \emph{bound} on the current state and a set of \emph{locally-owned} tokens.

The construction follows the idea of STSs as described in CaReSL \cite{caresl}.
We first lift the transition relation to $\STSS \times \wp(\STST)$ (implementing a \emph{law of token conservation}) and define a stepping relation for the \emph{frame} of a given token set:
\begin{align*}
 (s, T) \stsstep (s', T') \eqdef{}& s \stsstep s' \land \STSL(s) \uplus T = \STSL(s') \uplus T' \\
 s \stsfstep{T} s' \eqdef{}& \Exists T_1, T_2. T_1 \disj \STSL(s) \cup T \land (s, T_1) \stsstep (s', T_2)
\end{align*}

We further define \emph{closed} sets of states (given a particular set of tokens) as well as the \emph{closure} of a set:
\begin{align*}
\STSclsd(S, T) \eqdef{}& \All s \in S. \STSL(s) \disj T \land \left(\All s'. s \stsfstep{T} s' \Ra s' \in S\right) \\
\upclose(S, T) \eqdef{}& \setComp{ s' \in \STSS}{\Exists s \in S. s \stsftrans{T} s' }
\end{align*}

The STS RA is defined as follows
\begin{align*}
  \monoid \eqdef{}& \setComp{\STSauth((s, T) \in \STSS \times \wp(\STST))}{\STSL(s) \disj T} +{}\\& \setComp{\STSfrag((S, T) \in \wp(\STSS) \times \wp(\STST))}{\STSclsd(S, T) \land S \neq \emptyset} + \bot \\
  \STSfrag(S_1, T_1) \mtimes \STSfrag(S_2, T_2) \eqdef{}& \STSfrag(S_1 \cap S_2, T_1 \cup T_2) \qquad\qquad\qquad \text{if $T_1 \disj T_2$ and $S_1 \cap S_2 \neq \emptyset$} \\
\STSfrag(S, T) \mtimes \STSauth(s, T') \eqdef{}& \STSauth(s, T') \mtimes \STSfrag(S, T) \eqdef \STSauth(s, T \cup T') \qquad \text{if $T \disj T'$ and $s \in S$} \\
  \mcore{\STSfrag(S, T)} \eqdef{}& \STSfrag(\upclose(S, \emptyset), \emptyset) \\
  \mcore{\STSauth(s, T)} \eqdef{}& \STSfrag(\upclose(\set{s}, \emptyset), \emptyset) \\
  \isnew{\melt \mstep \melt' \eqdef{}}& \isnew{\exists s, T, \meltB, s', T', \meltB'.\,
  \melt = \STSauth(s, T) \mtimes \meltB \land
  \melt' = \STSauth(s', T') \mtimes \meltB' \land {}}\\
  & \quad \isnew{(s, T) \ststrans (s', T')}
\end{align*}
The remaining cases are all $\bot$.

We will need the following frame-preserving update:
\begin{mathpar}
  \inferH{sts-step}{(s, T) \ststrans (s', T')}
  {\STSauth(s, T) \mupd \STSauth(s', T')}

  \inferH{sts-weaken}
  {\STSclsd(S_2, T_2) \and S_1 \subseteq S_2 \and T_2 \subseteq T_1}
  {\STSfrag(S_1, T_1) \mupd \STSfrag(S_2, T_2)}
\end{mathpar}

\isnew{At the moment we do not make use of the non-trivial step structure on STS's -- all instances of STS's that are used in our present proofs are wrapped in another construction that makes the step structure trivial.}

\paragraph{The core is not a homomorphism.}
The core of the STS construction is only satisfying the RA axioms because we are \emph{not} demanding the core to be a homomorphism---all we demand is for the core to be monotone with respect to the \ruleref{ra-incl} \isnew{and have \ruleref{ra-core-distrib} property. This last rule is kind of like homomorphism \emph{for elements that are already cores}, which is weaker than normal homomorphism.}

In other words, the following does \emph{not} hold for the STS core as defined above:
\[ \mcore\melt \mtimes \mcore\meltB = \mcore{\melt\mtimes\meltB} \]

To see why, consider the following STS:
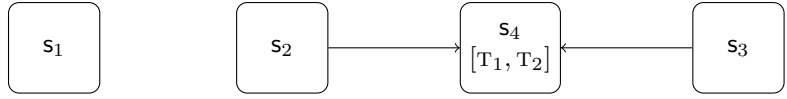
\newcommand\st{\textlog{s}}
\newcommand\tok{\textmon{t}}
\begin{center}
  \begin{tikzpicture}[sts]
    \node at (0,0)   (s1) {$\st_1$};
    \node at (3,0)  (s2) {$\st_2$};
    \node at (9,0) (s3) {$\st_3$};
    \node at (6,0)  (s4) {$\st_4$\\$[\tok_1, \tok_2]$};
    
    \path[sts_arrows] (s2) edge  (s4);
    \path[sts_arrows] (s3) edge  (s4);
  \end{tikzpicture}
\end{center}
Now consider the following two elements of the STS RA:
\[ \melt \eqdef \STSfrag(\set{\st_1,\st_2}, \set{\tok_1}) \qquad\qquad
  \meltB \eqdef \STSfrag(\set{\st_1,\st_3}, \set{\tok_2}) \]

We have:
\begin{mathpar}
  {\melt\mtimes\meltB = \STSfrag(\set{\st_1}, \set{\tok_1, \tok_2})}

  {\mcore\melt = \STSfrag(\set{\st_1, \st_2, \st_4}, \emptyset)}

  {\mcore\meltB = \STSfrag(\set{\st_1, \st_3, \st_4}, \emptyset)}

  {\mcore\melt \mtimes \mcore\meltB = \STSfrag(\set{\st_1, \st_4}, \emptyset) \neq
    \mcore{\melt \mtimes \meltB} = \STSfrag(\set{\st_1}, \emptyset)}
\end{mathpar}

 \endgroup\begingroup
\section{Language}

A \emph{language} $\Lang$ consists of a set \textdom{Expr} of \emph{expressions} (metavariable $\expr$), a set \textdom{Val} of \emph{values} (metavariable $\val$), and a set \textdom{State} of \emph{states} (metvariable $\state$) such that
\begin{itemize}
\item There exist functions $\ofval : \textdom{Val} \to \textdom{Expr}$ and $\toval : \textdom{Expr} \pfn \textdom{val}$ (notice the latter is partial), such that
\begin{mathpar} {\All \expr, \val. \toval(\expr) = \val \Ra \ofval(\val) = \expr} \and {\All\val. \toval(\ofval(\val)) = \val} 
\end{mathpar}
\item There exists a \emph{primitive reduction relation} \[(-,- \step -,-,-) \subseteq \textdom{Expr} \times \textdom{State} \times \textdom{Expr} \times \textdom{State} \times (\textdom{Expr} \uplus \set{\bot})\]
  We will write $\expr_1, \state_1 \step \expr_2, \state_2$ for $\expr_1, \state_1 \step \expr_2, \state_2, \bot$. \\
  A reduction $\expr_1, \state_1 \step \expr_2, \state_2, \expr_\f$ indicates that, when $\expr_1$ reduces to $\expr_2$, a \emph{new thread} $\expr_\f$ is forked off.
\item All values are stuck:
\[ \expr, \_ \step  \_, \_, \_ \Ra \toval(\expr) = \bot \]
\end{itemize}

\begin{defn}
  An expression $\expr$ and state $\state$ are \emph{reducible} (written $\red(\expr, \state)$) if
  \[ \Exists \expr_2, \state_2, \expr_\f. \expr,\state \step \expr_2,\state_2,\expr_\f \]
\end{defn}

\begin{defn}
  An expression $\expr$ is said to be \emph{atomic} if it reduces in one step to a value:
  \[ \All\state_1, \expr_2, \state_2, \expr_\f. \expr, \state_1 \step \expr_2, \state_2, \expr_\f \Ra \Exists \val_2. \toval(\expr_2) = \val_2 \]
\end{defn}

\begin{defn}[Context]
  A function $\lctx : \textdom{Expr} \to \textdom{Expr}$ is a \emph{context} if the following conditions are satisfied:
  \begin{enumerate}[itemsep=0pt]
  \item $\lctx$ does not turn non-values into values:\\
    $\All\expr. \toval(\expr) = \bot \Ra \toval(\lctx(\expr)) = \bot $
  \item One can perform reductions below $\lctx$:\\
    $\All \expr_1, \state_1, \expr_2, \state_2, \expr_\f. \expr_1, \state_1 \step \expr_2,\state_2,\expr_\f \Ra \lctx(\expr_1), \state_1 \step \lctx(\expr_2),\state_2,\expr_\f $
  \item Reductions stay below $\lctx$ until there is a value in the hole:\\
    $\All \expr_1', \state_1, \expr_2, \state_2, \expr_\f. \toval(\expr_1') = \bot \land \lctx(\expr_1'), \state_1 \step \expr_2,\state_2,\expr_\f \Ra \Exists\expr_2'. \expr_2 = \lctx(\expr_2') \land \expr_1', \state_1 \step \expr_2',\state_2,\expr_\f $
  \end{enumerate}
\end{defn}

\subsection{Concurrent language}

For any language $\Lang$, we define the corresponding thread-pool semantics. \isnew{The step relation for thread pool configurations is indexed by the number of the thread that performed a step.}

\paragraph{Machine syntax}
\[
	\tpool \in \textdom{ThreadPool} \eqdef \bigcup_n \textdom{Expr}^n
\]
\[
	\cfgvar \in \textdom{Config} \eqdef \textdom{ThreadPool} \times \textdom{State}
\]

\judgment[Machine reduction]{\cfg{\tpool}{\state} \istep{\isnew{i}}
  \cfg{\tpool'}{\state'}}
\begin{mathpar}
\infer
  {\expr_1, \state_1 \step \expr_2, \state_2, \expr_\f \and \expr_\f \neq \bot \and \isnew{|T| = i}}
  {\cfg{\tpool \dplus [\expr_1] \dplus \tpool'}{\state_1} \istep{\isnew{i}}
     \cfg{\tpool \dplus [\expr_2] \dplus \tpool' \dplus [\expr_\f]}{\state_2}}
\and\infer
  {\expr_1, \state_1 \step \expr_2, \state_2 \and \isnew{|T| = i}}
  {\cfg{\tpool \dplus [\expr_1] \dplus \tpool'}{\state_1} \istep{\isnew{i}}
     \cfg{\tpool \dplus [\expr_2] \dplus \tpool'}{\state_2}}
\end{mathpar}

\isnew{
\begin{defn} We say thread index $i$ is \emph{enabled} in $\cfg{\tpool}{\state}$ if there exists $\tpool'$ and $\state'$ such that $\cfg{\tpool}{\state} \istep{i} \cfg{\tpool'}{\state'}$.
\end{defn}
}

\isnew{
\begin{defn}
A diverging execution\footnote{This may also be defined co-inductively; in the Coq formalization we use a co-inductive definition and give the definition here as a derived one.} of $\cfg{\tpool}{\state}$ is a function ${F: \mathbb{N} \rightarrow \textdom{Config} \times \mathbb{N}}$ such that:
\begin{enumerate} 
 \item $F(0) = (\cfg{\tpool}{\state}, i)$ for some $i$.
 \item For all $n$, if $F(n) = (\cfg{\tpool_n}{\state_n}, j)$ 
                    and $F(n + 1) = (\cfg{\tpool_{n+1}}{\state_{n+1}}, j')$ 
        then $\cfg{\tpool_{n}}{\state_{n}} \istep{j} \cfg{\tpool_{n+1}}{\state_{n+1}}$.
\end{enumerate}
\end{defn}
}

\isnew{
\begin{defn} We say that thread index $i$ is \emph{eventually always enabled} in a diverging execution $F$ if there exists N such that $\forall n \geq N$, $i$ is enabled in $\pi_1(F(n))$.
\end{defn}
}

\isnew{
\begin{defn} We say that thread index $i$ \emph{always eventually steps} in a diverging execution $F$ if for all $n$, there exists $n' \geq n$ such that $\pi_2(F(n')) = i$.
\end{defn}
}

\isnew{
\begin{defn} A diverging execution is \emph{(weakly) fair} if for all $i$, if $i$ is eventually always enabled in $F$, then $i$ always eventually steps in $F$.
\end{defn}
}
\section{Logic}
\label{sec:logic}

To instantiate Iris, you need to define the following parameters:
\begin{itemize}
\item A language $\Lang$, and
\item a locally contractive bifunctor $\iFunc : \COFEs \to \CMRAs$ defining the ghost state, such that for all COFEs $A$, the CMRA $\iFunc(A)$ has a unit. (By \lemref{lem:cmra-unit-total-core}, this means that the core of $\iFunc(A)$ is a total function.)
\end{itemize}

\noindent
As usual for higher-order logics, you can furthermore pick a \emph{signature} $\Sig = (\SigType, \SigFn, \SigAx)$ to add more types, symbols and axioms to the language.
You have to make sure that $\SigType$ includes the base types:
\[
	\SigType \supseteq \{ \textlog{Val}, \textlog{Expr}, \textlog{State}, \textlog{M}, \textlog{InvName}, \textlog{InvMask}, \Prop \}
\]
Elements of $\SigType$ are ranged over by $\sigtype$.

Each function symbol in $\SigFn$ has an associated \emph{arity} comprising a natural number $n$ and an ordered list of $n+1$ types $\type$ (the grammar of $\type$ is defined below, and depends only on $\SigType$).
We write
\[
	\sigfn : \type_1, \dots, \type_n \to \type_{n+1} \in \SigFn
\]
to express that $\sigfn$ is a function symbol with the indicated arity.

Furthermore, $\SigAx$ is a set of \emph{axioms}, that is, terms $\term$ of type $\Prop$.
Again, the grammar of terms and their typing rules are defined below, and depends only on $\SigType$ and $\SigFn$, not on $\SigAx$.
Elements of $\SigAx$ are ranged over by $\sigax$.

\subsection{Grammar}\label{sec:grammar}

\paragraph{Syntax.}
Iris syntax is built up from a signature $\Sig$ and a countably infinite set $\textdom{Var}$ of variables (ranged over by metavariables $x$, $y$, $z$):

\begin{align*}
  \type \bnfdef{}&
      \sigtype \mid
      1 \mid
      \type \times \type \mid
      \type \to \type
\\[0.4em]
  \term, \prop, \pred \bnfdef{}&
      \var \mid
      \sigfn(\term_1, \dots, \term_n) \mid
      () \mid
      (\term, \term) \mid
      \pi_i\; \term \mid
      \Lam \var:\type.\term \mid
      \term(\term)  \mid
      \munit \mid
      \mcore\term \mid
      \term \mtimes \term \mid
\\&
    \FALSE \mid
    \TRUE \mid
    \isnew{\EMP} \mid
    \term =_\type \term \mid
    \prop \Ra \prop \mid
    \prop \land \prop \mid
    \prop \lor \prop \mid
    \prop * \prop \mid
    \prop \wand \prop \mid
\\&
    \MU \var:\type. \term  \mid
    \Exists \var:\type. \prop \mid
    \All \var:\type. \prop \mid
\\&
    \knowInv{\term}{\prop} \mid
    \ownGGhost{\term} \mid
    \isnew{\ownGLGhost{\term} \mid \mval(\term)} \mid
    \isnew{\stopped} \mid
    \mval(\term) \mid
    \ownPhys{\term} \mid
    \always\prop \mid
    \isnew{\aff\prop} \mid
    {\later\prop} \mid
    \pvs[\term][\term] \prop\mid
    \isnew{\psvs[\term][\term] \prop}\mid
    \wpre{\term}[\term]{\Ret\var.\term}
\end{align*}
Recursive predicates must be \emph{guarded}: in $\MU \var. \term$, the variable $\var$ can only appear under the later $\later$ modality.

Note that $\always$, $\later$ bind more tightly than $*$, $\wand$, $\land$, $\lor$, and $\Ra$.
We will write $\pvs[\term] \prop$ for $\pvs[\term][\term] \prop$, \isnew{and similarly for $\psvs[\term] \prop$}.
If we omit the mask, then it is $\top$ for weakest precondition $\wpre\expr{\Ret\var.\prop}$ and $\emptyset$ for primitive view shifts $\pvs \prop$ \isnew{and primitive step shifts $\psvs \prop$}.

Some propositions are \emph{timeless}, which intuitively means that step-indexing does not affect them.
This is a \emph{meta-level} assertion about propositions, defined as follows:

\[ \vctx \proves \timeless{\prop} \eqdef \vctx\mid\later\prop \proves \prop \lor \later\FALSE \]

\isnew{
Similarly, some propositions are \emph{affine timeless}, which means that step-indexing does not affect them \emph{when under an affine modality}:

\[ \vctx \proves \atimeless{\prop} \eqdef \vctx\mid\aff{\later\prop} \proves \aff{\prop} \lor \later\FALSE \]
}

\paragraph{Metavariable conventions.}
We introduce additional metavariables ranging over terms and generally let the choice of metavariable indicate the term's type:
\[
\begin{array}{r|l}
 \text{metavariable} & \text{type} \\\hline
  \term, \termB & \text{arbitrary} \\
  \val, \valB & \textlog{Val} \\
  \expr & \textlog{Expr} \\
  \state & \textlog{State} \\
\end{array}
\qquad\qquad
\begin{array}{r|l}
 \text{metavariable} & \text{type} \\\hline
  \iname & \textlog{InvName} \\
  \mask & \textlog{InvMask} \\
  \melt, \meltB & \textlog{M} \\
  \prop, \propB, \propC & \Prop \\
  \pred, \predB, \predC & \type\to\Prop \text{ (when $\type$ is clear from context)} \\
\end{array}
\]

\paragraph{Variable conventions.}
We assume that, if a term occurs multiple times in a rule, its free variables are exactly those binders which are available at every occurrence.

\subsection{Types}\label{sec:types}

Iris terms are simply-typed.
The judgment $\vctx \proves \wtt{\term}{\type}$ expresses that, in variable context $\vctx$, the term $\term$ has type $\type$.

A variable context, $\vctx = x_1:\type_1, \dots, x_n:\type_n$, declares a list of variables and their types.
In writing $\vctx, x:\type$, we presuppose that $x$ is not already declared in $\vctx$.

\judgment[Well-typed terms]{\vctx \proves_\Sig \wtt{\term}{\type}}
\begin{mathparpagebreakable}
	\axiom{x : \type \proves \wtt{x}{\type}}
\and
	\infer{\vctx \proves \wtt{\term}{\type}}
		{\vctx, x:\type' \proves \wtt{\term}{\type}}
\and
	\infer{\vctx, x:\type', y:\type' \proves \wtt{\term}{\type}}
		{\vctx, x:\type' \proves \wtt{\term[x/y]}{\type}}
\and
	\infer{\vctx_1, x:\type', y:\type'', \vctx_2 \proves \wtt{\term}{\type}}
		{\vctx_1, x:\type'', y:\type', \vctx_2 \proves \wtt{\term[y/x,x/y]}{\type}}
\and
	\infer{
		\vctx \proves \wtt{\term_1}{\type_1} \and
		\cdots \and
		\vctx \proves \wtt{\term_n}{\type_n} \and
		\sigfn : \type_1, \dots, \type_n \to \type_{n+1} \in \SigFn
	}{
		\vctx \proves \wtt {\sigfn(\term_1, \dots, \term_n)} {\type_{n+1}}
	}
\and
	\axiom{\vctx \proves \wtt{()}{1}}
\and
	\infer{\vctx \proves \wtt{\term}{\type_1} \and \vctx \proves \wtt{\termB}{\type_2}}
		{\vctx \proves \wtt{(\term,\termB)}{\type_1 \times \type_2}}
\and
	\infer{\vctx \proves \wtt{\term}{\type_1 \times \type_2} \and i \in \{1, 2\}}
		{\vctx \proves \wtt{\pi_i\,\term}{\type_i}}
\and
	\infer{\vctx, x:\type \proves \wtt{\term}{\type'}}
		{\vctx \proves \wtt{\Lam x. \term}{\type \to \type'}}
\and
	\infer
	{\vctx \proves \wtt{\term}{\type \to \type'} \and \wtt{\termB}{\type}}
	{\vctx \proves \wtt{\term(\termB)}{\type'}}
\and
        \infer{}{\vctx \proves \wtt\munit{\textlog{M}}}
\and
	\infer{\vctx \proves \wtt\melt{\textlog{M}}}{\vctx \proves \wtt{\mcore\melt}{\textlog{M}}}
\and
	\infer{\vctx \proves \wtt{\melt}{\textlog{M}} \and \vctx \proves \wtt{\meltB}{\textlog{M}}}
		{\vctx \proves \wtt{\melt \mtimes \meltB}{\textlog{M}}}
\\
	\axiom{\vctx \proves \wtt{\FALSE}{\Prop}}
\and
	\axiom{\vctx \proves \wtt{\TRUE}{\Prop}}
\and
	\isnew{\axiom{\vctx \proves \wtt{\EMP}{\Prop}}}
\and
	\infer{\vctx \proves \wtt{\term}{\type} \and \vctx \proves \wtt{\termB}{\type}}
		{\vctx \proves \wtt{\term =_\type \termB}{\Prop}}
\and
	\infer{\vctx \proves \wtt{\prop}{\Prop} \and \vctx \proves \wtt{\propB}{\Prop}}
		{\vctx \proves \wtt{\prop \Ra \propB}{\Prop}}
\and
	\infer{\vctx \proves \wtt{\prop}{\Prop} \and \vctx \proves \wtt{\propB}{\Prop}}
		{\vctx \proves \wtt{\prop \land \propB}{\Prop}}
\and
	\infer{\vctx \proves \wtt{\prop}{\Prop} \and \vctx \proves \wtt{\propB}{\Prop}}
		{\vctx \proves \wtt{\prop \lor \propB}{\Prop}}
\and
	\infer{\vctx \proves \wtt{\prop}{\Prop} \and \vctx \proves \wtt{\propB}{\Prop}}
		{\vctx \proves \wtt{\prop * \propB}{\Prop}}
\and
	\infer{\vctx \proves \wtt{\prop}{\Prop} \and \vctx \proves \wtt{\propB}{\Prop}}
		{\vctx \proves \wtt{\prop \wand \propB}{\Prop}}
\and
	\infer{
		\vctx, \var:\type \proves \wtt{\term}{\type} \and
		\text{$\var$ is guarded in $\term$}
	}{
		\vctx \proves \wtt{\MU \var:\type. \term}{\type}
	}
\and
	\infer{\vctx, x:\type \proves \wtt{\prop}{\Prop}}
		{\vctx \proves \wtt{\Exists x:\type. \prop}{\Prop}}
\and
	\infer{\vctx, x:\type \proves \wtt{\prop}{\Prop}}
		{\vctx \proves \wtt{\All x:\type. \prop}{\Prop}}
\and
	\infer{
		\vctx \proves \wtt{\prop}{\Prop} \and
		\vctx \proves \wtt{\iname}{\textlog{InvName}}
	}{
		\vctx \proves \wtt{\knowInv{\iname}{\prop}}{\Prop}
	}
\and
	\infer{\vctx \proves \wtt{\melt}{\textlog{M}}}
		{\vctx \proves \wtt{\ownGGhost{\melt}}{\Prop}}
\and
	\infer{\vctx \proves \wtt{\melt}{\type} \and \text{$\type$ is a CMRA}}
		{\vctx \proves \wtt{\mval(\melt)}{\Prop}}
\and
	\infer{\vctx \proves \wtt{\state}{\textlog{State}}}
		{\vctx \proves \wtt{\ownPhys{\state}}{\Prop}}
\and
	\infer{\vctx \proves \wtt{\prop}{\Prop}}
		{\vctx \proves \wtt{\always\prop}{\Prop}}
\and
        \isnew{\infer{\vctx \proves \wtt{\prop}{\Prop}}
	  {\vctx \proves \wtt{\aff{\prop}}{\Prop}}}
\and
	\infer{\vctx \proves \wtt{\prop}{\Prop}}
		{\vctx \proves \wtt{\later\prop}{\Prop}}
\and
	\infer{
		\vctx \proves \wtt{\prop}{\Prop} \and
		\vctx \proves \wtt{\mask}{\textlog{InvMask}} \and
		\vctx \proves \wtt{\mask'}{\textlog{InvMask}}
	}{
		\vctx \proves \wtt{\pvs[\mask][\mask'] \prop}{\Prop}
	}
\and
        \isnew{
	\infer{
		\vctx \proves \wtt{\prop}{\Prop} \and
		\vctx \proves \wtt{\mask}{\textlog{InvMask}} \and
		\vctx \proves \wtt{\mask'}{\textlog{InvMask}}
	}{
		\vctx \proves \wtt{\psvs[\mask][\mask'] \prop}{\Prop}
	}}
\and
	\infer{
		\vctx \proves \wtt{\expr}{\textlog{Expr}} \and
		\vctx,\var:\textlog{Val} \proves \wtt{\term}{\Prop} \and
		\vctx \proves \wtt{\mask}{\textlog{InvMask}}
	}{
		\vctx \proves \wtt{\wpre{\expr}[\mask]{\Ret\var.\term}}{\Prop}
	}
\end{mathparpagebreakable}

\subsection{Proof rules}
\label{sec:proof-rules}

The judgment $\vctx \mid \pfctx \proves \prop$ says that with free variables $\vctx$, proposition $\prop$ holds whenever all assumptions $\pfctx$ hold.
We implicitly assume that an arbitrary variable context, $\vctx$, is added to every constituent of the rules.
Furthermore, an arbitrary \emph{boxed} assertion context $\always\pfctx$ may be added to every constituent.
Axioms $\vctx \mid \prop \provesIff \propB$ indicate that both $\vctx \mid \prop \proves \propB$ and $\vctx \mid \propB \proves \prop$ can be derived.

\judgment{\vctx \mid \pfctx \proves \prop}
\paragraph{Laws of intuitionistic higher-order logic with equality.}
This is entirely standard.
\begin{mathparpagebreakable}
\infer[Asm]
  {\prop \in \pfctx}
  {\pfctx \proves \prop}
\and
\infer[Eq]
  {\pfctx \proves \prop \\ \pfctx \proves \term =_\type \term'}
  {\pfctx \proves \prop[\term'/\term]}
\and
\infer[Refl]
  {}
  {\pfctx \proves \term =_\type \term}
\and
\infer[$\bot$E]
  {\pfctx \proves \FALSE}
  {\pfctx \proves \prop}
\and
\infer[$\top$I]
  {}
  {\pfctx \proves \TRUE}
\and
\infer[$\wedge$I]
  {\pfctx \proves \prop \\ \pfctx \proves \propB}
  {\pfctx \proves \prop \wedge \propB}
\and
\infer[$\wedge$EL]
  {\pfctx \proves \prop \wedge \propB}
  {\pfctx \proves \prop}
\and
\infer[$\wedge$ER]
  {\pfctx \proves \prop \wedge \propB}
  {\pfctx \proves \propB}
\and
\infer[$\vee$IL]
  {\pfctx \proves \prop }
  {\pfctx \proves \prop \vee \propB}
\and
\infer[$\vee$IR]
  {\pfctx \proves \propB}
  {\pfctx \proves \prop \vee \propB}
\and
\infer[$\vee$E]
  {\pfctx \proves \prop \vee \propB \\
   \pfctx, \prop \proves \propC \\
   \pfctx, \propB \proves \propC}
  {\pfctx \proves \propC}
\and
\infer[$\Ra$I]
  {\pfctx, \prop \proves \propB}
  {\pfctx \proves \prop \Ra \propB}
\and
\infer[$\Ra$E]
  {\pfctx \proves \prop \Ra \propB \\ \pfctx \proves \prop}
  {\pfctx \proves \propB}
\and
\infer[$\forall$I]
  { \vctx,\var : \type\mid\pfctx \proves \prop}
  {\vctx\mid\pfctx \proves \forall \var: \type.\; \prop}
\and
\infer[$\forall$E]
  {\vctx\mid\pfctx \proves \forall \var :\type.\; \prop \\
   \vctx \proves \wtt\term\type}
  {\vctx\mid\pfctx \proves \prop[\term/\var]}
\and
\infer[$\exists$I]
  {\vctx\mid\pfctx \proves \prop[\term/\var] \\
   \vctx \proves \wtt\term\type}
  {\vctx\mid\pfctx \proves \exists \var: \type. \prop}
\and
\infer[$\exists$E]
  {\vctx\mid\pfctx \proves \exists \var: \type.\; \prop \\
   \vctx,\var : \type\mid\pfctx , \prop \proves \propB}
  {\vctx\mid\pfctx \proves \propB}
\end{mathparpagebreakable}
Furthermore, we have the usual $\eta$ and $\beta$ laws for projections, $\lambda$ and $\mu$.

\paragraph{Laws of bunched implications.}
\begin{mathpar}
\begin{array}{rMcMl}
  \isnew{\EMP} * \prop &\provesIff& \prop \\
  \prop * \propB &\provesIff& \propB * \prop \\
  (\prop * \propB) * \propC &\provesIff& \prop * (\propB * \propC)
\end{array}
\and
\infer[$*$-mono]
  {\prop_1 \proves \propB_1 \and
   \prop_2 \proves \propB_2}
  {\prop_1 * \prop_2 \proves \propB_1 * \propB_2}
\and
\inferB[$\wand$I-E]
  {\prop * \propB \proves \propC}
  {\prop \proves \propB \wand \propC}
\end{mathpar}

\paragraph{Laws for ghosts and physical resources.}
\begin{mathpar}
\begin{array}{rMcMl}
\ownGGhost{\melt} * \ownGGhost{\meltB} &\provesIff&  \ownGGhost{\melt \mtimes \meltB} \\
\ownGGhost{\melt} &\proves& \mval(\melt) \\
\TRUE &\proves&  \ownGGhost{\munit}
\end{array}
\and
\begin{array}{c}
\ownPhys{\state} * \ownPhys{\state'} \proves \FALSE
\end{array}
\and
\isnew{
\infer{\forall n.\, \nexists \melt'.\, \melt \mstepN{n} \melt'}{\ownGLGhost{\melt} \proves \stopped}
}
\end{mathpar}

\isnew{Similar rules hold for $\ownGLGhost{\melt}$}. 

\paragraph{Laws for the later modality.}
\begin{mathpar}
\infer[$\later$-mono]
  {\pfctx \proves \prop}
  {\pfctx \proves \later{\prop}}
\and
\infer[L{\"o}b]
  {}
  {(\later\prop\Ra\prop) \proves \prop}
\and
\isnew{\infer[U-L{\"o}b]
  {}
  {\aff{\always{(\aff{\always{\later\prop}}\wand\aff{\always{\prop}})}} \proves \always\prop}}
\and
\infer[$\later$-$\exists$]
  {\text{$\type$ is inhabited}}
  {\later{\Exists x:\type.\prop} \proves \Exists x:\type. \later\prop}
\\\\
\begin{array}[c]{rMcMl}
  \later{(\prop \wedge \propB)} &\provesIff& \later{\prop} \wedge \later{\propB}  \\
  \later{(\prop \vee \propB)} &\provesIff& \later{\prop} \vee \later{\propB} \\
\end{array}
\and
\begin{array}[c]{rMcMl}
  \later{\All x.\prop} &\provesIff& \All x. \later\prop \\
  \Exists x. \later\prop &\proves& \later{\Exists x.\prop}  \\
  \later{(\prop * \propB)} &\provesIff& \later\prop * \later\propB
\end{array}
\end{mathpar}
A type $\type$ being \emph{inhabited} means that $ \proves \wtt{\term}{\type}$ is derivable for some $\term$.

\begin{mathpar}
\infer
{\text{$\term$ or $\term'$ is a discrete COFE element}}
{\timeless{\term =_\type \term'}}

\infer
{\text{$\melt$ is a discrete COFE element}}
{\timeless{\ownGGhost\melt}}

\infer
{\text{$\melt$ is an element of a discrete CMRA}}
{\timeless{\mval(\melt)}}

\infer{}
{\timeless{\ownPhys\state}}

\infer
{\vctx \proves \timeless{\propB}}
{\vctx \proves \timeless{\prop \Ra \propB}}

\infer
{\vctx \proves \timeless{\propB}}
{\vctx \proves \timeless{\prop \wand \propB}}

\infer
{\vctx,\var:\type \proves \timeless{\prop}}
{\vctx \proves \timeless{\All\var:\type.\prop}}

\infer
{\vctx,\var:\type \proves \timeless{\prop}}
{\vctx \proves \timeless{\Exists\var:\type.\prop}}
\end{mathpar}

\isnew{
\begin{mathpar}
\infer
{\text{$\term$ or $\term'$ is a discrete COFE element}}
{\atimeless{\term =_\type \term'}}

\infer
{\text{$\melt$ is a discrete COFE element}}
{\atimeless{\ownGGhost\melt}}

\infer
{\text{$\melt$ is an element of a discrete CMRA}}
{\atimeless{\mval(\melt)}}

\infer{}
{\atimeless{\ownPhys\state}}

\infer
{\vctx \proves \atimeless{\propB}}
{\vctx \proves \atimeless{\prop \Ra \propB}}

\infer
{\vctx \proves \atimeless{\propB}}
{\vctx \proves \atimeless{\prop \wand \propB}}

\infer
{\vctx \proves \atimeless{\prop} \and \vctx \proves \atimeless{\propB}}
{\vctx \proves \atimeless{\prop \land \propB}}

\infer
{\vctx \proves \atimeless{\aff{\prop}} * \vctx \proves \atimeless{\aff{\propB}}}
{\vctx \proves \atimeless{\aff{\prop} * \aff{\propB}}}

\infer
{\vctx \proves \atimeless{\prop} \and \vctx \proves \atimeless{\propB}}
{\vctx \proves \atimeless{\prop \lor \propB}}

\infer
{\vctx,\var:\type \proves \atimeless{\prop}}
{\vctx \proves \atimeless{\All\var:\type.\prop}}

\infer
{\vctx,\var:\type \proves \atimeless{\prop}}
{\vctx \proves \atimeless{\Exists\var:\type.\prop}}
\end{mathpar}
}

\paragraph{Laws for the always/relevant modality.}
\begin{mathpar}
\infer[$\always$I]
  {\always{\pfctx} \proves \prop}
  {\always{\pfctx} \proves \always{\prop}}
\and
\infer[$\always$E]{}
  {\always{\prop} \proves \prop}
\and
\begin{array}[c]{rMcMl}
  \always{(\prop \land \propB)} &\proves& \always{(\prop * \propB)} \\
  \always{\prop} \land \propB &\proves& \always{\prop} * \propB \\
  \isnew{\always{\later\prop}} &\isnew{\provesIff}& \isnew{\always\later\always{\prop}} \\
  \isnew{\always{\later\prop}} &\isnew{\proves}& \isnew{\later\always{\prop}} \\
\end{array}
\and
\begin{array}[c]{rMcMl}
  \isnew{\always\prop * \always\propB} &\isnew{\proves}& \isnew{\always{\prop * \propB}}  \\
  \isnew{\always\prop} &\isnew{\proves}& \isnew{\always{\prop} * \always{\prop}} \\
  \always{(\prop \land \propB)} &\provesIff& \always{\prop} \land \always{\propB} \\
  \always{(\prop \lor \propB)} &\provesIff& \always{\prop} \lor \always{\propB} \\
  \always{\Exists x. \prop} &\provesIff& \Exists x. \always{\prop} \\
\end{array}
\and
{ \term =_\type \term' \proves \always \term =_\type \term'}
\and
{ \knowInv\iname\prop \proves \always \knowInv\iname\prop}
\and
{ \ownGGhost{\mcore\melt} \proves \always \ownGGhost{\mcore\melt}}
\and
{ \mval(\melt) \proves \always \mval(\melt)}
\and
\isnew{\infer{\text{$\type$ is inhabited}}
  {\always{\All x : \type. \prop} \provesIff \All x. \always{\prop}}} \\
\end{mathpar}

\paragraph{\isnew{Laws for the affine modality.}}
\begin{mathpar}
\begin{array}[c]{rMcMl}
  \aff{\prop} &\proves& \prop \\
  \aff{\TRUE} &\provesIff& \EMP \\
  \aff{\aff{\prop}} &\provesIff& \aff{\prop} \\
  \prop * \aff{\propB} &\proves& \prop \\
  \aff{\aff{\prop} * \aff{\propB}} &\provesIff& \aff{\prop} * \aff{\propB} \\
  \aff{\prop} \land \propB &\provesIff& \aff{\prop \land \propB} \\
  {\aff{\later\prop}} &\provesIff& \aff{\later\aff{\prop}} \\
  {\aff{\later\prop}} &\proves& \later\aff{\prop} \\
  {\aff{\later(\aff{\prop} * \aff{\propB})}} &\proves& \aff{\later\prop} * \aff{\later\prop} \\
  {\aff{\always\prop}} &\provesIff& \always{\aff{\prop}} \\
  { \term =_\type \term'} &\proves& {\aff{\term =_\type \term'}} \\
  { \knowInv\iname\prop} &\proves& {\aff{\knowInv\iname\prop}} \\
  { \ownGGhost{\mcore\melt}} &\proves&\aff{\ownGGhost{\mcore\melt}}  \\
  { \mval(\melt)} &\proves& \aff{\mval(\melt)}
\end{array}
\end{mathpar}
\begin{mathpar}
{\infer{\text{$\type$ is inhabited}}
  {\aff{\All x : \type. \prop} \provesIff \All x. \aff{\prop}}}
\and
{\infer{\text{$\type$ is inhabited}}
  {\aff{\Exists x : \type. \prop} \provesIff \Exists x. \aff{\prop}}}
\end{mathpar}

\paragraph{Laws of primitive view shifts.}
\begin{mathpar}
\infer[pvs-intro]
{}{\prop \proves \pvs[\mask] \prop}

\infer[pvs-mono]
{\prop \proves \propB}
{\pvs[\mask_1][\mask_2] \prop \proves \pvs[\mask_1][\mask_2] \propB}

\infer[pvs-timeless]
{\timeless\prop}
{\later\prop \proves \pvs[\mask] \prop}

\isnew{
\infer[pvs-atimeless]
{\atimeless\prop}
{\aff{\later\prop} \proves \pvs[\mask] \aff{\prop}}
}

\infer[pvs-trans]
{\mask_2 \subseteq \mask_1 \cup \mask_3}
{\pvs[\mask_1][\mask_2] \pvs[\mask_2][\mask_3] \prop \proves \pvs[\mask_1][\mask_3] \prop}

\infer[pvs-mask-frame]
{}{\pvs[\mask_1][\mask_2] \prop \proves \pvs[\mask_1 \uplus \mask_\f][\mask_2 \uplus \mask_\f] \prop}

\infer[pvs-frame]
{}{\propB * \pvs[\mask_1][\mask_2]\prop \proves \pvs[\mask_1][\mask_2] \propB * \prop}

\inferH{pvs-allocI}
{\text{$\mask$ is infinite}}
{\isnew{\aff{\later\prop}} \proves \pvs[\mask] \Exists \iname \in \mask. \knowInv\iname\prop}

\inferH{pvs-openI}
{}{\knowInv\iname\prop \proves \pvs[\set\iname][\emptyset] \isnew{\aff{\later\prop}}}

\inferH{pvs-closeI}
{}{\isnew{\knowInv\iname\prop * \aff{\later\prop} \proves \pvs[\emptyset][\set\iname] \EMP}}

\inferH{pvs-update}
{\melt \mupd \meltsB}
{\ownGGhost\melt \proves \pvs[\mask] \Exists\meltB\in\meltsB. \ownGGhost\meltB}

\isnew{\infer[pvs-affine]
{}{\aff{\pvs[\mask_1][\mask_2]\prop} \provesIff \pvs[\mask_1][\mask_2]\aff{\prop}}}
\end{mathpar}

\paragraph{\isnew{Laws of primitive step shifts.}}
\begin{mathpar}
\infer[psvs-intro]
{}{\prop \proves \pvs[\mask] \prop}

\infer[psvs-mono]
{\prop \proves \propB}
{\psvs[\mask_1][\mask_2] \prop \proves \psvs[\mask_1][\mask_2] \propB}

\infer[pvs-psvs]
{\mask_2 \subseteq \mask_1 \cup \mask_3}
{\pvs[\mask_1][\mask_2] \psvs[\mask_2][\mask_3] \prop \proves \psvs[\mask_1][\mask_3] \prop}

\infer[psvs-pvs]
{\mask_2 \subseteq \mask_1 \cup \mask_3}
{\psvs[\mask_1][\mask_2] \pvs[\mask_2][\mask_3] \prop \proves \psvs[\mask_1][\mask_3] \prop}

\infer[psvs-mask-frame]
{}{\psvs[\mask_1][\mask_2] \prop \proves \psvs[\mask_1 \uplus \mask_\f][\mask_2 \uplus \mask_\f] \prop}

\infer[psvs-frame]
{}{\aff{\propB} * \psvs[\mask_1][\mask_2]\prop \proves \psvs[\mask_1][\mask_2](\aff{\propB} * \prop)}

\inferH{psvs-step}
{\melt, \meltB \msupd \meltsB}
{\ownGGhost\melt * \ownGLGhost\meltB \proves \psvs[\mask] \Exists\melt', \meltB'\in\meltsB. \ownGGhost{\melt'} * \ownGLGhost{\meltB'}}

\end{mathpar}

\paragraph{Laws of weakest preconditions.}
\begin{mathpar}
\infer[wp-value]
{}{\prop[\val/\var] \proves \wpre{\val}[\mask]{\Ret\var.\prop}}

\infer[wp-mono]
{\mask_1 \subseteq \mask_2 \and \var:\textlog{val}\mid\prop \proves \propB}
{\wpre\expr[\mask_1]{\Ret\var.\prop} \proves \wpre\expr[\mask_2]{\Ret\var.\propB}}

\infer[pvs-wp]
{}{\pvs[\mask] \wpre\expr[\mask]{\Ret\var.\prop} \proves \wpre\expr[\mask]{\Ret\var.\prop}}

\infer[wp-pvs]
{}{\wpre\expr[\mask]{\Ret\var.\pvs[\mask] \prop} \proves \wpre\expr[\mask]{\Ret\var.\prop}}

\infer[wp-atomic]
{\mask_2 \subseteq \mask_1 \and \physatomic{\expr}}
{\pvs[\mask_1][\mask_2] \wpre\expr[\mask_2]{\Ret\var. \pvs[\mask_2][\mask_1]\prop}
 \proves \wpre\expr[\mask_1]{\Ret\var.\prop}}

\infer[wp-frame]
{}{\isnew{\aff{\propB}} * \wpre\expr[\mask]{\Ret\var.\prop} 
  \proves \wpre\expr[\mask]{\Ret\var.\isnew{\aff{\propB}}*\prop}}

\infer[wp-frame-step]
{\toval(\expr) = \bot \and \mask_2 \subseteq \mask_1}
{\wpre\expr[\mask]{\Ret\var.\prop} * \isnew{\aff{\pvs[\mask_1][\mask_2]\later\pvs[\mask_2][\mask_1]\propB}} \proves \wpre\expr[\mask \uplus \mask_1]{\Ret\var.\isnew{\aff{\propB}}*\prop}}

\infer[wp-bind]
{\text{$\lctx$ is a context}}
{\wpre\expr[\mask]{\Ret\var. \wpre{\lctx(\ofval(\var))}[\mask]{\Ret\varB.\prop}} \proves \wpre{\lctx(\expr)}[\mask]{\Ret\varB.\prop}}
\end{mathpar}

\paragraph{Lifting of operational semantics.}~
\begin{mathpar}
  \infer[wp-lift-step]
  {\mask_2 \subseteq \mask_1 \and
   \toval(\expr_1) = \bot}
  { {\begin{inbox} %
        ~~\pvs[\mask_1][\mask_2] \Exists \state_1. \isnew{\aff{\red(\expr_1,\state_1) \land \later\ownPhys{\state_1}}} * {}\\\qquad\qquad\qquad \later\All \expr_2, \state_2, \expr_\f. \left( (\expr_1, \state_1 \step \expr_2, \state_2, \expr_\f) \land \ownPhys{\state_2} \right) \wand \isnew{\psvs[\mask_2][\mask_1]} \wpre{\expr_2}[\mask_1]{\Ret\var.\prop} * \wpre{\expr_\f}[\top]{\Ret\any.\isnew{\stopped}}  {}\\\proves \wpre{\expr_1}[\mask_1]{\Ret\var.\prop}
      \end{inbox}} }
\\\\
  \infer[wp-lift-pure-step]
  {\toval(\expr_1) = \bot \and
   \All \state_1. \red(\expr_1, \state_1) \and
   \All \state_1, \expr_2, \state_2, \expr_\f. \expr_1,\state_1 \step \expr_2,\state_2,\expr_\f \Ra \state_1 = \state_2 }
  {\later\All \state, \expr_2, \expr_\f. (\expr_1,\state \step \expr_2, \state,\expr_\f)  \wand \isnew{\psvs[\mask_1]} \wpre{\expr_2}[\mask_1]{\Ret\var.\prop} * \wpre{\expr_\f}[\top]{\Ret\any.\isnew{\stopped}} \proves \wpre{\expr_1}[\mask_1]{\Ret\var.\prop}}
\end{mathpar}
Notice that primitive view shifts cover everything to their right, \ie $\pvs \prop * \propB \eqdef \pvs (\prop * \propB)$, \isnew{and similarly for primitive step shifts}.

Here we define $\wpre{\expr_\f}[\mask]{\Ret\var.\prop} \eqdef \isnew{\EMP}$ if $\expr_\f = \bot$ (remember that our stepping relation can, but does not have to, define a forked-off expression).

\subsection{Adequacy}

\paragraph{Finite Executions and Safety.}
The adequacy statement concerning functional correctness reads as follows:
\begin{align*}
 &\All \mask, \expr, \val, \pred, \state, \melt, \isnew{\meltB}, \state', \tpool'.
 \\&(\All n. \melt \mtimes \isnew{\meltB} \in \mval_n) \Ra
 \\&( \ownPhys\state * \ownGGhost\melt * \isnew{\ownGLGhost\meltB} \proves \wpre{\expr}[\mask]{x.\; \pred(x)}) \Ra
 \\&\cfg{\state}{[\expr]} \step^\ast
     \cfg{\state'}{[\val] \dplus \tpool'} \Ra
     \\&\pred(\val)
\end{align*}
where $\pred$ is a \emph{meta-level} predicate over values, \ie it can mention neither resources nor invariants.

Furthermore, the following adequacy statement shows that our weakest preconditions imply that the execution never gets \emph{stuck}: Every expression in the thread pool either is a value, or can reduce further.
\begin{align*}
 &\All \mask, \expr, \state, \melt, \isnew{\meltB}, \state', \tpool'.
 \\&(\All n. \melt * \isnew{\meltB} \in \mval_n) \Ra
 \\&( \ownPhys\state * \ownGGhost\melt * \isnew{\ownGLGhost\meltB}
 \proves \wpre{\expr}[\mask]{x.\; \pred(x)}) \Ra
 \\&\cfg{\state}{[\expr]} \step^\ast
     \cfg{\state'}{\tpool'} \Ra
     \\&\All\expr'\in\tpool'. \toval(\expr') \neq \bot \lor \red(\expr', \state')
\end{align*}
Notice that this is stronger than saying that the thread pool can reduce; we actually assert that \emph{every} non-finished thread can take a step.

\paragraph{\isnew{Diverging Executions.}}

Remember that our goal is to show that if we have proved
$\wpre{\expr}{x.\;\propB}$ and $\expr$ has a fair diverging execution,
then there is a fair diverging execution of some corresponding
source program. Of course, there is no such notion of a ``source
program'' baked into the logic at this point yet. 

All we have is the step relation $\mstepN{n}$ on CMRA elements. The
rules above for $\wpre{\expr}{x.\; Q}$ suggest that\footnote{We shall see that this is indeed the case when we examine the definition of weakest precondition in the model.} for each step
$\expr$ takes, we are required to perform a step-shift, thereby
performing reduction steps on CMRA resources.

We can lift the $\mstepN{n}$
relation from CMRA elements to lists of CMRA elements, much as we lift
the per-thread step relation to an indexed relation on threadpools:
\begin{mathpar}
\infer
  {\melt_i \mstepN{n} (\melt_i' \mtimes \melt_\f)}
  {[\melt_0, \dots, \melt_i, \dots, \melt_k] \mstepNi{n}{i}
     [\melt_0, \dots, \melt_i', \dots, \melt_k, \melt_\f]}
\and
\infer
  {\melt_i \mstepN{n} \melt_i'}
  {[\melt_0, \dots, \melt_i, \dots, \melt_k] \mstepNi{n}{i}
     [\melt_0, \dots, \melt_i', \dots, \melt_k]}
\end{mathpar}
In general we shall use $\mlist$ and $\mlistB$ as metavariables for such lists of CMRA elements, and write $\bigast \mlist$ to represent the product of the elements of $\mlist$.

\begin{defn} We say index $i$ is $n$-\emph{enabled} in $\mlist$ if there exists $\mlistB$ such that
  $\mlist \mstepNi{n}{i} \mlistB$. 
\end{defn}

Somehow, we want to connect up these CMRA steps to steps in some source language, for a suitably chosen CMRA.
Still working a bit more abstractly than that for the moment, let $\cofeB$ be some COFE equipped with a family of relations 
$(\istep{i} : \maybe\cofeB \times \maybe\cofeB)_{i\in \mathbb{N}}$. We can adapt all of our definitions about fairness to the setting of this reduction on $\maybe\cofeB$, e.g.:

\begin{defn} We say index $i$ is \emph{enabled} in $\celt \in \maybe\cofeB$ if there exists $\celtB$ such that $\celt \istep{i} \celtB$. 
\end{defn}

\begin{defn}
A diverging execution of $\celt$ is a function ${F: \mathbb{N} \rightarrow \maybe\cofeB \times \mathbb{N}}$ such that:
\begin{enumerate} 
 \item $F(0) = (\celt, i)$ for some $i$.
 \item For all $n$, if $F(n) = (\celtB, j)$
                    and $F(n + 1) = (\celtB', j')$ 
        then $\celtB \istep{j} \celtB'$.
\end{enumerate}
\end{defn}
and so on.

\begin{defn}
We say $H: \mathbb{N} \rightarrow \textdom{List \monoid} \rightarrow \maybe\cofeB$ is a step-preserving map if the following conditions hold:
\begin{enumerate}
\item If $i$ is enabled in $H(n, \mlist)$ then $i$ is $n$-enabled in $\mlist$.
\item If $\bigast \mlist \in \mval_{n}$, $\bigast \mlistB \in \mval_{n}$,
  $\mlist \mstepNi{n}{i} \mlistB$ and $H(n, \mlist) \in \cofeB$, then $H(n, \mlistB) \in \cofeB$ and $H(n, \mlist) \istep{i} H(n, \mlistB)$
\item If $\bigast \mlist \in \mval_{n}$ then $\forall n' \leq n$, $H(n, \mlist) = H(n', \mlist)$
\item If $H(n, \mlist) \istep{i} H(n, \mlistB)$ then $\forall n'\leq n$, $H(n', \mlist) \istep{i} H(n', \mlistB)$
\end{enumerate}
\end{defn}

\begin{defn}
$\cofeB$ has \emph{bounded non-determinism} if $\forall \celt \in \cofeB$, the set $\{ \celtB \ | \exists i.\, \celt \istep{i} \celtB\}$ is finite.
\end{defn}

We are now able to state the infinite adequacy theorem:
\begin{thm}
Assume $\cofeB$ has bounded non-determinism under a step relation
$(\istep{i} : \maybe\cofeB \times \maybe\cofeB)_{i \in \mathbb{N}}$. 
Let $H$ be a step-preserving map to $\cofeB$. If
$\cfg{\expr}{\state}$ has a diverging execution, and all of the following hold:
\begin{enumerate}
 \item $n > 2$,
 \item $\All n'. \melt \mtimes \meltB \in \mval_{n'}$,
 \item $\ownPhys\state * \ownGGhost\melt * \ownGLGhost\meltB \proves \wpre{\expr}[\mask]{x.\; \pred(x)})$,
 \item $\forall x.\, \pred(x) \proves \stopped$,
 \item $H(n, [\meltB]) \in \cofeB$
\end{enumerate}
then there exists a diverging execution of $H(n, [\meltB])$. Moreover, if the execution of $\cfg{\expr}{\state}$ was fair, so too is the execution of $H(n, [\meltB])$.
\end{thm}

 \endgroup\begingroup
\section{Model and semantics}
\label{sec:model}

The semantics closely follows the ideas laid out in~\cite{catlogic}.

\subsection{Generic model of base logic}
\label{sec:upred-logic}

The base logic including equality, later, always, and a notion of ownership is defined on $\UPred(\monoid)$ for any CMRA $\monoid$.

\typedsection{Interpretation of base assertions}{\Sem{\vctx \proves \term : \Prop} : \Sem{\vctx} \nfn \UPred(\monoid)}
The type $\UPred(\monoid)$ is isomorphic to $\monoid \monra \isnew{\monoid \ra} \SProp$.
We are thus going to define the assertions as mapping \isnew{pairs of} CMRA elements to sets of step-indices.

We introduce an additional logical connective $\ownM\melt$ and \isnew{$\ownMl\melt$}, which will later be used to encode all of $\knowInv\iname\prop$, $\ownGGhost\melt$, $\isnew{\ownGLGhost\melt}$ and $\ownPhys\state$.

\begin{align*}
	\Sem{\vctx \proves t =_\type u : \Prop}_\gamma &\eqdef
	\Lam \any, \meltB. \setComp{n}{\Sem{\vctx \proves t : \type}_\gamma \nequiv{n} \Sem{\vctx \proves u : \type}_\gamma \wedge \isnew{\meltB \nequiv{n} \munit}} \\
	\Sem{\vctx \proves \FALSE : \Prop}_\gamma &\eqdef \Lam \any, \any. \emptyset \\
	\Sem{\vctx \proves \TRUE : \Prop}_\gamma &\eqdef \Lam \any, \any. \mathbb{N} \\
	\isnew{\Sem{\vctx \proves \EMP : \Prop}_\gamma } & \isnew{\ \eqdef \Lam \any, \meltB. \setComp{n}{\meltB \nequiv{n} \munit}} \\
	\Sem{\vctx \proves \prop \land \propB : \Prop}_\gamma &\eqdef
	\Lam \melt, \meltB. \Sem{\vctx \proves \prop : \Prop}_\gamma(\melt, \meltB) \cap \Sem{\vctx \proves \propB : \Prop}_\gamma(\melt, \meltB) \\
	\Sem{\vctx \proves \prop \lor \propB : \Prop}_\gamma &\eqdef
	\Lam \melt, \meltB. \Sem{\vctx \proves \prop : \Prop}_\gamma(\melt, \meltB) \cup \Sem{\vctx \proves \propB : \Prop}_\gamma(\melt, \meltB) \\
	\Sem{\vctx \proves \prop \Ra \propB : \Prop}_\gamma &\eqdef
	\Lam \melt, \meltB. \setComp{n}{\begin{aligned}
            \All m, \melt'.& m \leq n \land \melt \mincl \melt' \land \melt' \in \mval_m  \land \meltB \in \mval_{m} \Ra {} \\
            & m \in \Sem{\vctx \proves \prop : \Prop}_\gamma(\melt', \meltB) \Ra {}\\& m \in \Sem{\vctx \proves \propB : \Prop}_\gamma(\melt', \meltB)\end{aligned}}\\
	\Sem{\vctx \proves \All x : \type. \prop : \Prop}_\gamma &\eqdef
	\Lam \melt, \meltB. \setComp{n}{ \All v \in \Sem{\type}. n \in \Sem{\vctx, x : \type \proves \prop : \Prop}_{\gamma[x \mapsto v]}(\melt, \meltB) } \\
	\Sem{\vctx \proves \Exists x : \type. \prop : \Prop}_\gamma &\eqdef
        \Lam \melt, \meltB. \setComp{n}{ \Exists v \in \Sem{\type}. n \in \Sem{\vctx, x : \type \proves \prop : \Prop}_{\gamma[x \mapsto v]}(\melt, \meltB) } \\
  ~\\
	\Sem{\vctx \proves \always{\prop} : \Prop}_\gamma &\eqdef \Lam\melt,\meltB. \Sem{\vctx \proves \prop : \Prop}_\gamma(\mcore\melt, \mcore\meltB) \cap 
        \isnew{\setComp{n}{\meltB \nequiv{n} \mcore\meltB}} \\
	\isnew{\Sem{\vctx \proves \aff{\prop} : \Prop}_\gamma} &\isnew{\ \eqdef \Lam\melt,\meltB. \Sem{\vctx \proves \prop : \Prop}_\gamma(\melt, \meltB) \cap 
        \setComp{n}{\meltB \nequiv{n} \munit}} \\
	\Sem{\vctx \proves \later{\prop} : \Prop}_\gamma &\eqdef \Lam\melt,\meltB. \setComp{n}{n = 0 \lor n-1 \in \Sem{\vctx \proves \prop : \Prop}_\gamma(\melt,\meltB)}\\
	\Sem{\vctx \proves \prop * \propB : \Prop}_\gamma &\eqdef \Lam\melt,\meltB. \setComp{n}{\begin{aligned} &\Exists \melt_1, \melt_2, \meltB_1, \meltB_2. \melt \nequiv{n} \melt_1 \mtimes \melt_2 \land \isnew{\meltB \nequiv{n} \meltB_1 \mtimes \meltB_2} \land {}\\& \quad n \in \Sem{\vctx \proves \prop : \Prop}_\gamma(\melt_1, \meltB_1) \land n \in \Sem{\vctx \proves \propB : \Prop}_\gamma(\melt_2, \meltB_2)\end{aligned}}
\\
	\Sem{\vctx \proves \prop \wand \propB : \Prop}_\gamma &\eqdef
	\Lam \melt,\meltB. \setComp{n}{\begin{aligned}
            \All m, \melt', \meltB'.& m \leq n \land  \melt\mtimes\melt' \in \mval_m \land \isnew{\meltB\mtimes\meltB' \in \mval_m} \Ra {} \\
            & m \in \Sem{\vctx \proves \prop : \Prop}_\gamma(\melt', \meltB') \Ra {}\\& m \in \Sem{\vctx \proves \propB : \Prop}_\gamma(\melt\mtimes\melt', \meltB\mtimes\meltB')\end{aligned}} \\
        \Sem{\vctx \proves \ownM{\melt} : \Prop}_\gamma &\eqdef \Lam\melt',\meltB. \setComp{n}{\Sem{\vctx \proves \melt' : \textlog{M}} \mincl[n] \melt \land \isnew{\meltB \nequiv{n} \munit}}
        \\
        \isnew{\Sem{\vctx \proves \ownMl{\meltB} : \Prop}}_\gamma &\eqdef \isnew{\Lam\melt,\meltB'. \setComp{n}{\isnew{\meltB \nequiv{n} \meltB'}}}
        \\
        \Sem{\vctx \proves \mval(\melt) : \Prop}_\gamma &\eqdef \Lam\any, \meltB. \setComp{n}{\Sem{\vctx \proves \melt : \type} \in \mval_n \land \isnew{\meltB \nequiv{n} \munit}} \\
\end{align*}

\isnew{Note: There are slight differences between the definition here and the version in the Coq development. For instance, $\ownM{a}$ in the Coq development does not stipulate that the second component is in fact equivalent to $\munit$, which makes it non-affine; but in practice we almost always use $\aff{\ownM{a}}$, so here we just present a version equivalent to that.}

For every definition, we have to show all the side-conditions: The maps have to be non-expansive and monotone.

\subsection{Iris model}

\paragraph{Semantic domain of assertions.}
The first complicated task in building a model of full Iris is defining the semantic model of $\Prop$.
We start by defining the functor that assembles the CMRAs we need to the global resource CMRA:
\begin{align*}
  \textdom{ResF}(\cofe^\op, \cofe) \eqdef{}& \record{\wld: \mathbb{N} \fpfn \agm(\latert \cofe), \pres: \maybe{\exm(\textdom{State})}, \ghostRes: \iFunc(\cofe^\op, \cofe)}
\end{align*}
Above, $\maybe\monoid$ is the monoid obtained by adding a unit to $\monoid$.
(It's not a coincidence that we used the same notation for the range of the core; it's the same type either way: $\monoid + 1$.)
Remember that $\iFunc$ is the user-chosen bifunctor from $\COFEs$ to $\CMRAs$ (see~\Sref{sec:logic}).
$\textdom{ResF}(\cofe^\op, \cofe)$ is a CMRA by lifting the individual CMRAs pointwise.
Furthermore, since $\Sigma$ is locally contractive, so is $\textdom{ResF}$.

Now we can write down the recursive domain equation:
\[ \iPreProp \cong \UPred(\textdom{ResF}(\iPreProp, \iPreProp)) \]
$\iPreProp$ is a COFE defined as the fixed-point of a locally contractive bifunctor.
This fixed-point exists and is unique by America and Rutten's theorem~\cite{America-Rutten:JCSS89,birkedal:metric-space}.
We do not need to consider how the object is constructed. 
We only need the isomorphism, given by
\begin{align*}
  \Res &\eqdef \textdom{ResF}(\iPreProp, \iPreProp) \\
  \iProp &\eqdef \UPred(\Res) \\
	\wIso &: \iProp \nfn \iPreProp \\
	\wIso^{-1} &: \iPreProp \nfn \iProp
\end{align*}

We then pick $\iProp$ as the interpretation of $\Prop$:
\[ \Sem{\Prop} \eqdef \iProp \]

\paragraph{Interpretation of assertions.}
$\iProp$ is a $\UPred$, and hence the definitions from \Sref{sec:upred-logic} apply.
We only have to define the interpretation of the missing connectives, the most interesting bits being primitive view shifts and weakest preconditions.

\typedsection{World satisfaction}{\wsat{-}{-}{-}{-} : 
	\Delta\textdom{State} \times
	\Delta\pset{\mathbb{N}} \times
	\textdom{Res} \times \isnew{\textdom{Res}} \nfn \SProp }
\begin{align*}
  \wsatpre(n, \mask, \state, \rss, \rs, \rsB) & \eqdef \begin{inbox}[t]
    \isnew{(\rs \mtimes \rsB)} \in \mval_{n+1} \land \rs.\pres = \exinj(\sigma) \land 
    \dom(\rss) \subseteq \mask \cap \dom( \rs.\wld) \land {}\\
    \All\iname \in \mask, \prop \in \iProp. (\rs.\wld)(\iname) \nequiv{n+1} \aginj(\latertinj(\wIso(\prop))) \Ra n \in \prop(\rss(\iname), \isnew{\munit})
  \end{inbox}\\
	\wsat{\state}{\mask}{\rs, \rsB} &\eqdef \set{0}\cup\setComp{n+1}{\Exists \rss : \mathbb{N} \fpfn \textdom{Res}. \wsatpre(n, \mask, \state, \rss, \rs \mtimes \prod_\iname \rss(\iname), \isnew{\rsB})}
\end{align*}

\isnew{Notice that the assertion $\prop$ corresponding to the world $\iname$ must be \emph{affine}, in the sense that $n \in \prop(\rss(\iname), \isnew{\munit})$. Without this stipulation, threads would be able to put linear resources inside invariants and STS interpretations. We prevent this in the current formulation of the logic because otherwise the rules for dealing with invariants would have to be stricter if we still wanted to establish fair refinements.}

\typedsection{Primitive view-shift}{\mathit{pvs}_{-}^{-}(-) : \Delta(\pset{\mathbb{N}}) \times \Delta(\pset{\mathbb{N}}) \times \iProp \nfn \iProp}
\begin{align*}
	\mathit{pvs}_{\mask_1}^{\mask_2}(\prop) &= \Lam \rs, \rsB. \setComp{n}{\begin{aligned}
            \All \rs_\f, \rsB_\f, k, \mask_\f, \state.& 0 < k \leq n \land (\mask_1 \cup \mask_2) \disj \mask_\f \land k \in \wsat\state{\mask_1 \cup \mask_\f}
                               {\rs \mtimes \rs_\f, \isnew{\rsB \mtimes \rsB_\f}} 
            \Ra {}\\&
            \Exists \rs'. k \in \prop(\rs', \rsB) \land k \in \wsat\state{\mask_2 \cup \mask_\f}{\rs' \mtimes \rs_\f, \isnew{\rsB \mtimes \rsB_\f}}
          \end{aligned}}
\end{align*}

\typedsection{\isnew{Primitive step-shift}}{\mathit{psvs}_{-}^{-}(-) : \Delta(\pset{\mathbb{N}}) \times \Delta(\pset{\mathbb{N}}) \times \iProp \nfn \iProp}
\begin{align*}
	\mathit{psvs}_{\mask_1}^{\mask_2}(\prop) &= \Lam \rs, \rsB. \setComp{n}{\begin{aligned}
            \All \rs_\f, \rsB_\f, k, \mask_\f, \state.& 0 < k \leq n \land (\mask_1 \cup \mask_2) \disj \mask_\f \land k \in \wsat\state{\mask_1 \cup \mask_\f}
                               {\rs \mtimes \rs_\f, \rsB \mtimes \rsB_\f} 
            \Ra {}\\&
            \Exists \rs', \rsB'. k \in \prop(\rs', \rsB') \land%
            k \in \wsat\state{\mask_2 \cup \mask_\f}{\rs' \mtimes \rs_\f, \rsB' \mtimes \rsB_\f} \land%
            \rsB \mstepN{k} \rsB'
          \end{aligned}}
\end{align*}

\typedsection{Weakest precondition}{\mathit{wp}_{-}(-, -) : \Delta(\pset{\mathbb{N}}) \times \Delta(\textdom{Exp}) \times (\Delta(\textdom{Val}) \nfn \iProp) \nfn \iProp}

$\textdom{wp}$ is defined as the fixed-point of a contractive function.
{
\scriptsize{
\begin{align*}
  \textdom{pre-wp}(\textdom{wp})(\mask, \expr, \pred) &\eqdef \Lam\rs,\rsB.\setComp{n}{\begin{aligned}
        \All &\rs_\f,\rsB_\f, m, \mask_\f, \state. 0 \leq m < n \land \mask \disj \mask_\f \land m+1 \in \wsat\state{\mask \cup \mask_\f}{\rs \mtimes \rs_\f, \rsB \mtimes \rsB_\f} \Ra {}\\
        &(\All\val. \toval(\expr) = \val \Ra \Exists \rs'. m+1 \in \pred(\val)(\rs', \isnew{\rsB}) \land m+1 \in \wsat\state{\mask \cup \mask_\f}{\rs' \mtimes \rs_\f, \isnew{\rsB \mtimes \rsB_\f}})
        \land {}\\
        &(\toval(\expr) = \bot \land 0 < m \Ra \red(\expr, \state) \land \All \expr_2, \state_2, \expr_\f. \expr,\state \step \expr_2,\state_2,\expr_\f \Ra {}\\
        &\qquad \Exists \rs_1, \rs_2, \isnew{\rsB_1, \rsB_2}. m \in \wsat\state{\mask \cup \mask_\f}{\rs_1 \mtimes \rs_2 \mtimes \rs_\f} \land  m \in \textdom{wp}(\mask, \expr_2, \pred)(\rs_1, \rsB_1) 
        \land {}&\\
        &\qquad\qquad ((\expr_\f = \bot \wedge \isnew{\rsB_2 \nequiv{m} \munit})
        \lor m \in \textdom{wp}(\top, \expr_\f, \Lam\any.\Lam\any.\mathbb{N})(\rs_2, \isnew{\rsB_2}))
        \land {}&\\
        &\qquad\qquad\isnew{\rsB \mstepN{m} \rsB_1 \mtimes \rsB_2}
    \end{aligned}} \\
  \textdom{wp}_\mask(\expr, \pred) &\eqdef \mathit{fix}(\textdom{pre-wp})(\mask, \expr, \pred)
\end{align*}
}
}

\typedsection{Interpretation of program logic assertions}{\Sem{\vctx \proves \term : \Prop} : \Sem{\vctx} \nfn \iProp}

$\knowInv\iname\prop$, $\ownGGhost\melt$ and $\ownPhys\state$ are just syntactic sugar for forms of $\ownM{-}$.
\begin{align*}
	\knowInv{\iname}{\prop} &\eqdef \ownM{[\iname \mapsto \aginj(\latertinj(\wIso(\prop)))], \munit, \munit} \\
	\ownGGhost{\melt} &\eqdef \ownM{\munit, \munit, \melt} \\
	\isnew{\ownGLGhost{\melt}} &\eqdef \isnew{\ownMl{\munit, \munit, \melt}} \\
	\ownPhys{\state} &\eqdef \ownM{\munit, \exinj(\state), \munit} \\
~\\
	\Sem{\vctx \proves \pvs[\mask_1][\mask_2] \prop : \Prop}_\gamma &\eqdef
	\textdom{pvs}^{\Sem{\vctx \proves \mask_2 : \textlog{InvMask}}_\gamma}_{\Sem{\vctx \proves \mask_1 : \textlog{InvMask}}_\gamma}(\Sem{\vctx \proves \prop : \Prop}_\gamma) \\
	\isnew{\Sem{\vctx \proves \psvs[\mask_1][\mask_2] \prop : \Prop}_\gamma} &\eqdef
	\isnew{\textdom{psvs}^{\Sem{\vctx \proves \mask_2 : \textlog{InvMask}}_\gamma}_{\Sem{\vctx \proves \mask_1 : \textlog{InvMask}}_\gamma}(\Sem{\vctx \proves \prop : \Prop}_\gamma)} \\
	\Sem{\vctx \proves \wpre{\expr}[\mask]{\Ret\var.\prop} : \Prop}_\gamma &\eqdef
	\textdom{wp}_{\Sem{\vctx \proves \mask : \textlog{InvMask}}_\gamma}(\Sem{\vctx \proves \expr : \textlog{Expr}}_\gamma, \Lam\val. \Sem{\vctx \proves \prop : \Prop}_{\gamma[\var\mapsto\val]})
\end{align*}

\paragraph{Remaining semantic domains, and interpretation of non-assertion terms.}

The remaining domains are interpreted as follows:
\[
\begin{array}[t]{@{}l@{\ }c@{\ }l@{}}
\Sem{\textlog{InvName}} &\eqdef& \Delta \mathbb{N}  \\
\Sem{\textlog{InvMask}} &\eqdef& \Delta \pset{\mathbb{N}} \\
\Sem{\textlog{M}} &\eqdef& F(\iProp)
\end{array}
\qquad\qquad
\begin{array}[t]{@{}l@{\ }c@{\ }l@{}}
\Sem{\textlog{Val}} &\eqdef& \Delta \textdom{Val} \\
\Sem{\textlog{Expr}} &\eqdef& \Delta \textdom{Expr} \\
\Sem{\textlog{State}} &\eqdef& \Delta \textdom{State} \\
\end{array}
\qquad\qquad
\begin{array}[t]{@{}l@{\ }c@{\ }l@{}}
\Sem{1} &\eqdef& \Delta \{ () \} \\
\Sem{\type \times \type'} &\eqdef& \Sem{\type} \times \Sem{\type} \\
\Sem{\type \to \type'} &\eqdef& \Sem{\type} \nfn \Sem{\type} \\
\end{array}
\]
For the remaining base types $\type$ defined by the signature $\Sig$, we pick an object $X_\type$ in $\COFEs$ and define
\[
\Sem{\type} \eqdef X_\type
\]
For each function symbol $\sigfn : \type_1, \dots, \type_n \to \type_{n+1} \in \SigFn$, we pick a function $\Sem{\sigfn} : \Sem{\type_1} \times \dots \times \Sem{\type_n} \nfn \Sem{\type_{n+1}}$.

\typedsection{Interpretation of non-propositional terms}{\Sem{\vctx \proves \term : \type} : \Sem{\vctx} \nfn \Sem{\type}}
\begin{align*}
	\Sem{\vctx \proves x : \type}_\gamma &\eqdef \gamma(x) \\
	\Sem{\vctx \proves \sigfn(\term_1, \dots, \term_n) : \type_{n+1}}_\gamma &\eqdef \Sem{\sigfn}(\Sem{\vctx \proves \term_1 : \type_1}_\gamma, \dots, \Sem{\vctx \proves \term_n : \type_n}_\gamma) \\
	\Sem{\vctx \proves \Lam \var:\type. \term : \type \to \type'}_\gamma &\eqdef
	\Lam \termB : \Sem{\type}. \Sem{\vctx, \var : \type \proves \term : \type}_{\gamma[\var \mapsto \termB]} \\
	\Sem{\vctx \proves \term(\termB) : \type'}_\gamma &\eqdef
	\Sem{\vctx \proves \term : \type \to \type'}_\gamma(\Sem{\vctx \proves \termB : \type}_\gamma) \\
	\Sem{\vctx \proves \MU \var:\type. \term : \type}_\gamma &\eqdef
	\mathit{fix}(\Lam \termB : \Sem{\type}. \Sem{\vctx, x : \type \proves \term : \type}_{\gamma[x \mapsto \termB]}) \\
  ~\\
	\Sem{\vctx \proves () : 1}_\gamma &\eqdef () \\
	\Sem{\vctx \proves (\term_1, \term_2) : \type_1 \times \type_2}_\gamma &\eqdef (\Sem{\vctx \proves \term_1 : \type_1}_\gamma, \Sem{\vctx \proves \term_2 : \type_2}_\gamma) \\
	\Sem{\vctx \proves \pi_i(\term) : \type_i}_\gamma &\eqdef \pi_i(\Sem{\vctx \proves \term : \type_1 \times \type_2}_\gamma) \\
  ~\\
	\Sem{\vctx \proves \munit : \textlog{M}}_\gamma &\eqdef \munit \\
	\Sem{\vctx \proves \mcore\melt : \textlog{M}}_\gamma &\eqdef \mcore{\Sem{\vctx \proves \melt : \textlog{M}}_\gamma} \\
	\Sem{\vctx \proves \melt \mtimes \meltB : \textlog{M}}_\gamma &\eqdef
	\Sem{\vctx \proves \melt : \textlog{M}}_\gamma \mtimes \Sem{\vctx \proves \meltB : \textlog{M}}_\gamma
\end{align*}

An environment $\vctx$ is interpreted as the set of
finite partial functions $\rho$, with $\dom(\rho) = \dom(\vctx)$ and
$\rho(x)\in\Sem{\vctx(x)}$.

\paragraph{Logical entailment.}
We can now define \emph{semantic} logical entailment.

\typedsection{Interpretation of entailment}{\Sem{\vctx \mid \pfctx \proves \prop} : \mProp}

\[
\Sem{\vctx \mid \pfctx \proves \prop} \eqdef
\begin{aligned}[t]
\MoveEqLeft
\forall n \in \mathbb{N}.\;
\forall \rs, \isnew{\rsB} \in \textdom{Res}.\; 
\forall \gamma \in \Sem{\vctx},\;
\\&
\bigl(\All \propB \in \pfctx. n \in \Sem{\vctx \proves \propB : \Prop}_\gamma(\rs, \isnew{\rsB})\bigr)
\Ra n \in \Sem{\vctx \proves \prop : \Prop}_\gamma(\rs, \isnew{\rsB})
\end{aligned}
\]

The soundness statement of the logic reads
\[ \vctx \mid \pfctx \proves \prop \Ra \Sem{\vctx \mid \pfctx \proves \prop} \]

 \endgroup\begingroup
\section{Derived proof rules and other constructions}

We will below abuse notation, using the \emph{term} meta-variables like $\val$ to range over (bound) \emph{variables} of the corresponding type.
We omit type annotations in binders and equality, when the type is clear from context.
We assume that the signature $\Sig$ embeds all the meta-level concepts we use, and their properties, into the logic.
(The Coq formalization is a \emph{shallow embedding} of the logic, so we have direct access to all meta-level notions within the logic anyways.)

\paragraph{Persistent/Relevant assertions.}
\begin{defn}
  An assertion $\prop$ is \emph{persistent} or \isnew{\emph{relevant}}
  if $\prop \proves \always\prop$.
\end{defn}

Of course, $\always\prop$ is persistent for any $\prop$.
Furthermore, by the proof rules given in \Sref{sec:proof-rules}, $t = t'$ as well as $\ownGGhost{\mcore\melt}$, $\mval(\melt)$ and $\knowInv\iname\prop$ are persistent.
Persistence is preserved by conjunction, disjunction, separating conjunction as well as universal and existential quantification.

In our proofs, we will implicitly add and remove $\always$ from persistent assertions as necessary.

\paragraph{\isnew{Affine assertions.}}
\begin{defn}
  An assertion $\prop$ is \emph{affine} if $\prop \proves \aff{\prop}$.
\end{defn}

In our proofs, we will implicitly add and remove $\aff{-}$ from persistent assertions as necessary.

\paragraph{Timeless assertions.}

We can show that the following additional closure properties hold for timeless assertions:

\begin{mathparpagebreakable}
  \infer
  {\vctx \proves \timeless{\prop} \and \vctx \proves \timeless{\propB}}
  {\vctx \proves \timeless{\prop \land \propB}}

  \infer
  {\vctx \proves \timeless{\prop} \and \vctx \proves \timeless{\propB}}
  {\vctx \proves \timeless{\prop \lor \propB}}

  \infer
  {\vctx \proves \timeless{\prop} \and \vctx \proves \timeless{\propB}}
  {\vctx \proves \timeless{\prop * \propB}}

  \infer
  {\vctx \proves \timeless{\prop}}
  {\vctx \proves \timeless{\always\prop}}
\end{mathparpagebreakable}

\isnew{Some similar rules apply for \atimeless{-}.}

\subsection{Program logic}

Hoare triples and view shifts are syntactic sugar for weakest (liberal) preconditions and primitive view shifts, respectively:
\[
\hoare{\prop}{\expr}{\Ret\val.\propB}[\mask] \eqdef \aff{\always{(\prop \wand \wpre{\expr}[\mask]{\lambda\Ret\val.\propB})}}
\qquad\qquad
\begin{aligned}
\prop \vs[\mask_1][\mask_2] \propB &\eqdef \always{(\prop \Ra \pvs[\mask_1][\mask_2] {\propB})} \\
\prop \svs[\mask_1][\mask_2] \propB &\eqdef \always{(\prop \Ra \psvs[\mask_1][\mask_2] {\propB})} \\
\prop \vsE[\mask_1][\mask_2] \propB &\eqdef \prop \vs[\mask_1][\mask_2] \propB \land \propB \vs[\mask2][\mask_1] \prop
\end{aligned}
\]
We write just one mask for a view shift when $\mask_1 = \mask_2$.
Clearly, all of these assertions are persistent.
The convention for omitted masks is similar to the base logic:
An omitted $\mask$ is $\top$ for Hoare triples and $\emptyset$ for view shifts.

\subsection{Derived Rules} \isnew{We omit many of the derived rules for Hoare triples, view shifts, and step shifts. The interested reader can consult the Coq mechanization.}

\begin{align*}
 (s, T) \stsstep (s', T') \eqdef{}& s \stsstep s' \land \STSL(s) \uplus T = \STSL(s') \uplus T' \\
 s \stsfstep{T} s' \eqdef{}& \Exists T_1, T_2. T_1 \disj \STSL(s) \cup T \land (s, T_1) \stsstep (s', T_2)
\end{align*}
  \begin{mathpar}
    \inferH{sts-alloc}{}
    {\aff{\pred(s)} \vs \Exists \iname, \gname.   \STSCtx^{\gname}(\STSS,\pred) * \STSSt^\gname{(s, \STST \setminus \STSL(s))}}

    \inferH{sts-st-split}{}
    {\STSSt^\gname(s, T_1 \uplus T_2) \Lra \STSSt^\gname(s, T_1) * \STSSt^\gname(s, T_2)}

    \inferH{sts-open}
    {\text{(additional side conditions omitted)} \and  \physatomic{\expr} \\ \All s. s_0 \stsftrans{T} s. \hoareHV[t]{\isnew{\aff{\pred(s)}} * P}{\expr}{\Ret \val. \Exists s', T'. (s, T) \ststrans (s', T') * \isnew{\aff{\pred(s')}} * Q}}
  {  \STSCtx^{\gname}(\STSS,\pred) \vdash \hoareHV[t]{\STSSt^\gname(s_0, T) * P}{\expr}{\Ret \val. \Exists s', T'. \STSSt^\gname(s', T') * Q}}
  \end{mathpar}

\subsection{Global functor, ghost ownership, and namespaces}

\isnew{For composability reasons, Iris makes it possible to combine a collection of CMRAs to get a larger ``global'' CMRA. This makes it possible to combine proofs that are done using a certain CMRA $\monoid$ with ones that are done doing another CMRA $\monoid'$. We do omit the descriptions of these mechanisms; the interested reader should consult the original Iris 2.0 documentation.}

 \endgroup\begingroup
\section{Refinement RA}

\isnew{We now briefly describe the RA used to model $\ownThreadNoArg$ assetions. This entire section is new.} 

Fix a source language $\Lang$. We say that a list of configurations,
$\cfglist$ is compatible with a list of thread indices, $\idxlist$,
written $\hcompat(\cfglist, \idxlist) L$, if:

\begin{mathpar}
  \infer{}{\hcompat([], [])}
\and
  \infer{}{\hcompat([\cfgvar], [])}
\and
  \infer{\hcompat(\cfglist \dplus [\cfgvar], \idxlist) \\
         \cfgvar \istep{i} \cfgvar'}
        {\hcompat(\cfglist \dplus [\cfgvar, \cfgvar'], \idxlist \dplus [i])}
\end{mathpar}

We define an RA $\refinem(\Lang)$:

\begin{align*}
  \refinem(\Lang) \eqdef{}& \textdom{View} \times \mathcal{P}^{\textrm{fin}}(\mathbb{N}) \times \textdom {List Config} \times \textdom{List Nat}  \\[-0.2em]
  \textnormal{where }& \textdom{View} \eqdef \set{\refmaster, \refsnapshot} \\
  \mval \eqdef{}& \setComp{(\view, \tidset, \cfglist, \idxlist) \in \refinem(\Lang)}{
    \begin{aligned}
      & (\cfglist = [] \land \idxlist = [] \land \tidset = \emptyset) \lor {} \\&
      (\exists \cfglist',\tpool,\state.\, 
      \cfglist = \cfglist',(\cfg{\tpool}{\state})
      \land{} \\ &\quad (\forall i\in \tidset, i < |\tpool|) \land {} \\&\quad
       \hcompat(\cfglist, \idxlist))
    \end{aligned}} \\
  \mcore{(\view, \tidset, \cfglist, \idxlist)} \eqdef{}& (\refsnapshot, \emptyset, \cfglist, \idxlist)
  \\
  \ (\view, \tidset, \cfglist, \idxlist) \mtimes (\view', \tidset', \cfglist', \idxlist')
  \eqdef{}& \left(\max(\view, \view'), \tidset \uplus \tidset', \max(\cfglist, \cfglist'), \max(\idxlist, \idxlist)' \right)
\end{align*}
where $\max(\view, \view')$ is $\refmaster$ if either $\view$ or $\view'$ is $\refmaster$, and the maximum of two lists is just the longer of the two. 
 We have an additional proviso stating that multiplication is only defined if all of the following hold:
 \begin{enumerate}
   \item Either $\view = \refsnapshot$ or $\view' = \refsnapshot$

   \item If $\view = \view' = \refsnapshot$ then $\exists \cfglist'', \idxlist''$
     such that either:
       \begin{enumerate}
         \item  $\cfglist = \cfglist' \dplus \cfglist''$, 
                $\idxlist = \idxlist' \dplus \idxlist''$
           and, $\forall i \in \tidset'$, $i \not\in \idxlist''$, or
         \item $\cfglist' = \cfglist \dplus \cfglist''$,
               $\idxlist' = \idxlist \dplus \idxlist''$
           and $\forall i \in \tidset$, $i \not\in \idxlist''$.
       \end{enumerate}

   \item If $\view = \refsnapshot$ and  $\view' = \refmaster$ then $\exists \cfglist'', \idxlist''$
     such that $\cfglist' = \cfglist \dplus \cfglist''$, $\idxlist' = \idxlist \dplus \idxlist''$
     and $\forall i \in \tidset'$, $i \not\in \idxlist''$.

   \item If $\view = \refmaster$ and  $\view' = \refsnapshot$ then $\exists \cfglist'', \idxlist''$
     such that $\cfglist = \cfglist' \dplus \cfglist''$, $\idxlist = \idxlist' \dplus \idxlist''$
     and, $\forall i \in \tidset$, $i \not\in \cfglist''$.
 \end{enumerate}

Intuitively, the second component of an element, $\tidset$ represents
a set of thread ID's ``owned'' by this element, and the $\cfglist$
and $\idxlist$ are some prefix of an execution of a source program.
Then, condition one for multiplication being defined says there can be at most one
master. Condition two says that, among two snapshots, one can be
longer than the other, but the longer one cannot contain any additional steps
by threads owned by the other. Condition three and four say that a snapshot must
be a prefix of master, subject to the constraint that the master cannot
contain any extra steps by threads owned by the snapshot.

The definition of $\mstep$ for this RA is somewhat complicated. Let us
motivate it in words -- a reader that wants details is advised to
consult the Coq formalization. Only snapshots may take
steps. Intuitively, the snapshot is obligated to step every thread it
controls (i.e. every index in $\tidset$) which can possibly take a
step. But, since the snapshot is only a partial prefix of the program
execution, other threads not controlled by this snapshot may have
taken steps. Thus, we first non-deterministically speculate some steps
performed by other threads and then perform all of the required steps
for the owned threads. Finally, since performing those steps may fork
off new threads, we add the thread ids of the new threads to
$\tidset$.

Of course, for extra flexibility, we are allowed to step each thread
more than once. We do not bake delay steps into this monoid. Rather,
if we want delay steps, we first transform $\Lang$ into a language
$\Lang'$ that has additional ``stutter'' steps for delay.

Finally, we can interpret $\ownThreadNoArg$ as:
\begin{align*}
    \ownThread{i}{\expr} \eqdef 
    \exists \tpool, \state, \cfglist, \idxlist.\,
    \ownGLGhost{(\refsnapshot, \set{i}, \cfglist \dplus [\cfg{\tpool}{\state}], \idxlist)}
                             \wedge (\tpool[i] = \expr)
\end{align*}

We also get an assertion $\ownSPhys{\state}$ for talking about the state of the assertion:
\begin{align*}
    \ownSPhys{\state} \eqdef 
    \exists \tpool, \cfglist, \idxlist.\,
    \ownGGhost{(\refmaster, \emptyset, \cfglist \dplus [\cfg{\tpool}{\state}], \idxlist)}
\end{align*}

We can use this ``large footprint'' assertion about source programs to derive smaller assertions, just as we do for the $\ownPhys{\state}$ assertion in Iris.

Finally, with some effort we can use this CMRA with the infinite adequacy theorem to get the refinement results stated in the body of the paper.
 \endgroup
}
\section{Case Studies}
\subsection{Session-Typed Language Translation}

In this appendix we develop the logical relation used for our compiler
correctness proof. The body of the paper has already explained the meaning of session types. But it
avoided the crucial issue of showing that this relation is
well-defined. What actually happens is we develop a state transition
system for sessions which is \emph{parameterized} by an interpretation
of types. As we'll see, we then take a fixed point which is defined using this
construction.

In our actual proof though, we do not work
with the Hoare triple versions of our rules. Instead, we work with a
more primitive form called weakest precondition, since it is much easier
to work with in a proof assistant. We first sketch the
connection between these and Hoare triples.

Then, we give this parameterized STS construction and state the
mechanized weakest precondition proof rules using this STS for the
message passing primitives.  This time around we will be more explicit
about the delay constants at first, but then show why it is OK to hide
them. Next, we describe some additional Iris features we'll need and
use them to define the logical relation (in particular, the $\later{}$
modality we mentioned without explanation in the main text).  We then
prove that our logical relation is sound (that is, it implies
refinement for closed terms).  Finally, we prove the fundamental
lemma, which shows that the logical relation holds between well-typed
expressions and their translation. This completes the proof.

\subsubsection{Weakest Precondition}

 In Iris, Hoare triples are not a primitive form.
 Instead, they are defined in terms of a \emph{weakest pre-condition} primitive as follows:
 \[ \hoare\prop\expr{\Ret\var.\propB} \eqdef \aff{\always(\prop \wand \wpre\expr{\Ret\var.\propB}} ) \]
 The assertion $\wpre\expr{\Ret\var.\propB}$ expresses ownership of resources that is strong enough to justify the safe execution of $\expr$, such that when $\expr$ terminates, $\propB$ holds.
 This is an \emph{ephemeral} assertion in the sense that, like \emph{e.g.} $l \mapsto v$, it can be used at most once and is the invalidated.

 The magic wand $\prop \wand \wpre\expr{\Ret\var.\propB}$ says that if we are given resources satisfying $\prop$, we have enough resources to satisfy the weakest pre-condition.
 The wand $\wand$ works like an implication, but is right adjoint to the separating conjunction $*$ instead of plain conjunction $\land$.

 Notice that in the example in \secref{sec:iris}, when we summarized
 the current state of the proof, we often said things like ``Our
 current resources are $\prop$ and we have to verify the following
 code $\expr$ (with post-condition $\propB$)''.  This exactly
 corresponds to a proof-state where our current logical context is
 $\prop$, and the goal is
 $\wpre\expr\propB$: \[ \prop \vdash \wpre\expr\propB \] 

It should not be surprising that in carrying out Iris proofs in Coq, we
 generally work with weakest pre-conditions.

 Next, we wrap the \emph{always} modality $\always$ around the wand.
 This is to enforce that proofs of Hoare triples
 be \emph{persistent}, \ie the modality makes sure that a Hoare
 triple, once established, will remain valid throughout the remaining
 verification.  By default, assertions in Iris are ephemeral and hence
 can be used only once.

Finally, there is an affine modality wrapping the whole thing -- there
 should be no other $\ownThreadNoArg$ hidden within such a Hoare
 triple. We only get to use ones in $\prop$.
 
\subsubsection{The Session STS}

The states, tokens, and transitions we gave in the main text
did not mention the logical relation, so there  is no concern
about circularity in their definition. The only problem was
the definition of the state interpretation, which was given as:

\begin{align*}
&\sessionInv{\styp}{\chanloc}{n_\scl, n_\scr, \heaploc_\scl, \heaploc_\scr} \eqdef 
 \exists L_{\scc}, L_{\sch}.\, 
   \nonumber \\
&\qquad \Big(\chanloc \smapsto (L_\scc, []) * \linklist{L_\sch}{\heaploc_\scl}{\heaploc_\scr} * {}
    \\ 
&\qquad ~~~(\interpListRel{L_\sch}{L_\scc}{}{S^{n_{\scl}}}) * n_{\scl} + \llength{L_\scc} = n_\scr \Big) \lor \dots
  \end{align*}

This implicitly relied on the definition of the logical relation via the lifting of the relation to lists of values:

\begin{mathpar}
\infer{}{\interpListRel{[]}{[]}{}{\styp}}

\infer{\later{(\interpValRel{\val}{\Val}{}{\typ})} *
       \interpListRel{L_\sch}{L_\scc}{}{\styp}}
      {\interpListRel{\val L_\sch }{\Val L_\scc}{}{\recvtyp{\typ}{\styp}}}
\end{mathpar}

What we now do is have the lifting of lists take a \emph{parameter}, $\Theta$, which
is a pre-existing interpretation of types (\ie a map from types to a relation between
values of the target and source):

\begin{mathpar}
\infer{}{\interpListRel{[]}{[]}{\typInterp}{\styp}}

\infer{\typInterp(\typ)(\val, \Val) *
       \interpListRel{L_\sch}{L_\scc}{\typInterp}{\styp}}
      {\interpListRel{\val L_\sch }{\Val L_\scc}{\typInterp}{\recvtyp{\typ}{\styp}}}
\end{mathpar}

Then, $\sessionInvNoArg$ will also take this parameter and pass it to the list relation:

\begin{align*}
&\sessionInvPre{\typInterp}{\styp}{\chanloc}{n_\scl, n_\scr, \heaploc_\scl, \heaploc_\scr} \eqdef 
 \exists L_{\scc}, L_{\sch}.\,  \\
&\quad \Big(\chanloc \smapsto (L_\scc, []) *
 * \linklist{L_\sch}{\heaploc_\scl}{\heaploc_\scr} * {} \\
& \quad \quad \interpListRel{L_\sch}{L_\scc}{\typInterp}{S^{n_{\scl}}}
 * n_{\scl} + \llength{L_\scc} = n_\scr\Big) \vee \dots
  \end{align*}

We write $\uparrow (n_\scl, -, l, -)$ for the set of all states that have first component $n_\scl$ and third component $\heaploc$,
and symmetrically for the right counts and heap pointer.

Finally, we can define an assertion $\SessionProt{\typInterp}{\heaploc}{\locside{\chanloc}{\side}}{\styp}$ which
asserts  (1) existence of an STS governing $\heaploc$ and $\chanloc$, (2) comes equpped with the tokens
needed for manipulating the end-point indicated by $\side$, and (3) ensures that the current type of the end-point for $\side$ has type $\styp$:
\begin{align*}
& \SessionProt{\typInterp}{\heaploc}{\locside{\chanloc}{\side}}{\styp} \eqdef \\
& \quad
  \left(\side = \lside \Ra
  \exists \styp_0, n_\scl, \gamma,  \STSCtx^\gname(\STSS, \sessionInvPre{\typInterp}{\styp}{\chanloc}{-})\right. \\ 
& \quad\quad * \left.\STSSt^\gname(\uparrow (n_\scl, -, l, -), \{\leftTok{n} \ | \ n > n_\scl) * \styp_0^n = \styp\right) \\ 
& \quad \vee
  \left(\side = \rside \Ra
  \exists \styp_0, n_\scr, \gamma,  \STSCtx^\gname(\STSS, \sessionInvPre{\typInterp}{\styp}{\chanloc}{-})\right. \\ 
& \quad\quad\quad * \left.\STSSt^\gname(\uparrow (-, n_\scr, -, l), \{\rightTok{n} \ | \ n > n_\scl) * \styp_0^n = \dual{\styp}\right) \\
\end{align*}

Of course, we've now only deferred the problem -- eventually we do
need to plug in the desired logical relation for $\typInterp$. For
now, the key is that we can still prove things about channels which
use this invariant; its just that the proof rules are also parameterized by $\typInterp$.

To show these rules, we need to be a bit more precise about delay
constants for a moment.  The triple about the receive primitive
mentioned in the body of the paper omitted delay constants, and in
general we have not been explicit about such constants in the main
text. Let us first be a bit more precise about this, so that we can
justify why it is safe to ignore them subsequently.

As we said when explaing the $\ownThreadNoArg$ assertion when we do a
proof we need to fix a number $\dmax$ that will be an upper bound
throughout for all delay constants.  In the Coq development,
we somewhat profligately proved things involving the above
STS assuming that this upper bound was at least $\geq 100$. So for
concreteness, let us just fix this $D$ now to be 100. Then we have
proved the following rules (written in the weakest precondition style
instead of Hoare triples):

\begin{mathpar}
\infer{1 < \delay \leq \dmax \\ 
        \delay' \leq \dmax}
{\ownThreadD{i}{\lctx[\newch]}{\delay} \proves
  \wpreL{\heapnewch}
       {
        \begin{aligned}[t]&
          (\heaploc, \heaploc). \exists \chanloc.\,
        \ownThreadD{i}{\lctx[(\locside{\chanloc}{\lside}, \locside{\chanloc}{\rside})]}{\delay'} 
        \\ & 
        * \SessionProt{\typInterp}{\heaploc}{\locside{\chanloc}{\lside}}{\styp}
        * \SessionProt{\typInterp}{\heaploc}{\locside{\chanloc}{\rside}}{\dual{\styp}}
        \end{aligned}
       }
}

\infer{1 < \delay \leq \dmax - 2 \\
          \delay' \leq \dmax - 2 \\
          \forall \val, \Val.\, \typInterp(\typ)(\val, \Val) \vdash \later{\prop(\val, \Val)}}
{
\begin{array}[t]{l}
\ownThreadD{i}{\lctx[\recv{\locside{\chanloc}{\side}}]}{\delay}
  * \SessionProt{\typInterp}{\heaploc}{\locside{\chanloc}{\side}}{\recvtyp{\tau}{\styp}} 
  \proves \arcr
\quad  \wpre{\heaprecv\ \heaploc}
       {\Ret(\heaploc', \val). \exists \Val.\,
        \ownThreadD{i}{\lctx[(\locside{\chanloc}{\side}, \Val)]}{\delay'}
         * \prop(\val, \Val)
         * \SessionProt{\typInterp}{\heaploc}{\locside{\chanloc}{\side}}{\styp}}
\end{array}
}

\infer{4 < \delay \leq \dmax \\
          d' \leq \dmax - 1}
{ 
\begin{array}[t]{l}
\ownThreadD{i}{\lctx[\send{\locside{\chanloc}{\side}}{\Val}]}{\delay}
         * \typInterp(\typ)(\val, \Val)
          * \SessionProt{\typInterp}{\heaploc}{\locside{\chanloc}{\side}}{\sendtyp{\tau}{\styp}}
  \proves \arcr
\quad
  \wpre{\heapsend\ \heaploc\ \val}
       {\Ret\heaploc'.
        \ownThreadD{i}{\lctx[\locside{\chanloc}{\side}]}{\delay'}
         * \SessionProt{\typInterp}{\heaploc'}{\locside{\chanloc}{\side}}{\styp}}
\end{array}
}

\end{mathpar}

In addition, to the explicit delays, we also have this business about
$\later{\prop}$ in the second rule, which we can ignore for now.  Note
that if $4 < k \leq \dmax - 2$, then $k$ satisfies all of the side
conditions placed on the delay constants that we start with in each
rule (\ie $\delay$). Moreover, such a $\delay$ satisfies all of the
constraints placed on the \emph{ending} delay constant (\ie
$\delay'$). That means if we start with such a $k$, at the end we can
continue with the same $k$. In particular, a choice of $k=50$ works.
It's also the case that $0$ satisfies the conditions on the $\delay'$
above; so whatever delay we start with, we can always end up with $0$,
if we wish.

Putting this together, we will define $\ownThread{i}{\Expr}$ (that is, \emph{without} the delay constant) as:

\begin{mathpar}
\ownThread{i}{\Expr} \eqdef \ownThreadD{i}{\Expr}{50}

\ownThread{i}{\Val} \eqdef \ownThreadD{i}{\Val}{0}
\end{mathpar}

That is, when we are working with an expression, we assume an implicit
delay constant of 50; when the source thread is a value, it is $0$'d
out.  Rewriting the rules above with this new form, we can ignore the delay constants:

\begin{mathpar}
\infer{}
{\ownThread{i}{\lctx[\newch]} \proves
  \wpreL{\heapnewch}
       {
        \begin{aligned}[t]&
          (\heaploc, \heaploc). \exists \chanloc.\,
        \ownThread{i}{\lctx[(\locside{\chanloc}{\lside}, \locside{\chanloc}{\rside})]} 
        \\ & 
        * \SessionProt{\typInterp}{\heaploc}{\locside{\chanloc}{\lside}}{\styp}
        * \SessionProt{\typInterp}{\heaploc}{\locside{\chanloc}{\rside}}{\dual{\styp}}
        \end{aligned}
       }
}

\infer{\forall \val, \Val.\, \typInterp(\typ)(\val, \Val) \vdash \later{\prop(\val, \Val)}}
{
\begin{array}[t]{l}
\ownThread{i}{\lctx[\recv{\locside{\chanloc}{\side}}]}
  * \SessionProt{\typInterp}{\heaploc}{\locside{\chanloc}{\side}}{\recvtyp{\tau}{\styp}} 
  \proves \arcr
\quad  \wpre{\heaprecv\ \heaploc}
       {\Ret(\heaploc', \val). \exists \Val.\,
        \ownThread{i}{\lctx[(\locside{\chanloc}{\side}, \Val)]}
         * \prop(\val, \Val)
         * \SessionProt{\typInterp}{\heaploc}{\locside{\chanloc}{\side}}{\styp}}
\end{array}
}

\infer{}
{
\begin{array}[t]{l}
\ownThread{i}{\lctx[\send{\locside{\chanloc}{\side}}{\Val}]}
         * \typInterp(\typ)(\val, \Val)
          * \SessionProt{\typInterp}{\heaploc}{\locside{\chanloc}{\side}}{\sendtyp{\tau}{\styp}}
  \proves \arcr
\quad
  \wpre{\heapsend\ \heaploc\ \val}
       {\Ret\heaploc'.
        \ownThread{i}{\lctx[\locside{\chanloc}{\side}]}
         * \SessionProt{\typInterp}{\heaploc'}{\locside{\chanloc}{\side}}{\styp}}
\end{array}
}
\end{mathpar}

To see that these implicit rules follow from the explicit delay constant form,
observe that if (1) $\lctx[\Expr]$ is a value, $\Expr$ must in fact be
a value, and (2) if $\lctx[\Val]$ is a value, then for all $\Val'$,
$\lctx[\Val']$ is a value. This means for each of the above rules, we
first determine whether the evaluation $\lctx$ will be a value after
we substitute the return values in the post-condition; if it is, we
apply the corresponding original rule taking $d' = 0$. If not, we take
$d' = 50$.

The reason we chose the ``implicit'' delay for values to be $0$ is so
that $\ownThread{i}{\Val} \vdash \stopped$, and so that our refinement
rule can also be written as:

\begin{mathpar}
  \infer
  {\hoare{\ownThread{i}{\Expr}}{\expr}{\Ret \var. \Exists\Val. \ownThread{i}{\Val} * \var \valobsrel \Val}[]}
  {\expr \refines \Expr}
\end{mathpar}

\subsubsection{Logical Relation}

Now that we have the STS defined in this parameterized way, we can
follow up by defining the logical relation in a non-circular way.  We
follow the standard set-up of defining an interpretation of types that
relates values, then lifting this to a relation on closed expressions,
and then using that to define a relation on open expressions. In
general, given a function $\typInterp$ which maps types to Iris
relations on values of the target and source, we write
$\interpValRel{\val}{\Val}{\typInterp}{\tau}$ to say that $\val$ and
$\Val$ are related at the interpretation of $\tau$ under $\typInterp$.

\paragraph{Lifting relations on values to expressions.}
The expression lifting operation is generic with respect to the final interpretation on values, 
so let us first define that. Given $\typInterp$, we define its lifting
between expressions of the target and source as:

\[ \interpExprRel{\expr}{\Expr}{\typInterp}{\tau} \eqdef
    \forall i, \lctx.\, \aff{\ownThread{i}{\lctx[\Expr]} \wand
          \wpre{\expr}{\val.\, \exists \Val.\, \aff{\interpValRel{\val}{\Val}{\typInterp}{\tau}} * \ownThread{i}{\lctx[\Val]}}}
\]

This is an assertion saying that for any choice of thread $i$ and
evaluation context $\lctx$, if we are given
$\ownThread{i}{\lctx[\Expr]}$ we can prove a weakest-precondition for
$\expr$ in which the executions end in related values according to
$\typInterp$ at the appropriate type. We wrap this wand in an affine
modality to ensure there are no implicit other threads owned. Also, in
the post-condition, the proof that the values are related must be
affine for a similar reason.

Let us note that the standard way such lifting relations are defined
(in for instance, the sequential case) is more or less to say that
$\expr$ and $\Expr$ are related if they evaluate to related
values. Indeed, that's essentially what we are saying here, but using
the weakest precondition and all of the other machinery we have built
up so that we're implicitly talking about \emph{concurrent} executions
and fairness.

Note that so long as $\interpValRel{\val}{\Val}{\typInterp}{\tau}$ is affine,
$\interpValRel{\val}{\Val}{\typInterp}{\tau} \proves
 \interpExprRel{\val}{\Val}{\typInterp}{\tau}$

\paragraph{Logical relation on values.}

Iris has two important features which we did not explain in the main
text.  First, there is the modality $\later \prop$ (called ``later''),
which we have asked the reader to ignore a few times.  This describes
resources which satisfy $\prop$ at \emph{one lower step-index}. So,
$\prop \proves \later \prop$, but not conversely. Second, we can
construct fixed-points of arbitrary recursive definition of a
predicate, so long as all recursive occurences are under the ``later''
modality.  Such recursive definitions are called ``guarded''.

We will define our interpretation of types by taking a fixed point of
a guarded recursive definition. First, we define the function $F$
which we are going to take the fixed-point of. Given an interpretation
of types, $\typInterp$ (that is, a map from types to  an Iris relation
on values), we define $F(\typInterp)$ as yet another interpretation of
types, defined by structural induction on types:

\begin{mathpar}
\infer{}{\interpValRel{n}{n}{F(\typInterp)}{\inttyp}}

\infer{}{\interpValRel{\vunit}{\vunit}{F(\typInterp)}{\unit}}

\infer
{(\interpValRel{\val_1}{\Val_1}{F(\typInterp)}{\tau_1}) *
(\interpValRel{\val_2}{\Val_2}{F(\typInterp)}{\tau_2})}
{\interpValRel{(\val_1, \val_2)}{(\Val_1, \Val_2)}{F(\typInterp)}{\tau_1 \otensor \tau_2}}

\infer
{\forall \val', \Val'.\, \aff{\interpValRel{\val'}{\Val'}{F(\typInterp)}{\tau'}
  \wand \interpExprRel{(\lambda x.\, \expr)\ \val'}{(\lambda x.\, \Expr) \Val'}{F(\typInterp)}{\tau}}}
{\interpValRel{\lambda x.\, \expr}{\lambda x.\, \Expr}{F(\typInterp)}{\tau' \lolli \tau}}

\infer
{\SessionProt{\later \typInterp}{\heaploc}{\locsidevar}{S}}
{\interpValRel{\heaploc}{\locsidevar}{F(\typInterp)}{S}}
\end{mathpar}

The first three rules are straight-forward. The third says that two
values are related at $\tau' \lolli \tau$, if, when we apply them to
values related at the intepretation of $\tau'$, the applications are
related at $\tau$. Since the applications are expressions and not
values, we use the lifting of the type interpretation.\footnote{A
  lifting of an interpretation of types is defined as the point-wise
  lifting at each type.} By $\later{\typInterp}$ we mean the type
interpretation which applies a later after applying $\typInterp$ to
its arguments.  Note that each occurence of $F(\typInterp)$ in a
premise is at a smaller type, so $F(\typInterp)$ is
well-defined. Second, the only occurence of $\typInterp$ is in the
rule for session types, where it occurs guarded by a later
modality. Thus, the fixed point of $F$ exists -- call it
$\thefp$. Then, our logical relation at values is this
$\interpValRel{-}{-}{\thefp}{\typ}$ (which in the text we wrote
without $\thefp$ annotation)

The fact that $\thefp$ is a fixed point of $F$ means that

\[ \interpValRel{\val}{\Val}{\thefp}{\tau} \provesIff 
   \interpValRel{\val}{\Val}{F(\thefp)}{\tau} \]
   
Note that $\interpValRel{\val}{\Val}{\thefp}{\typ}$ is affine
for all $\val$, $\Val$, and $\typ$, and so is the lifting of $\thefp$ to expressions.

Now we can explain the additional premise in the rule for receive. By
default when we have
$\interpValRel{\heaploc}{\locsidevar}{\thefp}{\recvtyp{\typ}{\styp}}$,
this would unfold to
$\SessionProt{\later{\thefp}}{\heaploc}{\locsidevar}{\recv{\typ}{\styp}}$
-- looking at the definition of state interpretation, we might be
afraid that the returned message values $\val$ and $\Val$, might
merely be related at $\later$ of $\thefp(\tau)$ in the postcondition, not at $\thefp(\tau)$. The key is that the premise above the line lets
us ``strip'' off such a later: take $\prop$ in that premise to be
$\interpValRel{-}{-}{\thefp}{\tau}$. Then, the premise of the rule
holds when $\typInterp$ is $\later \thefp$, because $\later{(\interpValRel{\val}{\Val}{\thefp}{\tau})} \vdash \later{(\interpValRel{\val}{\Val}{\thefp}{\tau})}$, and so in the post
condition we get out $\interpValRel{\val}{\Val}{\thefp}{\tau}$ without
the later.  This justifies the rule we presented in the paper (and
which we use below in the proof of the fundamental lemma).

We now lift this relation to one on open expressions in a given
context. Given a map $\substh$ from a finite domain of variables to
closed expressions in the target language, we write $[\substh]\expr$
for the result of simultaneously substituting each variable in the
domain of $\substh$ for its image under $\substh$. $[\substc]\Expr$ is
the analogous thing for source expressions.

   We define the relation $\interpOpenRel{\Gamma}{\expr}{\Expr}{\typ}$,
   where $\Gamma$ is a typing context, by:
\begin{align*}
&\interpOpenRel{\Gamma}{\expr}{\Expr}{\typ} \eqdef \forall \substh, \substc.\,
(\dom(\substh) = \dom(\substc) = \Gamma \wedge 
\fv{\Expr} = \fv{\expr} \subseteq \dom(\Gamma)) \rightarrow \\
& \quad
\left(\underset{x \in \Gamma}{\bigast} \interpExprRel{\substh(x)}{\substc(x)}{\thefp}{\Gamma(x)}\right)
  \vdash \interpExprRel{[\substh]\expr}{[\substc]\Expr}{\thefp}{\typ}
\end{align*}

where we take $\bigast$ over an empty set to be $\aff{\TRUE}$. We now show that the logical relation is \emph{sound}:

\begin{lem}\label{lem:logrel-soundness}
If $\interpOpenRel{\emptyset}{\expr}{\Expr}{\typ}$, then
$\expr \refines \Expr$.
\end{lem}
\begin{proof}
By induction on $\typ$:
\begin{itemize}
\item $\tau = \inttyp$: By assumption, we have that $\emp \vdash \interpExprRel{\expr}{\Expr}{\thefp}{\inttyp}$. 

By the refinement rule, it suffices to prove:

\[ \emp \vdash {\hoare{\ownThread{i}{\Expr}}{\expr}{\Ret \var. \Exists\Val. \ownThread{i}{\Val} * \var \valobsrel \Val}[]} \]

Since $\emp$ is affine and persistent, we can clear the implicit modalities in the Hoare triple, and move the pre-condition to the context by using the wand intro rule, so that we need to show:

\[ \ownThread{i}{\Expr} \vdash \wpre{\expr}{\Ret \var. \Exists\Val. \ownThread{i}{\Val} * \var \valobsrel \Val} \]

By our assumption, we can rewrite this as:

\[ \ownThread{i}{\Expr} * \interpExprRel{\expr}{\Expr}{\thefp}{\inttyp} \vdash \wpre{\expr}{\Ret \var. \Exists\Val. \ownThread{i}{\Val} * \var \valobsrel \Val} \]

Unfolding the definition of the lifting of $\thefp$, instantiating it with $i$ and $[]$, clearing the afine modality, and then eliminating the resulting wand with $\ownThread{i}{\expr}$, we need

\[ \wpre{\expr}{\Ret \var. \Exists\Val. \ownThread{i}{\Val} * \Exists n.\, \var = \Val = n} \vdash \wpre{\expr}{\Ret \var. \Exists\Val. \ownThread{i}{\Val} * \var \valobsrel \Val} \]

By the rule of consequence of weakest precondition (\ie, weakest precondition is covariant in the post condition), it suffices to show for aribtrary $\var$,

\[
\Exists\Val.\, \ownThread{i}{\Val} * \exists n.\, \var = \Val = n \vdash
\Exists\Val.\, \ownThread{i}{\Val} * \var \valobsrel \Val
\]

Which follows from the definition of $\valobsrel$.

\end{itemize}

The other cases are the same: in each case, we take the assumption, apply the refinement proof rule; then the ownership we get in the pre-condition of the refinement triple is the antecedent of the wand in the definition of the lifting of the type relation to expressions; eliminating that gives a weakest precondition, where the post condition implies that values are logically related; our logical relation implies $\valobsrel$, so we're done.

\end{proof}

We now prove what is called the fundamental theorem:
\begin{lem}\label{lem:fundamental} If $\Gamma \vdash \Expr : \typ$ then $\interpOpenRel{\Gamma}{\compile{\Expr}}{\Expr}{\typ}$
\end{lem}

\begin{proof}
The proof is by induction on the derivation of $\Gamma \vdash \Expr : \typ$.

\begin{itemize}

\item Case Var: $\compile{x} = x$, so given closing subsitutions $\substh$ and $\substc$, we need to show

\[
\left(\underset{x \in \Gamma}{\bigast} \interpExprRel{\substh(x)}{\substc(x)}{\thefp}{\Gamma(x)}\right)
  \vdash \interpExprRel{[\substh]x}{[\substc]x}{\thefp}{\typ}
\]
  
Then $[\substh]x = \substh(x)$ and similarly for $\substc$, so this becomes

\[
\left(\underset{x \in \Gamma}{\bigast} \interpExprRel{\substh(x)}{\substc(x)}{\thefp}{\Gamma(x)}\right)
  \vdash \interpExprRel{\substh(x)}{\substc(x)}{\thefp}{\typ}
\]

Note that $\Gamma(x) = \typ$, so, manipulating $\bigast$, this just becomes

\[
\interpExprRel{\substh(x)}{\substc(x)}{\thefp}{\typ} * \left(\underset{x \in \Gamma \setminus \{x\}}{\bigast} \interpExprRel{\substh(x)}{\substc(x)}{\thefp}{\Gamma(x)}\right)
  \vdash \interpExprRel{\substh(x)}{\substc(x)}{\thefp}{\typ}
\]

Since the interpretation is affine, we can throw away 
$\left(\underset{x \in \Gamma \setminus \{x\}}{\bigast} \interpExprRel{\substh(x)}{\substc(x)}{\thefp}{\Gamma(x)}\right)$, and we are done.

\item Case Int: Again, $\compile{n} = n$, and these are closed, so we just have to show, given substitutions:

\[
\left(\underset{x \in \Gamma}{\bigast} \interpExprRel{\substh(x)}{\substc(x)}{\thefp}{\Gamma(x)}\right)
  \vdash \interpExprRel{n}{n}{\thefp}{\typ}
\]

We can throw away the context, since it is affine and not needed. Then, our goal follows
immediately from the fact that $\interpValRel{n}{n}{\thefp}{\inttyp}$, and we have already said that related values are related under
the lifting, since $\thefp$ is affine.

\item Case {Fun-Intro}: Once more, $\compile{(\lambda x.\, \Expr)} = \lambda x.\, \compile{\Expr}$. 
 By our induction hypothesis, we have that $\interpOpenRel{\Gamma, x: \typ_1}{\compile{\Expr}}{\Expr}{\typ_2}$. Given $\substh$ and
$\substc$, we know $[\substc]\lambda x .\, \Expr = \lambda x.\,
[\substc]\Expr$ and $[\substh]\lambda x .\, \compile{\Expr} = \lambda
x.\, [\substh]\compile{\Expr}$. We can push the substitutions under the binders because
$x$ is explicitly not in $\Gamma$ hence not in the domain of the substitutions, and capture is not possible because
they are substituting closed expressions. Then it suffices to show

\[
\left(\underset{x \in \Gamma}{\bigast} \interpExprRel{\substh(x)}{\substc(x)}{\thefp}{\Gamma(x)}\right)
  \vdash \interpValRel{\lambda x .\, [\substh]\compile{\Expr}}{\lambda x .\, [\substc]\Expr}{\thefp}{\typ_1 \lolli \typ_2}
\]

Unfolding the definition on the right, we are given arbitrary $\val'$, $\Val'$ which are related at $\typ_1$,
which we put in the context using the rule for for all introduction, affine introduction, and wand introduction; we then need to show:

\begin{align*}
& \interpValRel{\val'}{\Val'}{\thefp}{\typ_1} * \left(\underset{x \in \Gamma}{\bigast} \interpExprRel{\substh(x)}{\substc(x)}{\thefp}{\Gamma(x)}\right) \\
& \quad \vdash
\interpExprRel{(\lambda x.\, [\substh]\compile{\Expr})\ \val'}{(\lambda x.\, [\substc]\Expr)\ \Val'}{\thefp}{\typ_1 \lolli \typ_2}
\end{align*}

After yet more unfolding, we need to show

\begin{align*}
& \ownThread{i}{\lctx[(\lambda x.\, [\substc]\Expr)\ \Val']}
   * \interpValRel{\val'}{\Val'}{\thefp}{\typ_1} 
   * \left(\underset{x \in \Gamma}{\bigast} \interpExprRel{\substh(x)}{\substc(x)}{\thefp}{\Gamma(x)}\right) \\
&\qquad \vdash
\wpre{(\lambda x.\, [\substh]\compile{\Expr})\ \val'}
     {\val''.\, \exists \Val''.\, \ownThread{i}{\lctx[\Val'']} * \interpValRel{\val''}{\Val''}{\thefp}{\typ_2}}
\end{align*}

We step the source and target to do the beta reductions:

\begin{align*}
& \ownThread{i}{\lctx[[\Val'/x][\substc]{\Expr}]]}
   * \interpValRel{\val'}{\Val'}{\thefp}{\typ_1} 
   * \left(\underset{x \in \Gamma}{\bigast} \interpExprRel{\substh(x)}{\substc(x)}{\thefp}{\Gamma(x)}\right) \\
&\qquad \vdash
\wpre{[\Val'/x][\substh]\compile{\Expr}}
     {\val''.\, \exists \Val''.\, \ownThread{i}{\lctx[\Val'']} * \interpValRel{\val''}{\Val''}{\thefp}{\typ_2}}
\end{align*}

We can then extend the substititons $\substh$ and $\substc$ to also map $x$ to $\val'$ and $\Val'$ respectively; call these $\substh'$ and $\substc'$. Then, rewriting the above and regrouping
our assertions to the left of the turnstile, we get:

\begin{align*}
& \ownThread{i}{\lctx[[\substc']{\Expr}]}
   * \left(\underset{y \in \Gamma, x:\typ_1}{\bigast} \interpExprRel{\substh'(y)}{\substc'(y)}{\thefp}{\Gamma(y)}\right) \\
&\qquad \vdash
\wpre{[\substh']\compile{\Expr}}
     {\val''.\, \exists \Val''.\, \ownThread{i}{\lctx[\Val'']} * \interpValRel{\val''}{\Val''}{\thefp}{\typ_2}}
\end{align*}

But now we can apply our induction hypothesis (suitably unrolling and instantiating everything).

\item Case \ruleref{Fun-Elim}: Our induction hypothesis says:
 $\interpOpenRel{\Gamma}{\compile{\Expr}}{\Expr}{\typ_1 \lolli \typ_2}$ and
 $\interpOpenRel{\Gamma}{\compile{\Expr'}}{\Expr'}{\typ_1}$ and
 we need to show, given closing substitutions $\substh$ and $\substc$:

\[
\left(\underset{x \in \Gamma \uplus \Gamma'}{\bigast} \interpExprRel{\substh(x)}{\substc(x)}{\thefp}{(\Gamma \uplus \Gamma')(x)}\right)
  \vdash \interpExprRel{[\substh]\compile{\Expr}\, [\substh]\compile{\Expr'}}
                       {[\substc]\Expr\, [\substc]\Expr'}{\thefp}{\typ_2}
\]

But we note now that $\Expr$ and $\Expr'$ (and their translations) must only have free variables within $\Gamma$ and $\Gamma'$ respectively.
Let $\substhB$ and $\substhB'$ be the restrictions of $\substh$ to $\Gamma$ and $\Gamma'$ respectively;
similarly for $\substcB$ and $\substcB'$. Then $[\substc]\Expr = [\substcB]\Expr$ and $[\substh]\compile{\Expr} = [\substcB]\compile{\Expr}$ and
similarly for the primed versions. Then substituting, and splitting our $\bigast$, the above becomes:

\begin{align*}
& \left(\underset{x \in \Gamma}{\bigast} \interpExprRel{\substhB(x)}{\substcB(x)}{\thefp}{\Gamma(x)}\right) *
\left(\underset{x \in \Gamma'}{\bigast} \interpExprRel{\substhB'(x)}{\substcB'(x)}{\thefp}{\Gamma'(x)}\right) \\
& \quad \vdash \interpExprRel{[\substhB]\compile{\Expr}\, [\substhB]\compile{\Expr'}}{{ [\substcB]\Expr\, [\substcB]\Expr'} }{\thefp}{\typ_2}
\end{align*}

Unfolding this becomes:
\begin{align*}
& \ownThread{i}{\lctx[[\substcB]\Expr\, [\substcB']\Expr']} *
\left(\underset{x \in \Gamma}{\bigast} \interpExprRel{\substhB(x)}{\substcB(x)}{\thefp}{\Gamma(x)}\right) *
\left(\underset{x \in \Gamma'}{\bigast} \interpExprRel{\substhB'(x)}{\substcB'(x)}{\thefp}{\Gamma'(x)}\right)
\\
&\quad \vdash \wpre{[\substhB]\compile{\Expr}\, [\substhB']\compile{\Expr'}}
     {\val''.\, \exists \Val''.\, \ownThread{i}{\lctx[\Val'']} * \interpValRel{\val''}{\Val''}{\thefp}{\typ_2}}
\end{align*}

Using the bind rule to focus on the left expression of each computation, we use the induction hypothesis for $\Expr$ to get
that there exists $\val$, $\Val$, for which it suffices to show that:
\begin{align*}
& \ownThread{i}{\lctx[\Val\, [\substcB']\Expr']} *
\interpValRel{\val}{\Val}{\thefp}{\typ_1 \lolli \typ_2} *
\left(\underset{x \in \Gamma'}{\bigast} \interpExprRel{\substhB'(x)}{\substcB'(x)}{\thefp}{\Gamma'(x)}\right)
\\
& \quad \vdash \wpre{\val\, [\substhB']\compile{\Expr'}}
     {\val''.\, \exists \Val''.\, \ownThread{i}{\lctx[\Val'']} * \interpValRel{\val''}{\Val''}{\thefp}{\typ_2}}
\end{align*}

Doing the same thing for the other induction hypothesis we get that there exists $\val'$, $\Val'$ for which we need to show that:
\begin{align*}
& \ownThread{i}{\lctx[\Val\, \Val'} *
\interpValRel{\val}{\Val}{\thefp}{\typ_1 \lolli \typ_2} *
\interpValRel{\val'}{\Val'}{\thefp}{\typ_1}
\\
& \quad \vdash \wpre{\val\, \val'}
     {\val''.\, \exists \Val''.\, \ownThread{i}{\lctx[\Val'']} * \interpValRel{\val''}{\Val''}{\thefp}{\typ_2}}
\end{align*}

But this follows from the definition of 
$\interpValRel{\val}{\Val}{\thefp}{\typ_1 \lolli \typ_2}$.

\item Case \ruleref{Pair-Intro}: We have $\interpOpenRel{\Gamma_1}{\compile{\Expr_1}}{\Expr_1}{\typ_1}$
and $\interpOpenRel{\Gamma_2}{\compile{\Expr_2}}{\Expr_2}{\typ_2}$, and must show, given $\substh$ and 
$\substc$ that:

\[
\left(\underset{x \in \Gamma_1 \uplus \Gamma_2}{\bigast} \interpExprRel{\substh(x)}{\substc(x)}{\thefp}{(\Gamma_1 \uplus \Gamma_2)(x)}\right)
  \vdash \interpExprRel{([\substh]\compile{\Expr_1}\, [\substh]\compile{\Expr_2})}
                       {([\substc]\Expr_1, [\substc]\Expr_2)}
                       {\thefp}{\typ_1 \otensor \typ_2}
\]

As in the previous case, we can restrict the substitions to $\Gamma_1$
and $\Gamma_2$; for $\substh$, call those $\substhB$ and $\substhB'$
respectively and similarly for $\substc$, and since $E_i$ has free variables only appearing in $\Gamma_i$, we can rewrite the above to be:

\begin{align*}
& \left(\underset{x \in \Gamma_1}{\bigast} 
   \interpExprRel{\substhB(x)}{\substcB(x)}{\thefp}{\Gamma_1(x)}\right) *
\left(\underset{x \in \Gamma_2}{\bigast} 
   \interpExprRel{\substhB'(x)}{\substcB'(x)}{\thefp}{\Gamma_2(x)}\right)
\\ & \quad \vdash 
\interpExprRel{([\substhB]\compile{\Expr_1}, [\substhB']\compile{\Expr_2})}{([\substcB]\Expr_1, [\substcB']\Expr_2)}{\thefp}{\typ_1 \otensor \typ_2}
\end{align*}

Now, we unfold the lifting on the right and introduce the corresponding source ownership assertion:

\begin{align*}
&\ownThread{i}{\lctx[([\substcB]\Expr_1, [\substcB']\Expr_2)]} 
* \left(\underset{x \in \Gamma_1}{\bigast} 
   \interpExprRel{\substhB(x)}{\substcB(x)}{\thefp}{\Gamma_1(x)}\right) *
\left(\underset{x \in \Gamma_2}{\bigast} 
   \interpExprRel{\substhB'(x)}{\substcB'(x)}{\thefp}{\Gamma_2(x)}\right)
\\
&\quad  \vdash 
\wpre{([\substhB]\compile{\Expr_1}, [\substhB']\compile{\Expr_2})}
     {(\val_1, \val_2).\, \exists \Val_1, \Val_2.\, 
     \ownThread{i}{\lctx[(\Val_1, \Val_2)]}
     * \interpValRel{(\val_1, \val_2)}{(\Val_1, \Val_2)}{\thefp}{\typ_1 \otensor \typ_2}}
\end{align*}

As in the previous case, we use the bind rule to focus on the left and then right subexpressions; applying our induction hypothesis about $\Expr_1$ and $\Expr_2$ then lets us trade our context assumptions to end up with some $\val_1, \Val_1, \val_2, \Val_2$ for which it suffices to show:

\begin{align*}
&\ownThread{i}{\lctx[(\Val_1, \Val_2)]} 
*  \interpValRel{\val_1}{\Val_1}{\thefp}{\typ_1} 
*  \interpValRel{\val_2}{\Val_2}{\thefp}{\typ_2} 
\\
&\quad  \vdash 
     \ownThread{i}{\lctx[(\Val_1, \Val_2)]}
     * \interpValRel{(\val_1, \val_2)}{(\Val_1, \Val_2)}{\thefp}{\typ_1 \otensor \typ_2}
\end{align*}

But this follows immediately from the unfolding of $\thefp$ at pair type.

\item Case \ruleref{Pair-Elim}: The argument is similar to {Fun-Intro}, in the way that \ruleref{Pair-Intro} corresponds to \ruleref{Fun-Elim}, so we omit it.
  
\item Case \ruleref{Fork}: Surprisingly, this too is similar to the
  cases for pair and function!  The reader might have expected that
  this would in fact be a hard case; but it is not \emph{precisely}
  because our lifting from value relation to expression relation is
  usiing the weakest precondition, so we'll be able to use the fork
  rule. To wit, we have by induction hypothesis that
  $\interpOpenRel{\Gamma_1}{\compile{\Expr_\f}}{\Expr_\f}{\typ'}$ and
  $\interpOpenRel{\Gamma_2}{\compile{\Expr}}{\Expr}{\typ}$. We must show:

\[
\left(\underset{x \in \Gamma_1 \uplus \Gamma_2}{\bigast} \interpExprRel{\substh(x)}{\substc(x)}{\thefp}{(\Gamma_1 \uplus \Gamma_2)(x)}\right)
  \vdash \interpExprRel{\fork{[\substh]\compile{\Expr_\f}}; \compile{[\substh]\Expr}}
                       {\fork{[\substc]\Expr_\f}; {[\substc]\Expr}}
                       {\thefp}{\typ}
\]

Doing the same routine of restricting the substitutions, splitting the context assumptions, unfolding the lifting relation, etc. we get that we need to show

\begin{align*}
&\ownThread{i}{\lctx[\fork{[\substcB]\Expr_\f}; [\substcB']\Expr]} 
* \left(\underset{x \in \Gamma_1}{\bigast} 
   \interpExprRel{\substhB(x)}{\substcB(x)}{\thefp}{\Gamma_1(x)}\right) *
\left(\underset{x \in \Gamma_2}{\bigast} 
   \interpExprRel{\substhB'(x)}{\substcB'(x)}{\thefp}{\Gamma_2(x)}\right)
\\
&\quad  \vdash 
\wpre{(\fork{[\substhB]\compile{\Expr_{f}}}; [\substhB']\compile{\Expr})}
     {\val.\, \exists \Val.\, 
     \ownThread{i}{\lctx[\Val]}
     * \interpValRel{\val}{\Val}{\thefp}{\typ}}
\end{align*}

We start by applying the fork rule; we then do a step shift in the
source to get $\ownThread{j}{[\substcB]\Expr_f}$ for some new $j$, and
the parent thread becomes $\ownThread{i}{\Expr}$. We now split the
context, passing the assumptions about $\Gamma_1$ the $j$ source
thread to a proof of weakest precondition for the child target thread;
the rest is used in the target parent thread's proof. So we need to prove the 
following two entailments:

\begin{align*}
&\ownThread{j}{[\substcB]\Expr_\f}
* \left(\underset{x \in \Gamma_1}{\bigast} 
   \interpExprRel{\substhB(x)}{\substcB(x)}{\thefp}{\Gamma_1(x)}\right)
\\
&\quad  \vdash 
\wpre{[\substhB]\compile{\Expr_\f}}
     {\val.\, \stopped}
\end{align*}

\begin{align*}
&\ownThread{i}{\lctx[[\substcB']\Expr]} 
\left(\underset{x \in \Gamma_2}{\bigast} 
   \interpExprRel{\substhB'(x)}{\substcB'(x)}{\thefp}{\Gamma_2(x)}\right)
\\
&\quad  \vdash 
\wpre{[\substhB']\compile{\Expr}}
     {\val.\, \exists \Val.\, 
     \ownThread{i}{\lctx[\Val]}
     * \interpValRel{\val}{\Val}{\thefp}{\typ}}
\end{align*}

The latter follows from the induction hypothesis. For the former, using the induction hypothesis and the rule of consequence for weakest precondition means we just need to show, for all $\val$, $\Val$:

\[
\ownThread{j}{\Val} * \interpValRel{\val}{\Val}{\thefp}{\typ'} \vdash
\stopped
\]

But the right side of the conjunction is affine, hence can be thrown away, and we have mentioned above that a \ownThreadNoArg\ of a value entails \stopped.

\item Case \ruleref{NewChTyp}: At last we come to a case involving a channel. This expression is closed, it suffices to show:
\begin{align*}
&\ownThread{i}{\lctx[\newch]} 
* \left(\underset{x \in \Gamma}{\bigast} 
   \interpExprRel{\substhB(x)}{\substcB(x)}{\thefp}{\Gamma_1(x)}\right) \\
&\quad \vdash 
\wpre{\heapnewch}
     {\val.\, \exists \Val.\, 
     \ownThread{i}{\lctx[\Val]}
     * \interpValRel{\val}{\Val}{\thefp}{\styp \otensor \dual{\styp}}}
\end{align*}
Throwing away the $\Gamma$ piece of the context and rewriting our proof rule for $\newch$ on the remaining $\ownThread{i}{\lctx[\newch]}$ yields:

\begin{align*}
& \wpre{\heapnewch}
       {(\heaploc, \heaploc). \exists \chanloc.\,
        \ownThread{i}{\lctx[(\locside{\chanloc}{\lside}, \locside{\chanloc}{\rside})]}
        * \SessionProt{\later \thefp}{\heaploc}{\locside{\chanloc}{\lside}}{\styp}
        * \SessionProt{\later \thefp}{\heaploc}{\locside{\chanloc}{\rside}}{\dual{\styp}}}
\\ & \quad
\vdash 
\wpre{\heapnewch}
     {\val.\, \exists \Val.\, 
     \ownThread{i}{\lctx[\Val]}
     * \interpValRel{\val}{\Val}{\thefp}{\styp \otensor \dual{\styp}}}
\end{align*}

Applying the rule of consequence, we merely need to show, for all $l$ and $c$:
\begin{align*}
& \ownThread{i}{\lctx[(\locside{\chanloc}{\lside},
                       \locside{\chanloc}{\lside})]}
        * \SessionProt{\later \thefp}{\heaploc}{\locside{\chanloc}{\lside}}{\styp}
        * \SessionProt{\later \thefp}{\heaploc}{\locside{\chanloc}{\rside}}{\dual{\styp}}
\\ & \quad \vdash 
 \ownThread{i}{\lctx[(\locside{\chanloc}{\lside},
                       \locside{\chanloc}{\rside})]}
     * \interpValRel{(l, l)}{(\locside{\chanloc}{\lside}, 
                              \locside{\chanloc}{\rside})}
                            {\thefp}{\styp \otensor \dual{\styp}}
\end{align*}
Unfolding the definition of value relation at pair type and session type, the left
side matches the right side so we are done.

\item Case \ruleref{Send}: Our induction hypothesis gives us that 
 $\interpOpenRel{\Gamma}{\compile{\Expr_1}}{\Expr_1}{\sendtyp{\typ}\styp}$ and
 $\interpOpenRel{\Gamma}{\compile{\Expr_2}}{\Expr_2}{\typ}$ and
 we need to show, given closing substitutions $\substh$ and $\substc$:

\[
\left(\underset{x \in \Gamma_1 \uplus \Gamma_2}{\bigast} \interpExprRel{\substh(x)}{\substc(x)}{\thefp}{(\Gamma_1 \uplus \Gamma_2)(x)}\right)
  \vdash \interpExprRel{\heapsend\ {[\substh]\compile{\Expr_1}}\ {[\substh]\compile{\Expr_2}}}
                       {\send{[\substc]\Expr_1}{[\substc]\Expr_2}}{\thefp}{\styp}
\]

Doing the usual routine of noting that the two subexpressions only have free variables in their respective contexts, taking restrictions, and unfolding everything:

\begin{align*}
&\ownThread{i}{\lctx[\send{[\substcB]\Expr_1}{[\substcB']\Expr_2}]}
* \left(\underset{x \in \Gamma_1}{\bigast} 
   \interpExprRel{\substhB(x)}{\substcB(x)}{\thefp}{\Gamma_1(x)}\right)
* \left(\underset{x \in \Gamma_2}{\bigast} 
   \interpExprRel{\substhB'(x)}{\substcB'(x)}{\thefp}{\Gamma_2(x)}\right)
\\ & \quad
\vdash 
\wpre{\heapsend\ {[\substhB]\compile{\Expr_1}}\ {[\substh]\compile{\Expr_2}}}
     {\val.\, \exists \Val.\, 
     \ownThread{i}{\lctx[\Val]}
     * \interpValRel{\val}{\Val}{\thefp}{\styp}}
\end{align*}

Notice that because the implementation of $\heapsend$ let-binds a pair containing the two ``arguments'', they are both evaluated first before substitution into the ``body'' of the send primitive. Hence, we first evaluate the left and then right subexpressions using our induction hypothesis, and then just have to show for all values $\val_1$, $\val_2$, $\Val_1$, $\Val_2$:

\begin{align*}
&\ownThread{i}{\lctx[\send{\Val_1}{\Val_2}]}
  * \interpValRel{\val_1}{\Val_1}{\thefp}{\sendtyp{\typ}{\styp}}
  * \interpExprRel{\val_2}{\Val_2}{\thefp}{\typ}
\\ & \quad
\vdash 
\wpre{\heapsend\ \val_1\ \val_2}
     {\val.\, \exists \Val.\, 
     \ownThread{i}{\lctx[\Val]}
     * \interpValRel{\val}{\Val}{\thefp}{\styp}}
\end{align*}

But then, the relation of $\val_1$ and $\Val_1$ unfolds to give us:
\begin{align*}
&\ownThread{i}{\lctx[\send{\Val_1}{\Val_2}]}
  * \SessionProt{\later \thefp}{\val_1}{\Val_1}{\sendtyp{\typ}{\styp}}
  * \interpExprRel{\val_2}{\Val_2}{\thefp}{\typ}
\\ & \quad
\vdash 
\wpre{\heapsend\ \val_1\ \val_2}
     {\val.\, \exists \Val.\, 
     \ownThread{i}{\lctx[\Val]}
     * \interpValRel{\val}{\Val}{\thefp}{\styp}}
\end{align*}

So we can use the proof rule for send primitive on the left to get

\begin{align*}
&  \wpre{\heapsend\ \val_1 \val_2}
       {\Ret\heaploc'. \ownThread{i}{\lctx[\locside{\chanloc}{\side}]}
         * \SessionProt{\later \thefp}{\val_1}{\Val_1}{\styp}}
\\ & \quad
\vdash 
\wpre{\heapsend\ \val_1\ \val_2}
     {\val.\, \exists \Val.\, 
     \ownThread{i}{\lctx[\Val]}
     * \interpValRel{\val}{\Val}{\thefp}{\styp}}
\end{align*}

Applying the rule of consequence, and unfolding the interpretation at session type,
we are done.

\item Case \ruleref{Recv}: 
 Our induction hypothesis gives us that 
 $\interpOpenRel{\Gamma}{\compile{\Expr}}{\Expr}{\recvtyp{\typ}\styp}$ and
 we need to show, given closing substitutions $\substh$ and $\substc$:

\[
\left(\underset{x \in \Gamma}{\bigast} \interpExprRel{\substh(x)}{\substc(x)}{\thefp}{(\Gamma)(x)}\right)
  \vdash \interpExprRel{\heaprecv\ {[\substh]\compile{\Expr}}}
                       {\recv{[\substc]\Expr}}{\thefp}{\typ \otensor \styp}
\]

Unfolding we have:
\begin{align*}
&\ownThread{i}{\lctx[\recv{[\substc]\Expr}]}
* \left(\underset{x \in \Gamma}{\bigast} 
   \interpExprRel{\substh(x)}{\substc(x)}{\thefp}{\Gamma(x)}\right)
\\ & \quad
\vdash 
\wpre{\heaprecv\ {[\substh]\compile{\Expr}}}
     {\val.\, \exists \Val.\, 
     \ownThread{i}{\lctx[\Val]}
     * \interpValRel{\val}{\Val}{\thefp}{\typ \otensor \styp}}
\end{align*}

We first evaluate the subexpression using our induction hypothesis, and then just to show for all values $\val$, $\val$:

\begin{align*}
&\ownThread{i}{\lctx[\recv{\Val}]}
  * \interpValRel{\val}{\Val}{\thefp}{\recvtyp{\typ}{\styp}}
\\ & \quad
\vdash 
\wpre{\heaprecv\ \val}
     {\val'.\, \exists \Val'.\, 
     \ownThread{i}{\lctx[\Val]}
     * \interpValRel{\val}{\Val}{\thefp}{\typ \otensor \styp}}
\end{align*}

Again, the relation of $\val$ and $\Val$ unfolds to give us:
\begin{align*}
&\ownThread{i}{\lctx[\recv{\Val}]}
  * \SessionProt{\later \thefp}{\val}{\Val}{\recvtyp{\typ}{\styp}}
\\ & \quad
\vdash 
\wpre{\heaprecv\ \val}
     {\val'.\, \exists \Val'.\, 
     \ownThread{i}{\lctx[\Val]}
     * \interpValRel{\val}{\Val}{\thefp}{\typ \otensor \styp}}
\end{align*}

So we can use the proof rule for receive primitive on the left to get
\begin{align*}
& \wpre{\heaprecv\ \val}
       {\Ret(\heaploc', \val'). \exists \Val'.\,
        \ownThread{i}{\lctx[(\locside{\chanloc}{\side}, \Val')]}
         * (\interpValRel{\val'}{\Val'}{\typInterp}{\tau})
         * \SessionProt{\later \thefp}{\heaploc'}{\Val'}  \styp}
\\ & \quad
\wpre{\heaprecv\ \val}
     {\val'.\, \exists \Val'.\, 
     \ownThread{i}{\lctx[\Val']}
     * \interpValRel{\val'}{\Val'}{\thefp}{\typ \otensor \styp}}
\end{align*}

Applying the rule of consequence, and unfolding the interpretation at tensor and session type,
we are done.

\end{itemize}
\end{proof}

This completes the proof of the fundamental lemma. Then, 
\thmref{thm:compiler} follows by combining the fundamental lemma and the
soundness of the logical relation (\lemref{lem:logrel-soundness}):
every well typed term is logically related to its translation, and the
left side of a logically related pair refines the right side.

\subsection{Craig-Landin-Hagersten Lock}

\newcommand\lockcode{\begin{figure}
\[
\begin{array}[t]{l>{\quad\quad}l}
\begin{array}{ll}
&\clhnew \eqdef \\
&\quad \lambda \_.\, \bind{d}{\alloc{\FALSE}}{} \\
&\quad\quad\quad (\alloc{d}, \alloc{d})
\end{array}
&
\begin{array}{ll}
&\ticketnew \eqdef \\
&\quad \lambda \_.\, (\alloc{0}, \alloc{0}) \\
\\
\end{array}
\\ \\ 
\begin{array}{ll}
&\clhwait \eqdef \recN{loop}\ {me}\ prev\ lk \\
&\quad \bind{w}{\load{prev}}{} \\
&\quad \ifmatch{w}\\
&\quad\quad loop\  me\  prev\  lk \\
&\quad \textlog{else}\\
&\quad\quad \store{(\fst lk)}{me}
\end{array}
&
\begin{array}{ll}
&\ticketwait \eqdef \recN{loop}\ {x}\ lk \\
&\quad \bind{o}{\load{(\fst{prev})}}{} \\
&\quad \ifmatch{x = o}\\
&\quad\quad () \\
&\quad \textlog{else}\\
&\quad\quad loop\ x\ lk
\end{array}
\\ \\ 
\begin{array}{ll}
&\clhacq \eqdef \lambda lk. \\
&\quad \bind{me}{\alloc{\TRUE}}{} \\
&\quad \bind{prev}{\swap{\snd{lk}}{me}}{} \\
&\quad \clhwait\ me\ prev\ lk
\end{array}
&
\begin{array}{ll}
&\ticketacq \eqdef \lambda lk. \\
&\quad \bind{n}{\fai{\snd{lk}}}{} \\
&\quad \ticketwait\ n\ lk \\
\\
\end{array}
\\ \\ 
\begin{array}{ll}
&\clhrel \eqdef \lambda lk. \\
&\quad\store{\load{(\fst lk)}}{\FALSE} 
\end{array}
&
\begin{array}{ll}
&\ticketrel \eqdef \lambda lk. \\
&\quad\store{(\fst lk)}{\load{(\fst lk)} + 1}
\end{array}
\end{array}
\]
\caption{Code for CLH and ticket locks.}
\label{fig:clhticket}
\end{figure}
}

\newcommand\typetransfig{
\begin{figure}
\begin{mathpar}
\inferN{Var}{\Gamma(\var) = \typ}{\Gamma \vdash \var \typtrans \var : \typ}

\inferN{Int}{}{\Gamma \vdash n \typtrans n : \inttyp}

\inferN{Unit}{}{\Gamma \vdash \vunit \typtrans \vunit : \unit}

\inferN{Bool}{}{\Gamma \vdash b \typtrans b : \booltyp}

\inferN{Pair-Intro}{\Gamma \vdash \expr_1 \typtrans \expr_1' : \typ_1 \\
                    \Gamma \vdash \expr_2 \typtrans \expr_2' : \typ_2}
       {\Gamma \vdash (\expr_1, \expr_2) \typtrans (\expr_1', \expr_2') : \typ_1 \times \typ_2}

\inferN{Pair-Fst}{\Gamma \vdash \expr \typtrans \expr' : \typ_1 \times \typ_2}
       {\Gamma \vdash \fst\ \expr \typtrans \fst\ \expr': \typ_1}

\inferN{Pair-Snd}{\Gamma \vdash \expr \typtrans \expr' : \typ_1 \times \typ_2}
       {\Gamma \vdash \snd\ \expr \typtrans \snd\ \expr': \typ_2}

\inferN{Sum-Left}{\Gamma \vdash \expr \typtrans \expr' : \typ_1}
       {\Gamma \vdash \inl\ \expr \typtrans \inl\ \expr': \typ_1 + \typ_2}

\inferN{Sum-Right}{\Gamma \vdash \expr \typtrans \expr' : \typ_2}
       {\Gamma \vdash \inr\ \expr \typtrans \inr\ \expr': \typ_1 + \typ_2}

\inferN{Sum-Elim}{\Gamma \vdash \expr \typtrans \expr' : \typ_1 + \typ_2 \\
       \Gamma \vdash \expr_1 \typtrans \expr_1' : \typ_1 \rightarrow \typ \\
       \Gamma \vdash \expr_2 \typtrans \expr_2' : \typ_2 \rightarrow \typ}
       {\Gamma \vdash \caseelim{\expr}{\expr_1}{\expr_2}
         \typtrans \caseelim{\expr'}{\expr_1'}{\expr_2'} : \typ}

\inferN{Seq}{\Gamma \vdash \expr_1 \typtrans \expr_1' : \typ_1 \\
             \Gamma \vdash \expr_2 \typtrans \expr_2' : \typ_2}
       {\Gamma \vdash \expr_1; \expr_2 \typtrans \expr_1';\expr_2' : \typ_2}
       
\inferN{Fun-Intro}
       {\Gamma, f : \typ_1 \rightarrow \typ_2, x : \typ_1 \vdash \expr \typtrans \expr' : \typ_2 }
       {\Gamma \vdash \rec{f}{x} \expr \typtrans \rec{f}{x}.\, \expr' : \typ_1 \rightarrow \typ_2}

\inferN{Fun-Elim}{\Gamma \vdash \expr_1 \typtrans \expr_1' : \typ_1 \rightarrow \typ_2 \\
       \Gamma \vdash \expr_2 \typtrans \expr_2' : \typ_1 }
      {\Gamma \vdash \expr_1\, \expr_2 \typtrans \expr_1'\, \expr_2' : \typ_2}

\inferN{Fork}{\Gamma \vdash \expr_\f \typtrans {\expr_\f}': \typ}
      {\Gamma \vdash \fork{\expr_\f} \typtrans \fork{{\expr_\f}'}  : \unit}
      
\inferN{Load}{\Gamma \vdash \expr \typtrans \expr' : \reftyp{\typ}}
      {\Gamma \vdash \load{\expr} \typtrans \load{\expr'}  : \typ}

\inferN{Store}{\Gamma \vdash \expr \typtrans \expr' : \reftyp{\typ} \\
               \Gamma \vdash \expr_s \typtrans \expr_s' : \typ}
      {\Gamma \vdash \store{\expr}{\expr_s} \typtrans \store{\expr'}{\expr_s'}  : \unit}
      
\inferN{Send}{\Gamma_1 \vdash \Expr_1 : \sendtyp{\typ}{\styp} \\
       \Gamma_2 \vdash \Expr_2 : \typ
      }
      {\Gamma_1 \uplus \Gamma_2 \vdash \send{\Expr_1}{\Expr_2}  : \styp}

\inferH{Lock-Intro}{}
      {\Gamma \vdash \ticketnew \typtrans \clhnew : \unit \rightarrow \locktyp}

\inferH{Lock-Elim}
       {\Gamma \vdash \expr_l \typtrans \expr_l' : \locktyp \\
        \Gamma \vdash \expr \typtrans \expr' : \unit \rightarrow \typ}
      {\Gamma \vdash \ticketsync\ \expr_l\ \expr \typtrans \clhsync\ \expr_l'\ \expr' : \typ}
      
\ticketsync\ \expr_l\ \expr \eqdef \lambda lk, f.\ \ticketacq\ lk; \bind{z}{\expr\,\vunit}
                  {(\ticketrel\ lk ; z)} \\
\clhsync\ \expr_l'\ \expr' \eqdef \lambda lk, f.\ \clhacq\ lk; \bind{z}{\expr'\,\vunit}
                  {(\clhrel\ lk ; z)}
\end{mathpar}
\caption{Type directed translation between programs using ticket lock and CLH locks.}
\label{fig:typetrans}
\end{figure}
}

\newcommand{\clhsync}{\textlog{CLHsync}}
\newcommand{\clhnew}{\textlog{CLHnew}}
\newcommand{\clhacq}{\textlog{CLHacq}}
\newcommand{\clhrel}{\textlog{CLHrel}}
\newcommand{\clhwait}{\textlog{CLHwait}}
\newcommand{\ticketsync}{\textlog{ticketsync}}
\newcommand{\ticketnew}{\textlog{ticketnew}}
\newcommand{\ticketacq}{\textlog{ticketacq}}
\newcommand{\ticketrel}{\textlog{ticketrel}}
\newcommand{\ticketwait}{\textlog{ticketwait}}

\lockcode

In this part, we describe our second case study, which shows that the
Craig-Landin-Hagersten queue lock\citep{Craig93, MagnussonLH94}
refines a fair ticket lock~\citep{MCS91}. Like the first case study, the results
mentioned here have all been mechanized in Coq. Our goal here is simply to give
a feel for what the result is about; we do not describe the actual proofs or
the state transition systems we use. The interested reader should consult the Coq source.

\subsubsection{Lock Implementations} The code for the two types of locks appears in
\figref{fig:clhticket}. The $\fai{-}$ operation is a fetch-and-increment:
it takes a location as a parameter and atomically increments that
location and returns the previous value. Meanwhile, $\swap{-}{-}$ takes
a location and a value and atomically loads the previous value of the location and
stores the second parameter as the new value.

Each ticket lock consists of two references storing integers: the
``owner'' counter and the ``next'' counter, called $o$ and $n$ in the
code, so $\ticketnew$ just allocates these two references. To acquire
the lock, a thread atomically fetches and increments the next counter
to get a number, which we call a ``ticket'' ($\ticketacq$). It then spins on the
owner counter waiting until the owner counter matches the thread's
ticket number ($\ticketwait$). Once it matches, the thread enters the critical
section. To release the lock, the thread increments the owner counter ($\ticketrel$).

The lock is fair in the sense that once a thread calls $\ticketacq$
and completes its fetch-and-increment operation to get its ticket, it
is guaranteed to enter the critical section \emph{before} other
threads that subsequently call $\ticketacq$. Thus, if every thread
which acquires the lock eventually releases it, every thread that
tries to acquire the lock will eventually enter the critical section.

However, one drawback to this design is that every thread
waiting to acquire the lock spins on the same owner counter
location. Depending on the machine, this can cause poor performance
for memory related reasons~\citep{MCS91}.  In contrast, in the CLH
lock, every waiting thread spins on a \emph{different} memory
location.  The rather extensive benchmarking done by \citet{DavidGT13}
suggests that in some settings the CLH lock performs better than the
ticket lock.

Conceptually, we can think of
the CLH lock as consisting of a queue of threads, in the order that
they tried to acquire the lock. The lock has two fields: a pointer to
the head of the queue, and a pointer to the tail of the queue. Each
node in the queue consists of a single boolean value. While this value
is true, the node that owns that node is either (1) still waiting to
enter the critical section, or (2) is in the critical section, but has
not yet released the lock. When creating a new CLH lock, we create a
dummy node in which this field is set to false, and set the head and
tail to point to it ($\clhnew$).  To acquire the lock, a thread
allocates a new node and inserts itself at the back of the list by
atomically swapping the location of the new node with the current
tail ($\clhacq$) -- the value returned by this atomic swap is the node of its
predecessor in the list. The thread then spins on this node's boolean value until
it is false -- this indicates that the predecessor has released the
lock, signalling that the thread is now at the head of the queue and
may enter the critical section ($\clhwait$). Before entering the critical section,
the thread therefore updates the head pointer to point to its node. To
release the lock, the thread looks at the node pointed to by the head
field, and sets that node to false ($\clhrel$).

Observe that, as claimed, each waiting thread spins on a different location: the
node of its predecessor in the list. 

There is one important difference between our implementation and
typical versions. Usually, there is no head pointer, and the API is designed so that the acquire
method take two parameters: (1) the lock, and (2) the new node to be
added to the list, while the release method takes only the releaser's
node. This means conventionally it is not a drop-in replacement for the ticket-lock,
because in the ticket lock, both acquire/release take the lock
structure as their only parameter. That said, the head pointer 
may have other uses and is included in some implementations~\citep{Lea05}.

\subsubsection{Refinement Specification}

Note that the CLH lock is fair in the same sense as the ticket lock:
once a thread completes its atomic swap to insert itself into the
queue, it must enter the critical section before threads that later
call $\clhacq$. This is what makes it possible to establish a fair
termination preserving refinement.

However, let us note that if we take a closed program $e$ and replace
all instances of the ticket lock primitives with CLH primitives to get
a program $e^\ast$, it is \emph{not} necessarily the case that $e^\ast$ refines
$e$. First, one can violate the ``abstraction'' of the lock datatype, like so:
\[
\begin{array}{ll}
&\bind{l}{(\alloc{0},\alloc{0})}{} \\
&\quad \ticketacq\ l
\end{array}
\]
This program does not trigger a fault, because $l$ happens to represent a valid ticket lock, but if we replace $\ticketacq$ with $\clhacq$ it gets stuck. Of course, we could use a type system in which the lock primitives manipulate values of some abstract type ``lock'', and this would rule out such examples. Then we might hope to prove that the refinement holds when the original program is well-typed.

However, this is not true. Consider the following example:
\[
\begin{array}{ll}
&\bind{l}{\ticketnew\, ()}{} \\
&\quad \ticketrel\ l \\
&\quad \ticketacq\ l
\end{array}
\]
This program diverges, because the first call to $\ticketrel$ will
increment the owner counter to $1$, but during the call to
$\ticketacq$, the ticket the thread gets will be number $0$, so it
will loop forever waiting for the owner counter to be $0$. In
contrast, if we use $\clhnew$, $\clhrel$, and $\clhacq$ in the above
program, the result terminates! The call to $\clhrel$ will set the
dummy node's boolean to $\FALSE$, but it was already $\FALSE$, so this
does nothing. Thus, the thread will see the dummy node is $\FALSE$
during the call to $\clhacq$ and will not spin. Thus, the CLH version is
not a termination-preserving refinement of the source.

Of course, this example uses the lock ``incorrectly'' because it releases
the lock before acquiring it. What's important is that as long as the lock
is used ``correctly'' the refinement should hold. This is captured by
the following set of rules that we prove about the CLH primitives:

\begin{mathpar}
\infer{5 < \delay \leq \dmax \\ 
        \delay' \leq \dmax}
{
\begin{array}[t]{l}
\ownThreadD{i}{\lctx[\ticketnew\,()]}{\delay} * \aff{R} \proves
\arcr 
 \quad \wpre{\clhnew\,()}
       {
          lk. \exists lks, \gamma.\,
        \islock{\gamma}{lks}{lk}{R} 
        * \ownThreadD{i}{\lctx[lks]}{\delay'} 
       }
\end{array}
}

\infer{5 < \delay \leq \dmax \\ 
        \delay' \leq \dmax - 2}
{
\begin{array}[t]{l}
\ownThreadD{i}{\lctx[\ticketacq\ lks]}{\delay} * \islock{\gamma}{lks}{lk}{R} \proves 
\arcr 
\quad  \wpre{\clhacq\,()}
       {
          v. v = () * \locked{\gamma}{lks}{lk} * \ownThreadD{i}{\lctx[()]}{\delay'}
          * \aff{R}
       }
\end{array}
}

\infer{5 < \delay \leq \dmax \\ 
        \delay' \leq \dmax - 2}
{
\begin{array}[t]{l}
\ownThreadD{i}{\lctx[\ticketrel\ lks]}{\delay} * \islock{\gamma}{lks}{lk}{R}
 * \locked{\gamma}{lks}{lk} * \aff{R}
 \proves
\arcr 
  \quad \wpre{\clhrel\,lk}
       {
          v. v = () * \ownThreadD{i}{\lctx[()]}{\delay'}
       }
\end{array}
}
\end{mathpar}
where $\islock{-}{-}{-}{-}$ and $\locked{-}{-}{-}$ are certain
predicates defined in the logic. If one ignores all the parts of the
above specification that have to do with our extensions (delay
constants, refinement resources, and the affine modality), this looks
like a standard specification for a lock primitive in a higher-order
concurrent separation logic (indeed, this kind of specification is proved
about the ticket lock in the original Iris repository). Such a
specification says that for an arbitrary assertion $R$, if one
initially owns $R$, it can be given up to create a lock protecting $R$
-- in exchange, one gets an assertion $\islock{-}{-}{-}{-}$ which is
duplicable. Then, a thread can use this assertion to call acquire, and
afterward it back $R$, and an assertion $\locked{-}{-}{-}$ which
signifies that it is the owner of the lock. Both $R$ and
$\locked{-}{-}{-}$ are needed to then call release. Thus this typical
formulation of a lock specification rules out the kind of
counter-examples we illustrated above. Re-reading the specification
with the parts relevant to our extension, we see that each shows a
refinement from the corresponding ticket lock primitive.

Of course, one can use these rules, along with the rest of the logic,
to prove that for a given program, if one replaces the ticket
primitives with the CLH versions, the result is a refinement. If the
reader is familiar with the ``usual'' specification of lock primitives
in higher-order CSL, hopefully the above description is compelling
evidence that this refinement specification is ``good enough''. In the
next section, we describe a particular use of these proof rules to 
show that for a large class of well-typed programs, the desired refinement holds.

\subsubsection{Type-directed translation} We now describe a simple type-directed
translation from programs using the ticket primitives to ones using
the CLH primitives.  We then construct a logical relation which uses
the weakest precondition specification described above to show that
the translated programs refine their sources.

The translation is shown in \figref{fig:typetrans}. We write $\Gamma
\vdash \expr \typtrans \expr' : \typ$ to indicate that in context
$\Gamma$, the expression $\expr$ translates to $\expr'$ at type
$\typ$. The type system used in the translation is a standard simple
type system for MiniML with references, extended with an abstract type
$\locktyp$ for locks. In contrast to the session-typed language, the
type system here is fully structral, not affine. The primitive
$\ticketnew$ translates to $\clhnew$ and returns a value of type
$\locktyp$ (\ruleref{Lock-Intro}). This lock can be used with the
command $\ticketsync$, which takes a ticket lock and a function of
type $\unit \rightarrow \typ$ and synchronizes execution of the
function with the lock (\ruleref{Lock-Elim}). This translates to
$\clhsync$, which does the same thing but using the CLH lock
implementation.

Because these ``synchronize'' commands properly acquire and then
release the lock, a the type system rules out the kinds of bad
programs we mentioned above. Of course, this is a very restricted way
of using the lock. Our point is just to show that the weakest
precondition specification supports this kind of use; it in fact is
more flexible, but a type system that demonstrated this would be more
complex (at which point, one might want to just appeal to the program
logic directly to prove the refinement).

\typetransfig

As usual, to state our refinement result, we need to give a notion of observational equivalence on values. Here, we restrict our attention to booleans and say that given two booleans $b$ and $b'$, $b \valobsrel b'$ if and only if $b = b'$. With this notion of value equivalence as our foundation for refinement, we have mechanized the following result:

\begin{thm}\label{thm:locks}
If $\cdot \vdash \expr \typtrans \expr' : \booltyp$ then $\expr' \refines \expr$.
\end{thm}

To prove this theorem, we once again construct a logical
relation.\footnote{The logical relation in our Coq formalization is
  actually formulated generically with respect to a pair of lock
  implementations, assumed to satisfy weakest precondition
  specifications like the ones given above for the ticket and CLH
  lock. The mechanization then instantiates this generic formulation
  with the particular results about the ticket and CLH locks.}  This
is very much like the relation given for the session typed language,
with two notable exceptions (1) the type system is not affine, so all
the interpretations of the types in the relation are both
\emph{affine} and \emph{relevant} propositions (meaning they can be
duplicated and thrown away, freely), (2) we need to model reference and lock types.
For references and locks we have:

\begin{mathpar}
\infer{\islock{\gamma}{lk}{lk'}{\TRUE}}{\interpValRel{lk}{lk'}{}{\locktyp}}

\infer{\knowInv{}{\exists v, v'.\, \heaploc \mapsto v * \heaploc' \smapsto v' * \interpValRel{v}{v'}{}{\typ}}}
         {\interpValRel{\heaploc}{\heaploc'}{}{\reftyp{\typ}}}
\end{mathpar}
where the boxed assertion in the premise of the second rule is an Iris \emph{invariant}. These were not described in the main text, so the reader unfamiliar with the original work on Iris can think of them as a single-state STS\footnote{In fact, the STS construct is \emph{derived} from invariants.}. The first rule means that lock types are just interpreted as the $\islock{-}{-}{-}{-}$ assertion we saw before. The fourth parameter is $\TRUE$ because for purposes of the type translation, we don't care what resource the programmer is intending to protect with the lock. The second says that two heap locations are related at reference type $\reftyp{\tau}$ if there is an invariant enforcing that they must always point to values that are related at the interpretation of $\tau$.

Using these definitions, we can prove a fundamental lemma showing that if $\cdot \proves \expr \typtrans \expr' : \typ$, then $\interpExprRel{\expr}{\expr'}{}{\typ}$ holds. Then, we show a soundness lemma that states that if $\interpExprRel{\expr}{\expr'}{}{\booltyp}$, then $\expr \refines \expr'$. By composing these two results, we obtain \thmref{thm:locks}.

 \endgroup}
\fi

\end{document}